\documentclass[pra,twocolumn,a4paper,showpacs,aps,10pt,nofootinbib,allowtoday,accepted=2019-09-01]{quantumarticle}
\pdfoutput=1
\usepackage{graphicx}
\usepackage{amsmath}
\usepackage{amsthm}
\usepackage{amssymb}
\usepackage{color}
\usepackage{comment}
\usepackage{dutchcal}
\usepackage{mathrsfs}

\newtheorem*{theorem*}{Theorem}
\newtheorem*{definition*}{Definition}
\newtheorem*{lemma*}{Lemma}
\newtheorem*{corollary*}{Corollary}
\newcommand{\Tr}{\text{Tr}}

\newtheorem{theorem}{Theorem}
\newtheorem{definition}{Definition}

\newtheorem{corollary}{Corollary}

\makeatletter
\newcommand\footnoteref[1]{\protected@xdef\@thefnmark{\ref{#1}}\@footnotemark}
\makeatother

\begin{document}

\title{Beyond the  Cabello-Severini-Winter framework: Making sense of contextuality without sharpness of measurements}

\author{Ravi Kunjwal}
\email{rkunjwal@perimeterinstitute.ca}
\affiliation{Perimeter Institute for Theoretical Physics,\\ 31 Caroline Street North, Waterloo, Ontario, Canada, N2L 2Y5.}
\orcid{0000-0002-3978-5971}

\date{\today}

\begin{abstract}
 We develop a hypergraph-theoretic framework for Spekkens contextuality applied to Kochen-Specker (KS) type scenarios that goes beyond the Cabello-Severini-Winter (CSW) framework. To do this, we add new hypergraph-theoretic ingredients to the CSW framework. We then obtain noise-robust noncontextuality inequalities in this generalized framework by applying the assumption of (Spekkens) noncontextuality to both preparations and measurements. The resulting framework goes beyond the CSW framework in both senses, conceptual and technical. On the conceptual level: 1) as in any treatment based on the generalized notion of noncontextuality \`a la Spekkens, we relax the assumption of outcome determinism inherent to the Kochen-Specker theorem but retain measurement noncontextuality, besides introducing preparation noncontextuality,  2) we do not require the {\em exclusivity principle} -- that pairwise exclusive measurement events must all be mutually exclusive -- as a fundamental constraint on measurement events of interest in an experimental test of contextuality, given that this property is not true of general quantum measurements, and 3) as a result, we do not need to presume that measurement events of interest are ``sharp" (for any definition of sharpness), where this notion of sharpness is meant to imply the exclusivity principle. On the technical level, we go beyond the CSW framework in the following senses: 1) we introduce a source events hypergraph -- besides the measurement events hypergraph usually considered -- and define a new operational quantity ${\rm Corr}$ that appears in our inequalities, 2) we define a new hypergraph invariant -- the {\em weighted max-predictability} -- that is necessary for our analysis and appears in our inequalities, and 3) our noise-robust noncontextuality inequalities quantify tradeoff relations between three operational quantities -- ${\rm Corr}$, $R$, and $p_0$ -- only one of which (namely, $R$) corresponds to the Bell-Kochen-Specker functionals appearing in the CSW framework; when ${\rm Corr}=1$, the inequalities formally reduce to CSW type bounds on $R$. Along the way, we also consider in detail the scope of our framework vis-\`a-vis the CSW framework, particularly the role of Specker's principle in the CSW framework, i.e., what the principle means for an operational theory satisfying it and why we don't impose it in our framework.

\end{abstract}

\pacs{03.65.Ta, 03.65.Ud}

\maketitle
\tableofcontents

\section{Introduction}
To say that quantum theory is counterintuitive, or that it requires a revision of our classical intuitions, requires us to be mathematically precise in our definition of these classical intuitions. Once we have a precise formulation of such {\em classicality}, we can begin to investigate those features of quantum theory that power its {\em nonclassicality}, i.e., its departure from our classical intuitions, and thus prove {\em theorems} about such nonclassicality. To the extent that a physical theory is provisional, likely to be replaced by a better theory in the future, it also makes sense to articulate such notions of classicality in as operational a manner as possible. By `operational', we refer to a formulation of the theory that takes the operations -- preparations, measurements, transformations -- that can be carried out in an experiment as primitives and which specifies the manner in which these operations combine to produce the data in the experiment. Such an operational formulation often suggests generalizations of the theory that can then be used to better understand its axiomatics \cite{hardy5axioms, masanesmueller, CDP}. At the same time, an operational formulation also lets us articulate our notions of nonclassicality in a manner that is experimentally testable and thus allows us to leverage this nonclassicality in applications of the theory. Indeed, a key area of research in quantum foundations and quantum information is the development of methods to assess nonclassicality in an experiment under minimal assumptions on the operational theory describing it. The paradigmatic example of this is the case of Bell's theorem and Bell experiments \cite{Bell64, Bell66, Bellbook, CHSH,  bellnonlocality, Belltest1, Belltest2, Belltest3}, where any operational theory that is non-signalling between the different spacelike separated wings of the experiment is allowed. The notion of classicality at play in Bell's theorem is the assumption of {\em local causality}: any non-signalling theory that violates the assumption of local causality is said to exhibit nonclassicality by the lights of Bell's theorem.

More recently, much work \cite{KunjSpek, exptlpaper, KrishnaSpekkensWolfe, SchmidSpek, KunjSpek17, algorithmic} has been devoted to obtaining constraints on operational statistics that follow from a generalized notion of  noncontextuality proposed by Spekkens \cite{Spe05}. This notion of classicality \cite{Spe05} has its roots in the Kochen-Specker (KS) theorem \cite{KochenSpecker}, a no-go theorem that rules out the possibility that a deterministic underlying ontological model \cite{HarriganSpekkens} could reproduce the operational statistics of (projective) quantum measurements in a manner that satisfies the assumption of {\em KS-noncontextuality}. KS-noncontextuality is the notion of classicality at play in the Kochen-Specker theorem. The Spekkens framework abandons the assumption of outcome determinism \cite{Spe05} -- the idea that the ontic state of a system fixes the outcome of any measurement deterministically -- that is intrinsic to KS-noncontextuality. It also applies to general operational theories and extends the notion of noncontextuality to general experimental procedures -- preparations, transformations, and measurements -- rather than measurements alone.

Parallel to work along the lines of Spekkens \cite{Spe05}, work seeking to directly operationalize the Kochen-Specker theorem (rather than revising the notion of noncontextuality at play) culminated in two recent approaches that classify theories by the degree to which they violate the assumption of KS-noncontextuality: the graph-theoretic framework of Cabello, Severini, and Winter (CSW) \cite{CSW10,CSW}, where a general approach to obtaining graph-theoretic bounds on linear Bell-KS functionals was proposed, and the related hypergraph framework of Ac\'in, Fritz, Leverrier, and Sainz (AFLS) \cite{AFLS}, where an approach to characterizing sets of correlations was proposed. The CSW framework relates well-known graph invariants to: 1) upper bounds on Bell-KS inequalities that follow from KS-noncontextuality, 2) upper bounds on maximum quantum violations of these inequalities that can be obtained from projective measurements, and 3) upper bounds on their violation in general probabilistic theories \cite{GPTs} -- denoted $E1$ -- which satisfy the ``exclusivity principle" \cite{CSW}. Complementary to this, the AFLS framework uses graph invariants in the service of deciding whether a given assignment of probabilities to measurement outcomes in a KS-contextuality experiment belongs to a particular set of correlations; they showed that membership in the quantum set of correlations (defined {\em only} for projective measurements in quantum theory) cannot be witnessed by a graph invariant, cf.~Theorem 5.1.3 of Ref.~\cite{AFLS}. Another recent approach due to Abramsky and Brandenburger \cite{sheaf} employs sheaf-theoretic ideas to formulate KS-contextuality. 

A key achievement of the frameworks of Refs.~\cite{CSW, AFLS, sheaf} is a formal unification of Bell scenarios with KS-contextuality scenarios, treating them on the same footing. Indeed, the perspective there is to consider Bell scenarios as a special case of KS-contextuality scenarios. What is lost in this mathematical unification, however, is the fact that Bell-locality and KS-noncontextuality have physically distinct, if related, motivations. 
The physical situation that Bell's theorem refers to requires (at least) two spacelike separated labs (where local measurements are carried out) so that the assumption of local causality (or Bell-locality) can be applied.\footnote{What do we mean by whether an assumption ``can be applied"? Of course, mathematically, one can ``apply" any assumption one wants in the service of proving a theorem. But insofar as the mathematics here is trying to model a real experiment, the consistency of those assumptions with some essential facts of the experiment is the minimal requirement for any no-go theorem derived from such assumptions to be physically interesting. Hence, in the presence of signalling (implying the absence of spacelike separation), it makes no sense to assume local causality in a Bell experiment and derive the resulting Bell inequalities: such an assumption on the ontological model is already in conflict with the fact of signalling across the labs and no Bell inequalities are needed to witness this fact. Bell inequalities only become physically interesting when the theories being compared relative to them are all non-signalling: if the experiment itself is signalling, any non-signalling description -- locally causal, quantum, or in a general probabilistic theory (GPT) -- is {\em ipso facto} ruled out.} On the other hand, the physical situation that the Kochen-Specker theorem refers to does not require spacelike separation as a necessary ingredient and one can therefore consider experiments in a single lab. However, the assumption of KS-noncontextuality entails outcome determinism \cite{Spe05}, something not required by local causality in Bell scenarios.\footnote{Note that this assumption of outcome determinism doesn't affect the conclusions in a Bell scenario even if one adopted it because of Fine's theorem \cite{fine}: a locally deterministic ontological model entails the same set of (Bell-local) correlations as a locally causal ontological model. Relaxing outcome determinism, however, doesn't mean the same thing for the kinds of experiments envisaged by the Kochen-Specker theorem -- in particular, it doesn't mean that models satisfying factorizability \`a la Ref.~\cite{sheaf} are the most general outcome-indeterministic models --  and thus considerations parallel to Fine's theorem \cite{fine} do not apply, cf.~\cite{finegen, ODUM}.} This difference in the physical situation for the two kinds of experiments is one of the reasons for generalizing KS-noncontextuality to the notion of noncontextuality in the Spekkens framework \cite{Spe05} (so that outcome determinism is not assumed) while leaving Bell's notion of local causality untouched.

In the present paper we build a bridge from the CSW approach, where KS-noncontextual correlations are bounded by Bell-KS inequalities, to noise-robust noncontextuality inequalities in the Spekkens framework \cite{Spe05}. That is, we show how the constraints from KS-noncontextuality in the framework of Ref.~\cite{CSW} translate to constraints from generalized  noncontextuality in the framework of Ref.~\cite{Spe05}. The resulting operational criteria for contextuality \`a la Spekkens are noise-robust and therefore applicable to arbitrary positive operator-valued measures (POVMs) and mixed states in quantum theory.  Note that the insights gleaned from frameworks such as those of Refs.~\cite{CSW,AFLS,sheaf} regarding Bell nonlocality require no revision in our approach. It is {\em only} in the application of such frameworks (in particular, the CSW framework) to the question of contextuality that we seek to propose an alternative hypergraph framework (formalizing Spekkens contextuality \cite{Spe05}) that is more operationally motivated for experimental situations where one cannot appeal to spacelike separation to justify locality of the measurements.\footnote{Nor the sharpness of the measurements to justify outcome determinism. We discuss these issues in detail -- in particular, the physical basis of KS-noncontextuality vis-\`a-vis Bell-locality and how that influences our framework --  in Appendix \ref{allegedstrawman} for the interested reader.} For Kochen-Specker type experimental scenarios, we will consider the twin notions of preparation noncontextuality and measurement noncontextuality -- taken together as a notion of classicality -- to obtain noise-robust noncontextuality inequalities that generalize the KS-noncontextuality inequalities of CSW. These inequalities witness nonclassicality even when quantum correlations arising from arbitrary (i.e., possibly nonprojective) quantum measurements on any quantum state are allowed. A key innovation of this approach is that it treats all measurements in an operational theory on an equal footing. No definition of ``sharpness" \cite{cabellotalk,giuliosharp,giuliosharp2} is needed to justify or derive noncontextuality inequalities in this approach. Furthermore, if certain idealizations are presumed about the operational statistics, then these inequalities formally recover the usual Bell-KS inequalities \`a la CSW. The Bell-KS inequalities can be viewed as an instance of the classical marginal problem \cite{fine,finegen,ChavesFritz1, ChavesFritz2,sheaf}, i.e., as constraints on the (marginal) probability distributions over subsets of a set of observables that follow from requiring the existence of global joint probability distribution over the set of all observables. Since the Bell-KS inequalities are only recovered under certain idealizations, but not otherwise, the noise-robust noncontextuality inequalities we obtain {\em cannot} in general be viewed as arising from a classical marginal problem. Hence, they cannot be understood within existing frameworks that rely on this (reduction to the classical marginal problem) property to formally unify the treatment of Bell-nonlocality and KS-contextuality \cite{CSW, AFLS, sheaf}. This is a {\em crucial} distinction relative to the usual Bell-KS inequality type witnesses of KS-contextuality.

This paper is based on a previous contribution \cite{KunjSpek17} that laid the conceptual groundwork for the progress we make here. Besides the noise-robust noncontextuality inequalities that generalize constraints from KS-noncontextuality in the CSW framework leveraging the graph invariants of CSW \cite{CSW} (cf.~Section \ref{derivingineqs}), the contributions of this paper also include:	
	\begin{itemize}
		\item An exposition of Specker's principle and how different implications of it (e.g., consistent exclusivity \cite{AFLS}) for a given operational theory arise in the hypergraph framework (cf.~Sections \ref{Speck1} and \ref{Speck2}), in particular the results in Theorems \ref{SPimpliesStrSP}, \ref{strimpliesstat}, and Corollary \ref{SPstrSPCE}.
		
		\item Introduction of a hypergraph invariant -- the weighted max-predictability -- that is key to our noise-robust noncontextuality inequalities, cf.~Section \ref{betainvariant}. This invariant is also key to the hypergraph framework of Ref.~\cite{kunjunc} which is complementary to the present framework.
		
		\item A detailed discussion of how KS-noncontextuality for POVMs has been previously treated in the literature and the limitations of those treatments, cf.~Appendices \ref{allegedstrawman} and \ref{trivial}. Also, unlike for the case of KS-noncontextuality inequalities, we show that trivial POVMs can never violate our noise-robust noncontextuality inequalities, cf.~Section \ref{trivialnoviol}.
		
		\item A discussion of coarse-graining relations in Section \ref{coarsegrainings} and their importance for contextuality no-go theorems, in particular a discussion of ontological models that do not respect coarse-graining relations in Appendix \ref{ontmodelwocoarsegraining}. We show, in Appendix \ref{ontmodelwocoarsegraining}, how relaxing the constraint from coarse-graining relations on an ontological model renders either notion of noncontextuality -- whether Kochen-Specker \cite{KochenSpecker} or Spekkens \cite{Spe05} -- vacuous.
		
		\item A discussion, by example, of why our  generalization of the CSW framework cannot accommodate contextuality scenarios that are KS-uncolourable in Appendix \ref{scopeviacega} and why one needs a distinct framework, i.e., the framework of  Ref.~\cite{kunjunc}, to treat KS-uncolourable scenarios.
	\end{itemize}
The structure of this paper follows: Section \ref{spekframework} reviews the Spekkens framework for generalized noncontextuality \cite{Spe05}. Section \ref{hypergraphsection} introduces a hypergraph framework that shares features of traditional frameworks for KS-contextuality \cite{CSW,AFLS} but is also augmented (relative to these traditional frameworks) with the ingredients necessary for obtaining noise-robust noncontextuality inequalities. In particular, its subsections \ref{Speck1} and \ref{Speck2} discuss Specker's principle \cite{speckerprinciple} and define its different implications for contextuality scenarios \`a la Ref.~\cite{AFLS}. Section \ref{betainvariant} defines a new hypergraph invariant -- the weighted max-predictability -- that we need later on as a crucial new ingredient in our inequalities. Section \ref{derivingineqs} obtains noise-robust noncontextuality inequalities within the framework defined in Section \ref{hypergraphsection} and using the hypergraph invariant of Section \ref{betainvariant} in addition to two graph invariants from the CSW framework \cite{CSW}. These inequalities can be seen as special cases of the general approach outlined in Ref.~\cite{KunjSpek17}. In Section \ref{discussion}, we include discussions on various features of our noise-robust noncontextuality inequalities, in particular the fact that trivial POVMs can never violate them. Section \ref{conclusions} concludes with some open questions and directions for future research.

\section{Spekkens framework}\label{spekframework} We concern ourselves with prepare-and-measure experiments. A schematic of such an experiment is shown in Figure \ref{prepmsr} where, for the sake of simplicity, we imagine a single source device that can perform any preparation procedure of interest (rather than a collection of source devices, each implementing a particular preparation procedure) and a single measurement device that can perform any measurement procedure of interest (rather than a collection of measurement devices, each implementing a particular measurement procedure). Note that this is just a conceptual abstraction: in particular, the various possible measurement settings on the measurement device may, for example, correspond to incompatible measurement procedures in quantum theory. The fact that we represent the different measurement settings by choices of knob settings $M\in\mathbb{M}$ on a single measurement device {\em does not} mean that it's physically possible to implement all the measurement procedures represented by $\mathbb{M}$ jointly; it only means that the experimenter can choose to implement any of the measurements in the set $\mathbb{M}$ in a particular prepare-and-measure experiment. The same is true for our abstraction of preparation procedures to knob settings ($S\in\mathbb{S}$) and outcomes ($s\in V_S$) of a single source device: it's not that the same device can physically implement all possible preparation procedures; it's just that an experimenter can choose to implement any procedure in the set $\mathbb{S}$ in a particular prepare-and-measure experiment.

\begin{figure}
	\centering
	\includegraphics[scale=0.5]{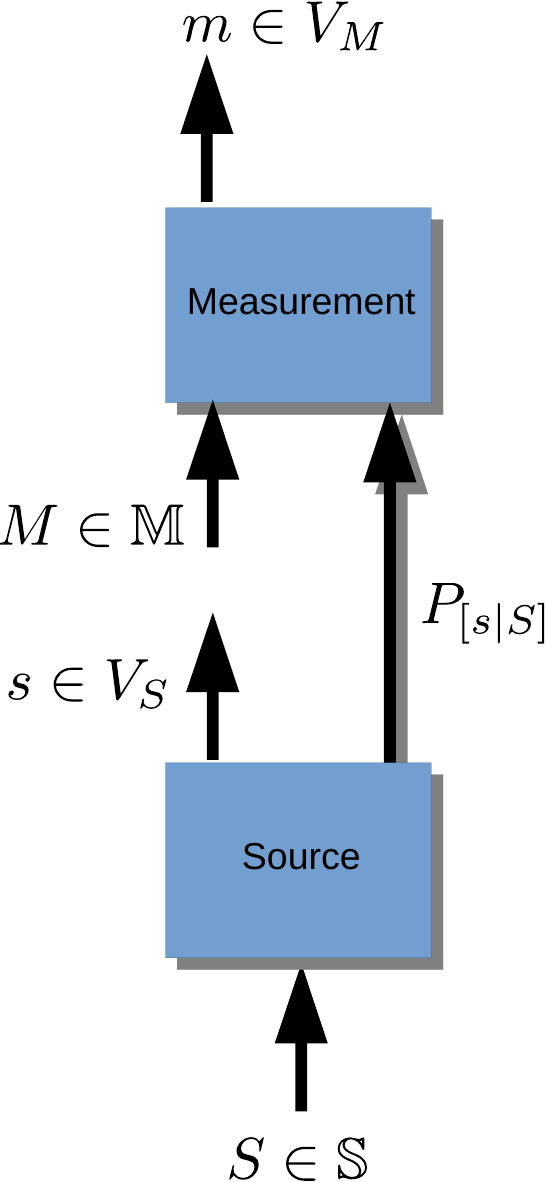}
	\caption{A prepare-and-measure experiment.}
	\label{prepmsr}
\end{figure}

We will consider two levels of description of prepare-and-measure experiments represented by Fig.~\ref{prepmsr}: {\em operational} and {\em ontological}. The operational description will be specified by an operational theory that takes  source and measurement devices as primitives and describes the experiment solely in terms of the probabilities associated to their input/output behaviour. The ontological description will be specified by an ontological model that takes the system that passes between the source and measurement devices as primitive and describes the experiment in terms of probabilities associated to properties of this system, deriving the operational description as a consequence of coarse-graining over these properties. Let us look at each description in turn.

\subsection{Operational theory}	
We now describe the various components of Fig.~\ref{prepmsr} in more detail.	
The source device has a {\em source setting} labelled by $S$ that can be chosen from a set $\mathbb{S}$. The set $\mathbb{S}$ represents, in general, some subset of the set of all source settings, $\mathscr{S}$, that are admissible in the operational theory, i.e., $\mathbb{S}\subseteq \mathscr{S}$. In a particular prepare-and-measure experiment, $\mathbb{S}$ will typically be a finite set of source settings. Choosing the setting $S$ prepares a system according to an ensemble of preparation procedures, denoted $\{(p(s|S),P_{[s|S]})\}_{s\in V_S}$, where $\{p(s|S)\}_{s\in V_S}$ is a probability distribution over the preparation procedures $\{P_{[s|S]}\}_{s\in V_S}$ in the ensemble. This means that the source device has one classical input $S$ and two outputs: one output is a classical label $s\in V_S$ identifying the preparation procedure (in the ensemble $\{(p(s|S),P_{[s|S]})\}_{s\in V_S}$) that is carried out when source outcome $s$ is observed for source setting $S$ (this {\em source event} is denoted $[s|S]$), and the other output is a system prepared according to the source event $[s|S]$, i.e., preparation procedure $P_{[s|S]}$, with probability $p(s|S)$.
Thus, the assemblage of possible ensembles that the source device can prepare can be denoted by $\{\{(p(s|S),P_{[s|S]})\}_{s\in V_S}\}_{S\in\mathbb{S}}$.

On the other hand, the measurement device has two inputs, one a classical input $M\in\mathbb{M}$ specifying the choice of measurement setting to be implemented, and the other input receives the system prepared according to prepartion procedure $P_{[s|S]}$ and on which this measurement $M$ is carried out. The measurement device has one classical output $m\in V_M$ denoting the outcome of the measurement $M$ implemented on a system prepared according to $P_{[s|S]}$, and which occurs with probability $p(m|M,S,s)$. The set $\mathbb{M}$ represents, in general, some subset of the set of all measurement settings, $\mathscr{M}$, that are admissible in the operational theory, i.e., $\mathbb{M}\subseteq\mathscr{M}$. In a particular prepare-and-measure experiment, $\mathbb{M}$ will typically be a finite set of measurement settings.

We will be interested in the operational joint probability $p(m,s|M,S)\equiv p(m|M,S,s)p(s|S)$ for this prepare-and-measure experiment for various choices of $M\in\mathbb{M}$,  $S\in\mathbb{S}$. Note how this operational description takes as primitive the {\em operations} carried out in the lab and restricts itself to specifying the probabilities of classical outcomes (i.e., $m,s$) given some interventions (i.e., classical inputs, $M,S$). So far, we haven't assumed any structure on the operational theory describing the schematic of Fig.~\ref{prepmsr} beyond the fact that it is a catalogue of input/output probabilities $$\{\{p(m,s|M,S)\in[0,1]\}_{m\in V_M, s\in V_S}\}_{M\in \mathbb{M}, S\in \mathbb{S}}$$ for various interventions $S\in\mathbb{S}$ and $M\in\mathbb{M}$ that we will consider in a prepare-and-measure experiment. We now require more structure in the operational theory underlying this experiment, beyond a mere specification of these probabilities.

We require that the operational theory admits {\em equivalence relations} that partition experimental procedures of any type, whether preparations or measurements, into equivalence classes of that type. These equivalence relations are defined relative to the operational probabilities (not necessarily restricted to a particular prepare-and-measure experiment) that are admissible in the theory. We will call these equivalence relations ``operational equivalences", in keeping with standard terminology \cite{Spe05}. This means that any distinctions of labels between procedures in an equivalence class of procedures do not affect the operational probabilities associated with the procedures. We specify these equivalence relations for measurement and preparation procedures below.

Two measurement events $[m|M]$ and $[m'|M']$ are said to be operationally equivalent, denoted $[m|M]\simeq[m'|M']$, if there exists no source event in the operational theory that can distinguish them, i.e.,
\begin{equation}\label{opequiv1mmt}
p(m,s|M,S)=p(m',s|M',S)\quad\forall [s|S], s\in V_S, S\in\mathscr{S}.
\end{equation}
Note that the statistical indistinguishability of $[m|M]$ and $[m'|M']$ must hold for all possible source settings $\mathscr{S}$ in the operational theory, not merely the source settings $\mathbb{S}$ that are of direct interest in a particular prepare-and-measure experiment. Hence, the ``distinction of labels", $[m|M]$ or $[m'|M']$, is empirically inconsequential since the two procedures are, in principle, indistinguishable by the lights of the operational theory.

Similarly, two source events $[s|S]$ and $[s'|S']$ are said to be operationally equivalent, denoted $[s|S]\simeq [s'|S']$, if there exists no measurement event in the operational theory that can distinguish them, i.e.,
\begin{align}\label{opequiv1prep}
&p(m,s|M,S)=p(m,s'|M,S'),\nonumber\\
&\forall [m|M], m\in V_M, M\in\mathscr{M}.
\end{align}
Again, the statistical indistinguishability of $[s|S]$ and $[s'|S']$ must hold for all possible measurement settings $\mathscr{M}$, not merely those (i.e., $\mathbb{M}$) that are of direct interest in a particular prepare-and-measure experiment. Similar to measurement events, the ``distinction of labels", $[s|S]$ or $[s'|S']$, is empirically inconsequential since the two procedures are, in principle, indistinguishable by the lights of the operational theory.

Given this equivalence structure for preparation and measurement procedures in the operational theory, we can now formalize the notion of a {\em context}: 

	\begin{definition}
		A context is any distinction of labels between operationally equivalent procedures in the operational theory.
	\end{definition}
To see concrete examples of the kinds of contexts that will be of interest to us in this paper, consider quantum theory. Any mixed quantum state admits multiple convex decompositions in terms of other quantum states, i.e., it can be prepared by coarse-graining over distinct ensembles of quantum states, each ensemble denoted by a different label. In this case, the ``distinction of labels" between different decompositions denotes a distinction of preparation ensembles, which instantiates our notion of a preparation context. Similarly, a given positive operator can be implemented by different positive operator-valued measures (POVMs), and the distinction of labels denoting these different POVMs instantiates our notion of a measurement context.

\subsection{Ontological model}
Given the operational description of the experiment in terms of probabilities $p(m,s|M,S)$, we want to explore the properties of any underlying ontological model for this operational description. Any such ontological model, defined within the ontological models framework \cite{HarriganSpekkens}, takes as primitive the {\em physical system} (rather than {\em operations} on it) that passes between the source and measurement devices, i.e., its basic objects are {\em ontic states} of the system, denoted $\lambda\in\Lambda$, that represent intrinsic properties of the physical system. When a preparation procedure $[s|S]$ is carried out, the source device samples from the space of ontic states $\Lambda$ according to a probability distribution $\{\mu(\lambda|S,s)\in[0,1]\}_{\lambda\in\Lambda}$, where $\sum_{\lambda\in\Lambda}\mu(\lambda|S,s)=1$, and the joint distribution over $s$ and $\lambda$ given $S$, i.e., $\{\mu(\lambda,s|S)\}_{\lambda\in\Lambda}$, is given by $\mu(\lambda,s|S)\equiv\mu(\lambda|S,s)p(s|S)$. On the other hand, when a system in ontic state $\lambda$ is input to the measurement device with measurement setting $M\in\mathbb{M}$, the probability distribution over the measurement outcomes is given by $\{\xi(m|M,\lambda)\in[0,1]\}_{m\in V_M}$, where $\sum_{m\in V_M}\xi(m|M,\lambda)=1$. The operational statistics $$\{\{p(m,s|M,S)\in[0,1]\}_{m\in V_M, s\in V_S}\}_{M\in \mathbb{M}, S\in \mathbb{S}}$$ results from a coarse-graining over $\lambda$, i.e.,
\begin{equation}
p(m,s|M,S)=\sum_{\lambda\in\Lambda}\xi(m|M,\lambda)\mu(\lambda,s|S),
\end{equation}
for all $m\in V_M, s\in V_S, M\in\mathbb{M}, S\in\mathbb{S}$.

Note that the definition of an ontological model above extends to the definition of an ontological model of the operational {\em theory} (as opposed to a particular fragment of the theory representing the experiment) when we take $\mathbb{M}=\mathscr{M}$ and $\mathbb{S}=\mathscr{S}$.

\subsection{Representation of coarse-graining}\label{coarsegrainings}
We will now specify how coarse-graining of procedures in a prepare-and-measure experiment is represented in its description, whether operational or ontological. Namely, if a procedure is defined as a coarse-graining of other procedures, then we require that the representation of such a procedure is defined by the same coarse-graining of the representation of the other procedures.\footnote{Quantum theory is an example of an operational theory that satisfies this requirement because of the linearity of the Born rule with respect to both preparations and measurements. The same is true, more generally, of general probabilistic theories (GPTs) \cite{hardy5axioms, GPTs}. We require this feature in any ontological model as well, regardless of its (non)contextuality.} Implicit in this discussion is the assumption that the operational theory allows one to define new procedures in the set $\mathscr{M}$ or $\mathscr{S}$ by coarse-graining other procedures in these sets, i.e., both $\mathscr{M}$ and $\mathscr{S}$ are closed under coarse-grainings. In particular, one can consider coarse-graining measurement and source settings (belonging to sets $\mathbb{M}$ and $\mathbb{S}$, respectively) actually implemented in the lab to define new measurement and source settings that belong to $\mathscr{M}\backslash\mathbb{M}$ and $\mathscr{S}\backslash\mathbb{S}$, respectively.\footnote{Similarly, we also allow probabilistic mixtures of (preparation or measurement) procedures in the operational theory to define new procedures, i.e., the theory is convex. See the last paragraph of Section \ref{noncontextualitysection} for the role of this convexity in experimental tests of contextuality and Section \ref{spekexample} for an example where a  probabilistic mixture of measurement settings is required in a proof of contextuality.}

\subsubsection{Coarse-graining of measurements}
Let us see how this works for the case of measurement procedures: if a measurement procedure $M$ with measurement events $\{[m|M]\}_{m\in V_M}$ is {\em defined} as a coarse-graining of another measurement procedure $\tilde{M}$ with measurement events $\{[\tilde{m}|\tilde{M}]\}_{\tilde{m}\in V_{\tilde{M}}}$, symbolically denoted by 
\begin{align}\label{mmtpostprocess}
&[m|M]\equiv \sum_{\tilde{m}}p(m|\tilde{m})[\tilde{m}|\tilde{M}],\nonumber\\
&\textrm{where }\forall m,\tilde{m}: p(m|\tilde{m})\in\{0,1\}, \sum_{m}p(m|\tilde{m})=1,
\end{align}
then its representation in the operational description as well as in the ontological description satisfies this coarse-graining relation.\footnote{Note that Eq.~\eqref{mmtpostprocess} is {\em not} an operational equivalence between independent procedures. It is a {\em definition} of a new procedure obtained by coarse-graining another procedure.} More explicitly, the coarse-graining relation of Eq.~\eqref{mmtpostprocess} denotes the following post-processing of $\tilde{M}$: for each $m\in V_M$, relabel each outcome $\tilde{m}\in V_{\tilde{M}}$ to outcome $m$ with probability $p(m|\tilde{m})\in\{0,1\}$; the logical disjunction of those $\tilde{m}$ which are relabelled to $m$ with probability $1$ then defines the measurement event $[m|M]$. Now, in the operational theory, this post-processing is represented by 
\begin{align}
&\forall [s|S],\textrm{ where } s\in V_S, S\in\mathscr{S}:\nonumber\\
&p(m,s|M,S)\equiv \sum_{\tilde{m}}p(m|\tilde{m})p(\tilde{m},s|\tilde{M},S),
\end{align}
and in the ontological model it is represented by 
\begin{equation}
\forall \lambda\in\Lambda: \xi(m|M,\lambda)\equiv \sum_{\tilde{m}}p(m|\tilde{m})\xi(\tilde{m}|\tilde{M},\lambda).
\end{equation}  

As an example, consider a three-outcome measurement $\tilde{M}$ with outcomes $\tilde{m}\in\{1,2,3\}$, which can be classically post-processed to obtain a two-outcome measurement $M$ with outcomes $m\in\{0,1\}$, such that $p(m=0|\tilde{m}=1)=p(m=0|\tilde{m}=2)=1$ and $p(m=1|\tilde{m}=3)=1$. The measurement events of $M$ are then just
\begin{align}
[m=0|M]&\equiv [\tilde{m}=1|\tilde{M}]+[\tilde{m}=2|\tilde{M}],\\
[m=1|M]&\equiv [\tilde{m}=3|\tilde{M}],
\end{align}
where the ``$+$" sign denotes (just as the summation sign in the definition of $[m|M]$ in Eq.~\eqref{mmtpostprocess} did) logical disjunction, i.e., measurement event $[m=0|M]$ is said to occur when $[\tilde{m}=1|\tilde{M}]$ or $[\tilde{m}=2|\tilde{M}]$ occurs. The operational and ontological representations of these measurement events are then given by 
\begin{align}
&\forall [s|S],\textrm{ where } s\in V_S, S\in\mathscr{S}:\nonumber\\
&p(m=0,s|M,S)\equiv \sum_{\tilde{m}=1}^2p(\tilde{m},s|\tilde{M},S),\\
&p(m=1,s|M,S)\equiv p(\tilde{m}=3,s|\tilde{M},S),\\\nonumber\\
&\forall \lambda\in \Lambda:\nonumber\nonumber\\
&\xi(m=0|M,\lambda)\equiv \sum_{\tilde{m}=1}^2\xi(\tilde{m}|\tilde{M},\lambda),\\
&\xi(m=1|M,\lambda)\equiv \xi(\tilde{m}=3|\tilde{M},\lambda).
\end{align}

This requirement on the representation of coarse-graining of measurements is particularly important (and often implicit) when the notion of a measurement context is instantiated by compatibility (or joint measurability), as in the case of KS-contextuality, where one needs to consider coarse-grainings of distinct measurements. For example, consider a measurement setting $M_{12}$ with outcomes $(m_1,m_2)\in V_1\times V_2$ that is coarse-grained over $m_2$ to define an effective measurement setting $M_1^{(2)}$ with measurement events $\{[m_1|M_1^{(2)}]\}_{m_1\in V_1}$. Symbolically, $[m_1|M_1^{(2)}]\equiv \sum_{m_2}[(m_1,m_2)|M_{12}]$, which is represented in the operational theory as $\forall [s|S]: p(m_1,s|M_1^{(2)},S)\equiv \sum_{m_2}p((m_1,m_2),s|M_{12},S)$ and in the ontological model as $\forall \lambda: \xi(m_1|M_1^{(2)},\lambda)\equiv \sum_{m_2}\xi((m_1,m_2)|M_{12},\lambda)$. Similarly, consider another measurement setting $M_{13}$ with outcomes $(m_1,m_3)\in V_1\times V_3$ that is coarse-grained over $m_3$ to define an effective measurement setting $M_1^{(3)}$ with measurement events $\{[m_1|M_1^{(3)}]\}_{m_1\in V_1}$. Symbolically, $[m_1|M_1^{(3)}]\equiv \sum_{m_3}[(m_1,m_3)|M_{13}]$, which is represented in the operational theory as $\forall [s|S]: p(m_1,s|M_1^{(3)},S)\equiv \sum_{m_3}p((m_1,m_3),s|M_{13},S)$ and in the ontological model as $\forall \lambda: \xi(m_1|M_1^{(3)},\lambda)\equiv \sum_{m_3}\xi((m_1,m_3)|M_{13},\lambda)$. 

Now, imagine that the following operational equivalence holds at the operational level: $[m_1|M_1^{(2)}]\simeq[m_1|M_1^{(3)}]$. KS-noncontextuality is then the assumption that  $\sum_{m_2}\xi((m_1,m_2)|M_{12},\lambda)=\sum_{m_3}\xi((m_1,m_3)|M_{13},\lambda)$ (i.e., $\xi(m_1|M_1^{(2)},\lambda)=\xi(m_1|M_1^{(3)},\lambda)$) for all $\lambda$ and that  $\xi((m_1,m_2)|M_{12},\lambda),\xi((m_1,m_3)|M_{13},\lambda)\in\{0,1\}$ for all $\lambda$. This assumption applied to multiple (compatible) subsets of a set of carefully chosen measurements can then provide a proof of the KS theorem, i.e., there exist sets of measurements in quantum theory such that their operational statistics cannot be emulated by a KS-noncontextual ontological model. 

The key point here is this: the requirement that coarse-graining relations between measurements be respected by their representations in the ontological model is independent of the KS-(non)contextuality of the ontological model.\footnote{In our example, this requirement has to do with the definitions of $\xi(m_1|M_1^{(2)},\lambda)$ and $\xi(m_1|M_1^{(3)},\lambda)$, not their ontological equivalence. The ontological equivalence only comes into play when invoking KS-noncontextuality.} However, this requirement is {\em necessary} for the assumption of KS-noncontextuality to produce a contradiction with quantum theory; on the other hand, a KS-contextual ontological model (while respecting the coarse-graining relations) can always emulate quantum theory. In this sense, the representation of coarse-grainings is baked into an ontological model from the beginning (just as it is baked into an operational description), before any claims about its (non)contextuality.\footnote{One could, of course, choose to not respect the coarse-graining relations and define a notion of an ontological model without them. In such a model, one could treat every measurement obtained by coarse-graining another (parent) measurement as a fundamentally new measurement with response functions {\em not} respecting the coarse-graining relations with the parent measurement's response functions, even if such coarse-graining relations are respected in the operational description. Such an ontological model, however, will not be able to articulate the ingredients needed for a proof of the KS theorem and we will not consider it here. Indeed, in the absence of the requirement that coarse-graining relations be respected in an ontological model, one can easily construct an ontological model that is ``KS-noncontextual" for any operational theory. The interested reader may look at Appendix \ref{ontmodelwocoarsegraining} for more details, perhaps after looking at Section \ref{noncontextualitysection} for the relevant definitions of noncontextuality.}

\subsubsection{Coarse-graining of preparations}
Let us now consider the representation of coarse-grainings for preparation procedures. This works in a way similar to the case of measurement procedures which we have already outlined. If an ensemble of source events $\{[s|S]\}_{s\in V_S}$ is defined as a coarse-graining of another ensemble, $\{[\tilde{s}|\tilde{S}]\}_{\tilde{s}\in V_{\tilde{S}}}$, symbolically denoted as
\begin{align}\label{preppostprocess}
&[s|S]\equiv\sum_{\tilde{s}}p(s|\tilde{s})[\tilde{s}|\tilde{S}], \textrm{where }\nonumber\\
&\forall s,\tilde{s}: p(s|\tilde{s})\in\{0,1\}, \sum_sp(s|\tilde{s})=1,
\end{align}
then its representation should satisfy the same coarse-graining relation in any description, operational or ontological. More explicitly, this coarse-graining denotes the following post-processing: for any $s\in V_S$, relabel each outcome $\tilde{s}\in V_{\tilde{S}}$ to outcome $s$ with probability $p(s|\tilde{s})\in\{0,1\}$; the logical disjunction of those $\tilde{s}$ which are relabelled to $s$ with probability 1 then defines the source event $[s|S]$. Now, in the operational theory, this coarse-graining is represented by 
\begin{align}
&\forall [m|M],\textrm{ where } m\in V_M, M\in\mathscr{M}:\nonumber\\
&p(m,s|M,S)\equiv \sum_{\tilde{s}}p(s|\tilde{s})p(m,\tilde{s}|M,\tilde{S}),
\end{align}
and in the ontological model it is represented by 
\begin{equation}
\forall \lambda\in\Lambda: \mu(\lambda,s|S)\equiv \sum_{\tilde{s}}p(s|\tilde{s})\mu(\lambda,\tilde{s}|\tilde{S}).
\end{equation}
In this paper, we will focus on a specific type of coarse-graining: namely, completely coarse-graining over the outcomes of a source setting, say $\{[\tilde{s}|\tilde{S}]\}_{\tilde{s}\in V_{\tilde{S}}}$, to yield an effective one-outcome source-setting, denoted $\tilde{S}_{\top}$, associated with a single source event $\{[\top|\tilde{S}_{\top}]\}$, where $[\top|\tilde{S}_{\top}]\equiv \sum_{\tilde{s}}[\tilde{s}|\tilde{S}]$. In the operational theory, this coarse-graining is represented by
\begin{align}
&\forall [m|M],\textrm{ where } m\in V_M, M\in\mathscr{M}:\nonumber\\
&p(m,\top|M,\tilde{S}_{\top})\equiv \sum_{\tilde{s}}p(m,\tilde{s}|M,\tilde{S}),
\end{align}
and in the ontological model it is represented by 
\begin{equation}
\forall \lambda\in\Lambda: \mu(\lambda,\top|\tilde{S}_{\top})\equiv \sum_{\tilde{s}}\mu(\lambda,\tilde{s}|\tilde{S}).
\end{equation}
Hence, we use the notation $[\top|\tilde{S}_{\top}]$ to denote the source event that ``at least one of the source outcomes in the set $V_{\tilde{S}}$ occurs for source setting $\tilde{S}$" (i.e., the logical disjunction of $\tilde{s}\in V_{\tilde{S}}$), formally denoting the choice of $\tilde{S}$ and the subsequent coarse-graining over $\tilde{s}$ by the ``source setting" $\tilde{S}_{\top}$ and the definite outcome of this source setting by ``$\top$". This source event always occurs, i.e., $p(\top|\tilde{S}_{\top})=1$, so $p(m,\top|M,\tilde{S}_{\top})=p(m|M,\tilde{S}_{\top},\top)$ and $\mu(\lambda,\top|\tilde{S}_{\top})= \mu(\lambda|\tilde{S}_{\top},\top)$.

This notion of coarse-graining over all the outcomes of a source setting allows us to define a notion of operational equivalence between the source settings themselves. More precisely, two {\em source settings} $S$ and $S'$ are said to be operationally equivalent, denoted $[\top|S_{\top}]\simeq [\top|S'_{\top}]$, if no measurement event can distinguish them once all their outcomes are coarse-grained over, i.e.,
\begin{eqnarray}\label{opequiv2prep}
&&\sum_{s\in V_S}p(m,s|M,S)=\sum_{s'\in V_{S'}}p(m,s'|M,S')\nonumber\\
&\quad&\forall [m|M], m\in V_M, M\in\mathscr{M}.
\end{eqnarray}

In quantum theory, this would correspond to the operational equivalence $\sum_sp(s|S)\rho_{[s|S]}=\sum_{s'}p(s'|S')\rho_{[s'|S']}$ for the density operator obtained by completely coarse-graining over two distinct ensembles of quantum states, $\{(p(s|S),\rho_{[s|S]})\}_{s\in V_S}$ and $\{(p(s'|S'),\rho_{[s'|S']})\}_{s'\in V_{S'}}$ on some Hilbert space $\mathcal{H}$.

\subsection{Joint measurability (or compatibility)}\label{compatibilitysection}
A given measurement procedure, $\{[m|M]\}_{m\in V_M}$ for some $M\in\mathscr{M}$, in the operational description can be coarse-grained in many different ways to define new effective measurement procedures. The coarse-grained measurement procedures thus obtained from $\{[m|M]\}_{m\in V_M}$ are then said to be jointly measurable (or compatible), i.e., they can be jointly implemented by the same measurement procedure $\{[m|M]\}_{m\in V_M}$ which we refer to as their parent or joint measurement. Formally, a set $\mathcal{C}$ of measurement procedures $$\{\{[m_i|M_i]\}_{m_i\in V_{M_i}}\big|i\in\{1,2,3,\dots,|\mathcal{C}|\}\}$$
is said to be jointly measurable (or compatible) if it arises from coarse-grainings of a single measurement procedure $M\in\mathscr{M}$, i.e., for all $ \{[m_i|M_i]\}_{m_i\in V_{M_i}}\in\mathcal{C}$
\begin{equation}
[m_i|M_i]\equiv\sum_{m\in V_M}p(m_i|m)[m|M], 
\end{equation}
where for all $i,m,m_i$: $p(m_i|m)\in\{0,1\}$ and $\sum_{m_i\in V_{M_i}}p(m_i|m)=1$. In terms of the operational probabilities, this means that
\begin{align}
&\forall [s|S], s\in V_S, S\in\mathscr{S}\textrm{ and }\forall \{[m_i|M_i]\}_{m_i\in V_{M_i}}\in\mathcal{C}:\nonumber\\
&p(m_i,s|M_i,S)\equiv \sum_{m\in V_M}p(m_i|m)p(m,s|M,S).
\end{align}

If, on the other hand, a set of measurement procedures cannot arise from coarse-grainings of any single measurement procedure, then the measurement procedures in the set are said to be {\em incompatible}, i.e., they cannot be jointly implemented. 

Note that we will also often refer to a measurement procedure $\{[m_i|M_i]\}_{m_i\in V_{M_i}}$ by just its measurement setting, $M_i$, and thus speak of the (in)compatibility of measurement settings. Another notion that we will need to refer to is the joint measurability of measurement events: a set of measurement events that arise as outcomes of a single measurement setting are said to be jointly measurable, e.g., all the measurement events in $\{[m|M]\}_{m\in V_M}$ are jointly measurable since they arise as outcomes of a single measurement setting $M$.

As a quantum example, consider a commuting pair of projective measurements, say $\{\Pi_1,I-\Pi_1\}$ and $\{\Pi_2,I-\Pi_2\}$, where $\Pi_1$ and $\Pi_2$ are projectors on some Hilbert space $\mathcal{H}$ such that $\Pi_1\Pi_2=\Pi_2\Pi_1$ and $I$ is the identity operator on $\mathcal{H}$. This pair is jointly implementable since they can be obtained by coarse-graining the outcomes of the joint projective measurement given by  $\{\Pi_1\Pi_2,\Pi_1(I-\Pi_2),(I-\Pi_1)\Pi_2,(I-\Pi_1)(I-\Pi_2)\}$.
\subsection{Noncontextuality}\label{noncontextualitysection}
It is always possible to build an ontological model reproducing the predictions of any operational theory, while respecting the coarse-graining relations.\footnote{Note that we will always assume coarse-graining relations are respected in any ontological model. The exception is (some of) the discussion in Section \ref{coarsegrainings} and Appendix \ref{ontmodelwocoarsegraining} where we consider the alternative possibility.} A trivial example of such an ontological model is one where ontic states $\lambda$ are identified with the preparation procedures $P_{[s|S]}$ (where $s\in V_S$ and $S\in \mathscr{S}$) and we have $\mu(\lambda,s|S)\equiv\delta_{\lambda,\lambda_{[s|S]}}p(s|S)$, where ontic state $\lambda_{[s|S]}$ is the one deterministically sampled by the preparation procedure $P_{[s|S]}$. Further, the response functions are identified with operational probabilities as $\xi(m|M,\lambda_{[s|S]})\equiv p(m|M,S,s)$. Then we have $\sum_{\lambda\in\Lambda}\xi(m|M,\lambda)\mu(\lambda,s|S)=\xi(m|M,\lambda_{[s|S]})p(s|S)=p(m,s|M,S)$. Also, coarse-graining relations of the type $[\tilde{m}|\tilde{M}]\equiv \sum_{m}p(\tilde{m}|m)[m|M]$ and $[\tilde{s}|\tilde{S}]\equiv \sum_{s}p(\tilde{s}|s)[s|S]$ that are respected in the operational description are also respected in this ontological description: that is, we have  $\forall\lambda\in\Lambda: \xi(\tilde{m}|\tilde{M},\lambda)\equiv \sum_{m}p(\tilde{m}|m)\xi(m|M,\lambda)$ and $\forall\lambda\in\Lambda: \mu(\lambda,\tilde{s}|\tilde{S})\equiv \sum_{s}p(\tilde{s}|s)\mu(\lambda,s|S)$. 

Hence, it is only when additional assumptions are imposed on an ontological model that deciding its existence becomes a nontrivial problem. Such additional assumptions must, of course, play an explanatory role to be worth investigating. The assumption we are interested in is {\em noncontextuality}, applied to both preparation and measurement procedures. Motivated by the methodological principle of the identity of indiscernables \cite{Spe05}, noncontextuality is an inference from the operational description to the ontological description of an experiment. It posits that the equivalence structure in the operational description is preserved in the ontological description, i.e., the reason one cannot distinguish two operationally equivalent representations of procedures based on their operational statistics is that there is, ontologically, no difference in their representations. We now formally define the notion of noncontextuality in its generalized form due to Spekkens \cite{Spe05}.

Mathematically, the assumption of measurement noncontextuality entails that
\begin{align}\label{ncmmt}
&[m|M]\simeq[m'|M']\nonumber\\
&\Rightarrow \xi(m|M,\lambda)=\xi(m'|M',\lambda), \forall\lambda\in\Lambda,
\end{align}
while the assumption of preparation noncontextuality entails that 
\begin{eqnarray}\label{ncprep}
&&[s|S]\simeq[s'|S']\Rightarrow\mu(\lambda,s|S)=\mu(\lambda,s'|S')\quad\forall\lambda\in\Lambda,\nonumber\\
&&[\top|S_{\top}]\simeq [\top|S'_{\top}]\Rightarrow\mu(\lambda|S)=\mu(\lambda|S')\quad\forall\lambda\in\Lambda.
\end{eqnarray}
Here we denote $\mu(\lambda|S)\equiv\sum_{s\in V_S}\mu(\lambda,s|S)$, etc., for simplicity of notation, rather than use the notation $\mu(\lambda,\top|S_{\top})$, etc., for these coarse-grained probability distributions. Note that since coarse-grainings are respected in any ontological model we consider, we indeed have that $\mu(\lambda,\top|S_{\top})\equiv \sum_{s\in V_S}\mu(\lambda,s|S)$.

These are the assumptions of noncontextuality -- termed {\em universal noncontextuality} -- that form the basis of our approach to noise-robust noncontextuality inequalities \cite{KunjSpek, exptlpaper, KrishnaSpekkensWolfe, SchmidSpek, KunjSpek17,algorithmic,cciosatalk}. Note that the traditional notion of KS-noncontextuality entails, besides measurement noncontextuality above, the assumption of {\em outcome-determinism}, i.e., for any measurement event $[m|M]$, $\xi(m|M,\lambda)\in\{0,1\}$ for all $\lambda\in\Lambda$.

It is important to note that in order for our notion of operational equivalence to be experimentally testable, we need that each of sets $\mathbb{M}$ and $\mathbb{S}$ includes a tomographically complete set of measurements and preparations, respectively. That is, the prepare-and-measure experiment testing contextuality can probe a tomographically complete set of preparations and measurements. Of course, the set of all possible measurements in a theory ($\mathscr{M}$) is (by definition) tomographically complete for any preparation in the theory and, similarly, the set of all possible preparations ($\mathscr{S}$) in a theory is tomographically complete for any measurement in the theory. However, there may exist smaller (finite) sets of preparations and measurements in the theory that are tomographically complete and in that case we require that $\mathbb{S}$ and $\mathbb{M}$ include such tomographically complete sets, even if they don't include all possible preparations and measurements in the theory. For example, when the operational theory is quantum theory for a qubit, the three spin measurements $\{\sigma_x,\sigma_y,\sigma_z\}$ are tomographically complete for any qubit preparation, so we require that $\mathbb{M}$ includes these three measurements even if it doesn't include every other possible measurement on a qubit. While the requirement that $\mathbb{S}$ and $\mathbb{M}$ include tomographically complete sets doesn't directly reflect in our {\em theoretical} derivation of the noise-robust noncontextuality inequalities later, it is crucial for {\em experimentally} verifying the operational equivalences (cf.~Eqs.~\eqref{opequiv1mmt},\eqref{opequiv1prep},\eqref{opequiv2prep}) we need to even invoke the assumption of noncontextuality (cf.~Eqs.~\eqref{ncmmt},\eqref{ncprep}). Further, this assumption on $\mathbb{M}$ and $\mathbb{S}$ has so far been necessary to be able to implement an actual noise-robust contextuality experiment \cite{exptlpaper}, besides the requirement that the operational theory be convex, i.e., probabilistic mixtures of procedures in the theory (whether preparations or measurements) are also valid procedures in the theory. We refer the reader to Refs.~\cite{exptlpaper,robtalk,gpttompaper} for a discussion of what tomographic completeness entails for (convex) operational theories formalized as general probabilistic theories (GPTs). Although we will not discuss it in this paper, see Ref.~\cite{wotomocomp} for some recent work towards relaxing the tomographic completeness requirement for the set of measurement settings.
\subsection{An example of Spekkens contextuality: the fair coin flip inequality}\label{spekexample} 
We recap here an example of Spekkens contextuality that has been experimentally demonstrated \cite{exptlpaper} to give the reader a flavour of the general approach we are going to adopt in the rest of this paper with regard to Kochen-Specker type scenarios. We call the inequality tested in Ref.~\cite{exptlpaper} the ``fair coin flip" inequality.

Consider a prepare-and-measure scenario with three source settings, denoted $\mathbb{S}\equiv \{S_1,S_2,S_3\}$, such that $V_{S_i}\equiv\{0,1\}$ and we have $p(s_i=0|S_i)=p(s_i=1|S_i)=1/2$ for all $i\in\{1,2,3\}$. Each $S_i$ thus corresponds to the ensemble of preparation procedures $\{(p(s_i|S_i),P_{[s_i|S_i]})\}_{s_i\in V_{S_i}}$ and we have the following operational equivalence among the source settings after coarse-graining:
\begin{equation}\label{fcfsourceeq}
[\top|S_{1_{\top}}]\simeq [\top|S_{2_{\top}}]\simeq [\top|S_{3_{\top}}].
\end{equation}
There are four measurement settings in this scenario, denoted $\mathbb{M}\equiv\{M_1,M_2,M_3, M_{\rm fcf}\}$, such that $V_{M_i}\in\{0,1\}$ for all $i\in\{1,2,3,{\rm fcf}\}$. The measurement setting $M_{\rm fcf}$ is a fair coin flip, i.e., it is insensitive to the preparation procedure preceding it and yields the outcome $m_{\rm fcf}=0$ or $1$ with equal probability for any preparation procedure $P_{[s|S]}$, i.e., $p(m_{\rm fcf}=0|M_{\rm fcf},S,s)=p(m_{\rm fcf}=1|M_{\rm fcf},S,s)=1/2$ for all $[s|S]$. 

We also {\em define} a measurement procedure $M_{\rm mix}$ as a classical post-processing of $M_1,M_2,M_3$, i.e., its measurement events $\{[m_{\rm mix}|M_{\rm mix}]\}_{m_{\rm mix}=0}^1$ are defined by the classical post-processing relation

\begin{equation}
[m_{\rm mix}|M_{\rm mix}]\equiv \sum_{i=1}^3 p(i)\sum_{m_i=0}^1p(m_{\rm mix}|m_i)[m_i|M_i],
\end{equation}
which symbolically denotes the following post-processing of measurements $M_1,M_2,M_3$: consider a uniform probability distribution $\left\{p(i)=\frac{1}{3}\right\}_{i=1}^3$ over the measurement settings $\{M_i\}_{i=1}^3$ and relabel the respective measurement outcomes, i.e., $\{m_i\in\{0,1\}\}_{i=1}^3$, to a measurement outcome $m_{\rm mix}\in\{0,1\}$ according to the probability distributions $$\{\{p(m_{\rm mix}|m_i)=\delta_{m_{\rm mix},m_i}\}_{m_{\rm mix}\in\{0,1\}}\}_{i=1}^3;$$ coarse-graining over $m_i$ and $i$ then yields the effective measurement setting $M_{\rm mix}$ with outcomes labelled by $m_{\rm mix}\in\{0,1\}$. In contrast to the kinds of coarse-graining (over measurement outcomes) that appear in KS-noncontextuality (which we discussed in Section \ref{coarsegrainings}), the (probabilistic) coarse-graining here is over the measurement settings themselves while retaining the outcome labels.\footnote{We did not discuss these more general types of classical post-processing in Section \ref{coarsegrainings} because they are not relevant to the treatment of Kochen-Specker type scenarios in the Spekkens framework. The example we present here is from Ref.~\cite{exptlpaper}, which is not of Kochen-Specker type. The general principle underlying the representation of such classical post-processings is, however, the same: they should be respected in the operational as well as the ontological description.} We require that this coarse-graining relation be respected in the operational as well as the ontological description. In the operational description, this coarse-graining is represented by 
\begin{align}
&\forall [s|S], b\in\{0,1\}:\nonumber\\
&p(m_{\rm mix}=b,s|M_{\rm mix},S)\equiv\frac{1}{3}\sum_{i=1}^3p(m_i=b,s|M_i,S).
\end{align}
 
We require the following operational equivalence between measurement events of $M_{\rm mix}$ and $M_{\rm fcf}$ with respect to which we invoke the assumption of measurement noncontextuality:
\begin{equation}\label{fcfmmteq}
\forall b\in\{0,1\}: [m_{\rm mix}=b|M_{\rm mix}]\simeq [m_{\rm fcf}=b|M_{\rm fcf}]
\end{equation}

If we then look at an operational quantity quantifying source-measurement correlations, namely,
\begin{equation}
{\rm Corr}_{\rm fcf}\equiv \sum_{i=1}^3\frac{1}{3}\sum_{m_i,s_i}\delta_{m_i,s_i}p(m_i,s_i|M_i,S_i),
\end{equation}
then the assumption of preparation noncontextuality applied to operational equivalence in Eq.~\eqref{fcfsourceeq} (so that $\mu(\lambda|S_1)=\mu(\lambda|S_2)=\mu(\lambda|S_3)$ for all $\lambda\in\Lambda$) and the assumption of measurement noncontextuality applied to the operational equivalence in Eq.~\eqref{fcfmmteq} (so that $\frac{1}{3}\xi(0|M_1,\lambda)+\frac{1}{3}\xi(0|M_2,\lambda)+\frac{1}{3}\xi(0|M_3,\lambda)=\frac{1}{2}$ for all $\lambda\in\Lambda$) lead to the following constraint:
\begin{equation}\label{fcfineq}
{\rm Corr}_{\rm fcf}\leq \frac{5}{6}.
\end{equation}
To see how this is obtained, note that
\begin{align}
&\sum_{i=1}^3\frac{1}{3}\sum_{m_i,s_i}\delta_{m_i,s_i}p(m_i,s_i|M_i,S_i)\nonumber\\
&=\sum_{i=1}^3\frac{1}{3}\sum_{m_i,s_i}\delta_{m_i,s_i}\sum_{\lambda\in\Lambda}\xi(m_i|M_i,\lambda)\mu(\lambda,s_i|S_i)\nonumber\\
&\leq \sum_{i=1}^3\frac{1}{3}\sum_{\lambda\in\Lambda}\max_{m_i}\xi(m_i|M_i,\lambda)\sum_{m_i,s_i}\delta_{m_i,s_i}\mu(\lambda,s_i|S_i)\nonumber\\
&=\sum_{i=1}^3\frac{1}{3}\sum_{\lambda\in\Lambda}\zeta(M_i,\lambda)\sum_{s_i}\mu(\lambda,s_i|S_i)\nonumber\\
&=\sum_{\lambda\in\Lambda}\sum_{i=1}^3\frac{1}{3}\zeta(M_i,\lambda)\nu(\lambda),
\end{align}
where we have that $\zeta(M_i,\lambda)\equiv\max_{m_i}\xi(m_i|M_i,\lambda)$ and that $\nu(\lambda)\equiv\mu(\lambda|S_i)$ for all $i\in\{1,2,3\}$. This allows us to put the upper bound 
\begin{align}
{\rm Corr}_{\rm fcf}&\leq\max_{\lambda\in\Lambda}\frac{1}{3}\sum_{i=1}^3\zeta(M_i,\lambda),
\end{align}
which, subject to the constraint (from measurement noncontextuality) that $\frac{1}{3}\xi(0|M_1,\lambda)+\frac{1}{3}\xi(0|M_2,\lambda)+\frac{1}{3}\xi(0|M_3,\lambda)=\frac{1}{2}$, yields Eq.~\eqref{fcfineq}.\footnote{The reader may look at Appendix B.1 of  Ref.~\cite{exptlpaper} to convince themselves that the maximum is achieved for an assignment of response functions of the type $\xi(0|M_1,\lambda)=1$, $\xi(0|M_2,\lambda)=\frac{1}{2}$ and $\xi(0|M_1,\lambda)=0$ for some $\lambda$.}
It turns out that in quantum theory the sources and measurements required for this scenario can be realized on a qubit and they can, in principle, achieve the value ${\rm Corr}=1$. This can be achieved by taking the three preparations to be the trine preparations on an equatorial plane (say, the Z-X plane) of the Bloch sphere and the measurements $\{M_i\}_{i=1}^3$ to be the trine measurements, i.e.,
\begin{align}
&\rho_{[s_i=0|S_i]}\equiv\frac{1}{2}(\mathbb{I}+
\vec{\sigma}.\vec{n}_i)\equiv\Pi^0_i,\nonumber\\
&\rho_{[s_i=1|S_i]}\equiv\frac{1}{2}(\mathbb{I}-
\vec{\sigma}.\vec{n}_i)\equiv\Pi^1_i,\nonumber\\
&E_{[m_i=0|M_i]}\equiv\Pi^0_i,\nonumber\\
&E_{[m_i=1|M_i]}\equiv\Pi^1_i,
\end{align}
where $\vec{n}_1\equiv(0,0,1)$, $\vec{n}_2\equiv(\frac{\sqrt{3}}{2},0,-\frac{1}{2})$, $\vec{n}_3\equiv(-\frac{\sqrt{3}}{2},0,-\frac{1}{2})$, and $\vec{\sigma}\equiv(\sigma_x,\sigma_y,\sigma_z)$ denotes the three Pauli matrices $\sigma_x=\begin{pmatrix}
0 & 1\\ 1 & 0
\end{pmatrix}$, $\sigma_y=\begin{pmatrix}
0 & -i\\ i & 0
\end{pmatrix}$, and $\sigma_z=\begin{pmatrix}
1 & 0\\ 0 & -1
\end{pmatrix}$. The operational equivalences are then easy to verify:
\begin{align}
\rho_{[\top|S_{i_{\top}}]}&=\frac{\mathbb{I}}{2},\quad\forall i\in\{1,2,3\},\nonumber\\
\frac{1}{3}\sum_{i=1}^3\Pi^0_i&=\frac{\mathbb{I}}{2}.
\end{align}
The quantity ${\rm Corr}_{\rm fcf}=1$ from this quantum realization.
The experimental violation of the noise-robust noncontextuality inequality, Eq.~\eqref{fcfineq}, was demonstrated in Ref.~\cite{exptlpaper}, where more details may be found. Note that the fair coin flip inequality, Eq.\eqref{fcfineq}, is not inspired by the kinds of operational equivalences that are relevant in a proof of the Kochen-Specker theorem, but employs other kinds of operational equivalences allowed in the Spekkens framework \cite{Spe05}, i.e., the operational equivalences in Eqs.~\eqref{fcfsourceeq} and \eqref{fcfmmteq} do not arise from the same measurement outcome being shared by different measurements.

Our goal in the present paper is to provide a framework for noise-robust noncontextuality inequalities obtained from statistical proofs of the KS theorem, in particular those that are covered by the CSW framework \cite{CSW}, so that such inequalities can be put to an experimental test along the lines of Ref.~\cite{exptlpaper} within the Spekkens framework. Hence, the operational equivalences between measurement events that will be of interest to us in this paper are precisely those which allow for a proof of the KS theorem, i.e., those which correspond to the same measurement outcome (e.g., a projector) being shared by different measurements (e.g., projective measurements).

\subsection{Connection to Bell scenarios}
As further motivation to study the questions we are posing, note that one can also view the general prepare-and-measure scenario we are considering in this paper (Fig.~\ref{prepmsr}) as arising on one wing of a two-party Bell experiment: that is, given two parties -- Alice and Bob -- sharing an entangled state and performing local measurements in a Bell experiment, one can view each choice of measurement setting on Alice's side as preparing an ensemble of states on Bob's side; on account of no-signalling, the reduced state on Bob's side will be the same regardless of Alice's choice of measurement setting, i.e., all the ensembles corresponding to Alice's measurement settings (hence, Bob's source settings) will be operationally equivalent. 

For example, consider a Bell experiment where Alice has two choices of measurement settings, $M^A_x\equiv\sigma_x$ or $M^A_z\equiv\sigma_z$, and she shares a Bell state with Bob: $|\psi\rangle=\frac{1}{\sqrt{2}}\left(|00\rangle+|11\rangle\right)$. Bob has access to some set of measurement settings $\mathbb{M}^B\equiv\{M^B_j\}_j$ on his system. When Alice measures $M^A_x$, she prepares the ensemble of states $S^A_x\equiv\{(1/2,\rho_{[s^A_x=0|S^A_x]}\equiv|+\rangle\langle+|),(1/2,\rho_{[s^A_x=1|S^A_x]}\equiv|-\rangle\langle-|)\}$ on Bob's side and when she measures $M^A_z$ she prepares the ensemble of states $S^A_z\equiv\{(1/2,\rho_{[s^A_z=0|S^A_z]}\equiv|0\rangle\langle0|),(1/2,\rho_{[s^A_z=1|S^A_z]}\equiv|1\rangle\langle1|)\}$. These ensembles are operationally equivalent, yielding the maximally mixed state on coarse-graining, i.e.,
\begin{equation}
\frac{1}{2}|0\rangle\langle 0|+\frac{1}{2}|1\rangle\langle 1|=\frac{1}{2}|+\rangle\langle +|+\frac{1}{2}|-\rangle\langle -|=\frac{\mathbb{I}}{2}.
\end{equation}

 The quantity of interest in a Bell experiment $p(m^A_i,m^B_j|M^A_i,M^B_j)$ ($i\in\{x,z\}$) is then formally the same as the quantity $p(s^A_i,m^B_j|S^A_i,M^B_j)$ that we are interested in our prepare-and-measure scenario. In the ontological model describing the effective prepare-and-measure experiment on Bob's system, we have the following: 
 \begin{align}
 &p(s^A_i,m^B_j|S^A_i,M^B_j)\nonumber\\
 &=\sum_{\lambda}{\rm Pr}(m^B_j|M^B_j,\lambda){\rm Pr}(\lambda,s^A_i|S^A_i)\nonumber\\
 &=\sum_{\lambda}{\rm Pr}(m^B_j|M^B_j,\lambda){\rm Pr}(s^A_i|S^A_i,\lambda){\rm Pr}(\lambda|S^A_i).
 \end{align} 
 Assuming preparation noncontextuality relative to the operational equivalence $[\top|S^A_x]\simeq[\top|S^A_z]$, we have ${\rm Pr}(\lambda|S^A_x)={\rm Pr}(\lambda|S^A_z)\equiv{\rm Pr}(\lambda)$, so that 
  \begin{align}
 &p(s^A_i,m^B_j|S^A_i,M^B_j)\nonumber\\
 &=\sum_{\lambda}{\rm Pr}(s^A_i|S^A_i,\lambda){\rm Pr}(m^B_j|M^B_j,\lambda){\rm Pr}(\lambda),
 \end{align} 
 which formally resembles the expression for local causality when applied to the corresponding two-party Bell experiment:
  \begin{align}
 &p(m^A_i,m^B_j|M^A_i,M^B_j)\nonumber\\
 &=\sum_{\lambda}{\rm Pr}(m^A_i|M^A_i,\lambda){\rm Pr}(m^B_j|M^B_j,\lambda){\rm Pr}(\lambda).
 \end{align}

If no other assumption of noncontextuality is invoked besides the one applied to the operational equivalence of source settings on Bob's system, then the constraints on $p(s^A_i,m^B_j|S^A_i,M^B_j)$ will be the same as the constraints on $p(m^A_i,m^B_j|M^A_i,M^B_j)$ from Bell inequalities.

Note, however, that the response functions ${\rm Pr}(m^B_j|M^B_j,\lambda)$ and ${\rm Pr}(m^A_j|M^A_j,\lambda)$ can be  completely arbitrary in a locally causal ontological model for the Bell experiment and the same applies to the distributions ${\rm Pr}(s^A_i|S^A_i,\lambda)$ and ${\rm Pr}(m^B_j|M^B_j,\lambda)$ in a preparation noncontextual model of the corresponding prepare-and-measure scenario on Bob's side. We will be interested in imposing {\em additional} constraints on the response functions ${\rm Pr}(m^B_j|M^B_j,\lambda)$ of the prepare-and-measure scenario (on Bob's side) that follow from the assumption of measurement noncontextuality applied to operational equivalences between measurement events on Bob's side. In particular, we are interested in those operational equivalences between measurement events that are required by any statistical proof of the Kochen-Specker theorem \cite{KCBS,KunjSpek17}. We develop this approach more carefully in the following sections.

\section{Hypergraph approach to Kochen-Specker scenarios in the Spekkens framework}\label{hypergraphsection}
Having set up the framework needed to articulate the relevant notions in Section \ref{spekframework}, we now proceed to consider Kochen-Specker type experimental scenarios in this framework. To do this, we will use the language of hypergraphs and their subgraphs to represent the operational equivalences between measurement events that are required in a Kochen-Specker argument as well as the operational equivalences between source settings that we will invoke in our generalization. The (hyper)graph-theoretic ingredients of our approach will represent those aspects of the general framework of Section \ref{spekframework} that are necessary to go from the CSW framework for KS-contextuality to a hypergraph framework for Spekkens contextuality applied to Kochen-Specker type experimental scenarios.

Our presentation will be a hybrid one, discussing features of the CSW framework \cite{CSW10, CSW} in the notation of the AFLS framework \cite{AFLS}, but extending both in ways appropriate for the purpose of this paper. Our goal is to demonstrate how the graph-theoretic invariants of CSW \cite{CSW} can be repurposed towards obtaining noise-robust noncontextuality inequalities.

 We do this in two parts: first, we define a representation of measurement events in the manner of Refs.~\cite{CSW,AFLS}, and then we define a representation of source events in the spirit of Ref.~\cite{KunjSpek}.

\subsection{Measurements}
The basic object for representing measurements is a hypergraph, $\Gamma$, with a finite set of vertices $V(\Gamma)$ such that each vertex $v\in V(\Gamma)$ denotes a measurement outcome, and a set of hyperedges $E(\Gamma)$ such that each hyperedge $e\in E(\Gamma)$ is a subset of $V(\Gamma)$ and denotes a measurement consisting of outcomes in $e$. Here, $E(\Gamma)\subseteq 2^{V(\Gamma)}$ and $\bigcup_{e\in E(\Gamma)}e=V(\Gamma)$. Such a hypergraph satisfies the definition of a {\em contextuality scenario} \`a la AFLS \cite{AFLS}. We will further assume, unless specified otherwise, that the hypergraph is {\em simple}: that is, for all $e_1, e_2 \in E(\Gamma)$, $e_1\subseteq e_2\Rightarrow e_1=e_2$, or that no hyperedge is a strict subset of another. Such hypergraphs are also called {\em Sperner families} \cite{sperner}. Two measurement events are said to be {\em (mutually) exclusive} if the vertices denoting them appear in a common hyperedge, i.e., if they can be realized as outcomes of a single measurement setting.

The structure of a {\em contextuality scenario} $\Gamma$ represents the operational equivalences between measurement events that are of interest in a Kochen-Specker argument. We emphasize here that we take the operational theory to be fundamental and the contextuality scenario for a particular Kochen-Specker argument to be derived from (and as a graphical representation of) the operational equivalences in the operational theory (cf.~Section \ref{spekframework}). In particular, depending on the operational equivalences that an operational theory can exhibit (by virtue of (in)compatibility relations between measurements), it may or may not allow some contextuality scenario to be realized by measurement events in the theory. The fact that a given vertex, say $v\in V(\Gamma)$, appears in multiple hyperedges, say $E'\equiv\{e\in E(\Gamma)|v\in e\}$,  means that the measurement events corresponding to this vertex, i.e., $\{[v|e]\}_{e\in E'}$, are operationally equivalent, and the equivalence class of these measurement events is denoted by the vertex $v$ itself. In the case of quantum theory, for example, $v$ can represent a positive operator that appears in different positive operator-valued measures (POVMs) represented by the hyperedges.

A {\em probabilistic model} on $\Gamma$ is an assignment of probabilities to the vertices $v\in V(\Gamma)$ such that $p(v)\geq0$ for all $v\in V(\Gamma)$ and $\sum_{v\in e}p(v)=1$ for all $e\in E(\Gamma)$. As we have noted, every vertex $v$ represents an {\em equivalence class} of measurement events, denoted $[\mathcal{m}|\mathcal{M}]$, and every hyperedge $e$ represents an {\em equivalence class} of measurement procedures, denoted $\mathcal{M}$.\footnote{Note that two measurement procedures with measurement settings $M$ and $M'$ are operationally equivalent if every measurement event of one is operationally equivalent to a distinct measurement event of the other. That is, there is a bijective correspondence (of operational equivalence) between the two sets of measurement events. In quantum theory, for example, a given POVM (which is what a hyperedge would represent), say $\{E_k\}_k$, can be implemented in many possible ways, each such measurement procedure corresponding to different quantum instrument. Mathematically, these different procedures can be represented by different sets of operators $\{O_k\}_k$ such that $E_k=O_k^{\dagger}O_k$ for all $k$ and $\sum_kO_k^{\dagger}O_k=\mathbb{I}$.} 
The fact that each $v$ represents an equivalence class of measurement events means that 
\begin{enumerate}
	\item any probabilistic model $p$ on $\Gamma$, realized by {\em operational} probabilities for a given source event -- that is, where for all $v\in V(\Gamma)$ and a given $[s|S]$, $p(v)\equiv p(v|S,s)\equiv p(\mathcal{m}|\mathcal{M},S,s)$ -- is consistent with the operational equivalences represented by $\Gamma$, and
	\item any probabilistic model on $\Gamma$, realized by {\em ontological} probabilities for a given ontic state -- that is, where for all $v\in V(\Gamma)$ and a given ontic state $\lambda$, $p(v)\equiv p(v|\lambda)\equiv\xi(\mathcal{m}|\mathcal{M},\lambda)$ -- respects (by definition) the assumption of measurement noncontextuality with respect to the presumed operational equivalences between measurement events.
\end{enumerate}

We will therefore often write $p(\mathcal{m},s|\mathcal{M},S)$ as $p(v,s|S)$ and $p(\mathcal{m}|\mathcal{M},S,s)$ as $p(v|S,s)$, where $[s|S]$ is a source event. Similarly, we will also write $\xi(\mathcal{m}|\mathcal{M},\lambda)$ as $p(v|\lambda)$, where $\lambda$ is an ontic state.

{\em Orthogonality graph of $\Gamma$, $O(\Gamma)$:} Given the hypergraph $\Gamma$, we construct its orthogonality graph $O(\Gamma)$: that is, the vertices of $O(\Gamma)$ are given by $V(O(\Gamma))\equiv V(\Gamma)$, and the edges of $O(\Gamma)$ are given by  $E(O(\Gamma))\equiv\{\{v,v'\}|v,v'\in e \textrm{ for some }e\in E(\Gamma)\}$. Each edge of $O(\Gamma)$ denotes the exclusivity of the two measurement events it connects, i.e., the fact that they can occur as outcomes of a single measurement.

For any Bell-KS inequality constraining correlations between measurement events from $O(\Gamma)$ (when all measurements are implemented on a given source event), we construct a subgraph $G$ of $O(\Gamma)$ such that the vertices of $G$, i.e., $V(G)$, correspond to measurement events that appear in the inequality with nonzero coefficients, and two vertices share an edge in $G$ if and only if they share an edge in $O(\Gamma)$. More explicitly, consider a Bell-KS expression 
\begin{equation}
R([s|S])\equiv\sum_{v\in V(G)}w_v p(v|S,s),
\end{equation}
where $w_v>0$ for all $v\in V(G)$. A Bell-KS inequality imposes a constraint of the form $R([s|S])\leq R_{\rm KS}$, where $R_{\rm KS}$ is the upper bound on the expression in any operational theory that admits a KS-noncontextual ontological model. Often, but not always, these inequalities are simply of the form where $w_v=1$ for all $v\in V(G)$. In keeping with the CSW notation \cite{CSW}, we will denote the general situation by a {\em weighted} graph $(G,w)$, where $w$ is a function that maps vertices $v\in V(G)$ to weights $w_v>0$. See Figures \ref{fig2} and \ref{fig3} for an example from the Klyachko-Can-Binicio\u{g}lu-Shumovsky (KCBS) scenario \cite{KCBS,CSW}.
\begin{figure}
	\centering
	\includegraphics[scale=0.3]{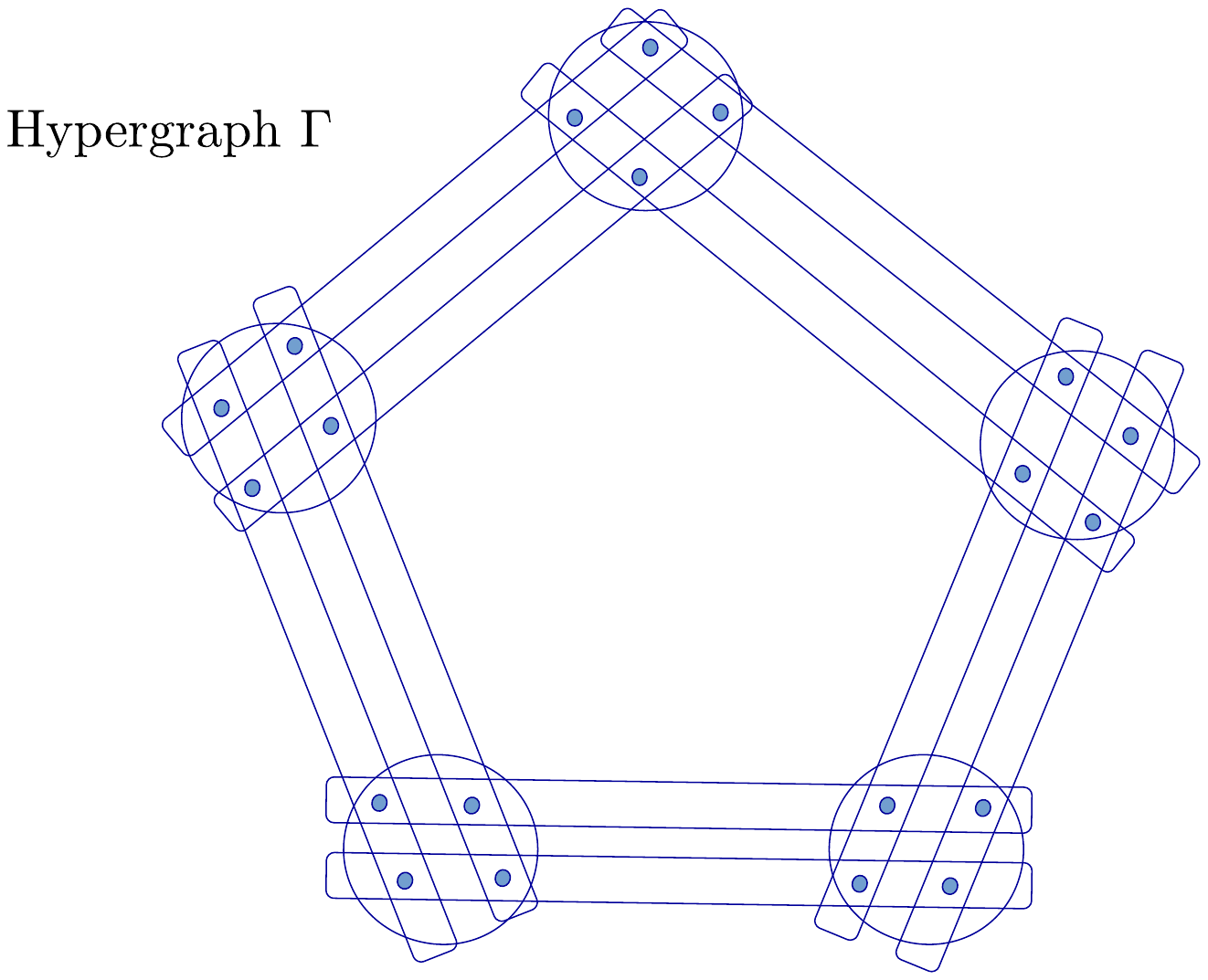}
	\caption{The KCBS scenario with 4-outcome joint measurements, visualized as a hypergraph $\Gamma$ \cite{KCBS, CSW, KunjSpek17}.}
	\label{fig2}
\end{figure}

\begin{figure}
	\centering
	\includegraphics[scale=0.3]{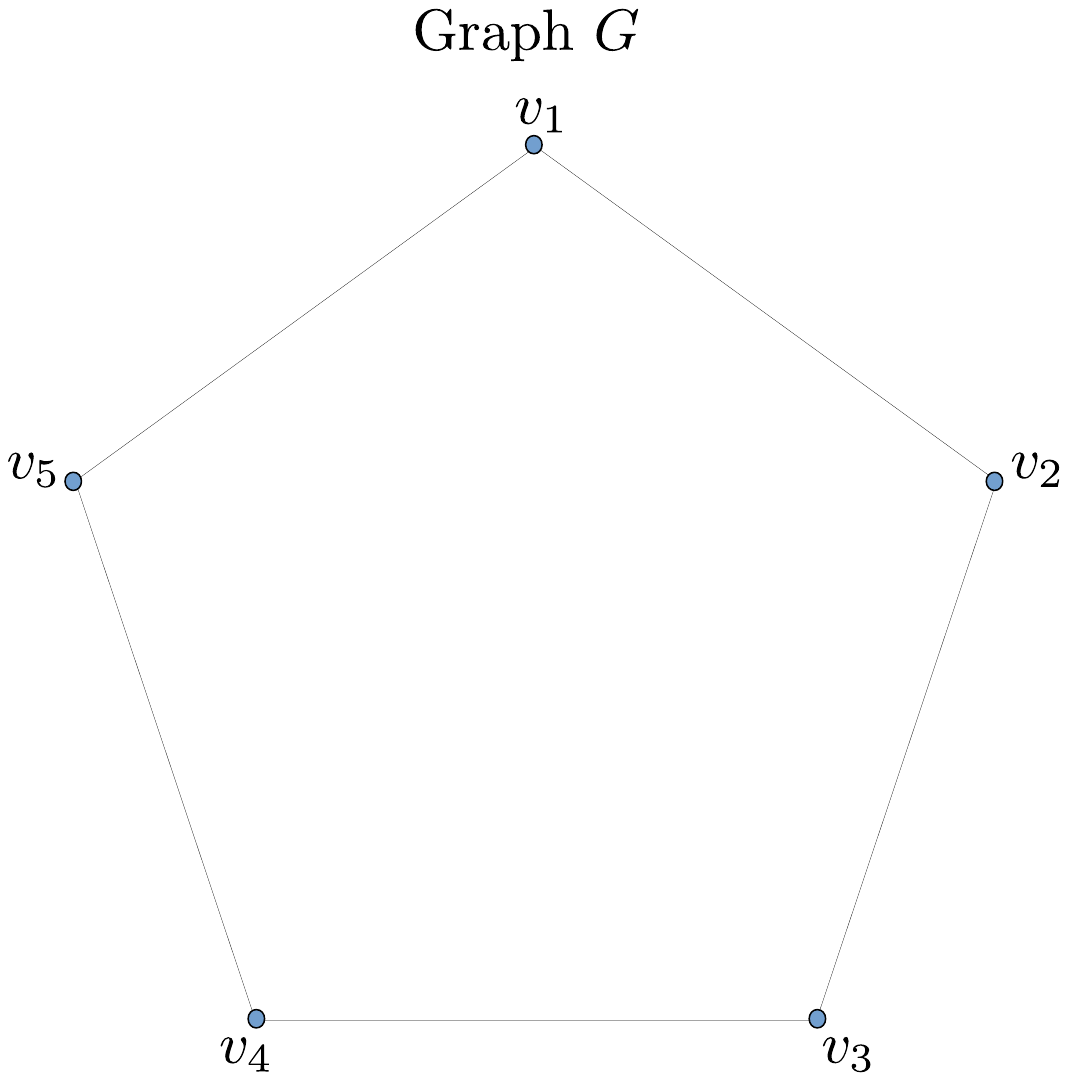}
	\caption{A subgraph of KCBS hypergraph $\Gamma$, representing orthogonality relations of the events of interest in the KCBS inequality \cite{KCBS,CSW}.}
	\label{fig3}
\end{figure}

Below, we make some remarks clarifying the scope of the framework described above before we move to the case of sources.

\subsubsection{Classification of probabilistic models} 
We classify the probabilistic models on a hypergraph $\Gamma$ as follows:
\begin{itemize}
	\item KS-noncontextual probabilistic models, $\mathcal{C}(\Gamma)$: a probabilistic model which is a convex combination of deterministic assignments $p:V(\Gamma)\rightarrow\{0,1\}$, where $\sum_{v\in e}p(v)=1$ for all $e\in E(\Gamma)$. In Ref.~\cite{AFLS}, this is referred to as a ``classical model".\footnote{We use a different term because we are advocating a revision of the notion of classicality from KS-noncontextuality to generalized  noncontextuality \`a la Spekkens.}
	
	Note that we call $\Gamma$ {\em KS-colourable} if $\mathcal{C}(\Gamma)\neq\varnothing$ and we call it {\em KS-uncolourable} if $\mathcal{C}(\Gamma)=\varnothing$. Our terminology here is inspired by the traditional usage of the term ``Kochen-Specker colouring" to refer to an assignment of two colours to vectors satisfying some orthogonality relations under the colouring constraints of the KS theorem \cite{SEPKS}.
	
	\item Consistent exclusivity satisfying probabilistic models, $\mathcal{CE}^1(\Gamma)$: a probabilistic model on $\Gamma$,
	$p:V(\Gamma)\rightarrow [0,1]$, such that (in addition to satisfying the definition of a {\em probabilistic model}), $\sum_{v\in c}p(v)\leq 1$ for all cliques $c$ in the orthogonality graph $O(\Gamma)$. This is the same as the set of $E1$ probabilistic models of Ref.~\cite{CSW}.
	
	Note that a {\em clique} in the orthogonality graph $O(\Gamma)$ is a set of vertices that are pairwise exclusive (i.e., every vertex in this set shares an edge with every other vertex). 
	
	\item General probabilistic models, $\mathcal{G}(\Gamma)$: Any $p$ that satisfies the definition of a probabilistic model is a general probabilistic model, i.e., it can arise from measurements in some general probabilistic theory \cite{hardy5axioms} that isn't necessarily quantum.
	
	The set of all probabilistic models $\mathcal{G}(\Gamma)$ (for any $\Gamma$) forms a polytope since it is defined by just the positivity and normalization constraints on the probabilities. The extremal points (or vertices) of this polytope fall into two categories that will interest us: deterministic and indeterministic. The deterministic extremal points are the $p:V(\Gamma)\rightarrow\{0,1\}$ such that $\sum_{v\in e}p(v)=1$ for all $e\in E(\Gamma)$ and we denote the set of these points by $\mathcal{G}(\Gamma)|_{\rm det}$. The indeterministic extremal points are the $p\in\mathcal{G}(\Gamma)$ which are not deterministic and which, furthermore, cannot be expressed as a convex mixture of other points in $\mathcal{G}(\Gamma)$. We denote the set of indeterministic extremal points by $\mathcal{G}(\Gamma)|_{\rm ind}$. Clearly, $\mathcal{G}(\Gamma)|_{\rm det}\subsetneq \mathcal{C}(\Gamma)$ and $\mathcal{G}(\Gamma)|_{\rm ind}\subseteq\mathcal{G}(\Gamma)\backslash\mathcal{C}(\Gamma)$.	
\end{itemize}
Overall, we have
\begin{equation}
\mathcal{C}(\Gamma)\subseteq\mathcal{CE}^1(\Gamma)\subseteq\mathcal{G}(\Gamma)
\end{equation}
for any hypergraph $\Gamma$.

\subsubsection{Distinguishing two consequences of Specker's principle: Structural Specker's principle vs.~Statistical Specker's principle}\label{Speck1}
The CSW framework \cite{CSW} restricts the scope of probabilistic models on a hypergraph to those satisfying consistent exclusivity (the $E1$ probabilistic models), motivated by what is sometimes called {\em Specker's principle} \cite{speckerprinciple}: that is, 
\begin{quote}
``if you have several questions and you can answer any two of them, then you can also answer all of them"	
\end{quote}
 If by ``questions" we understand {\em measurement settings}, then the principle says that a set of pairwise jointly implementable measurement settings is itself jointly implementable. Note that when we say a set of {\em measurement settings} is ``jointly implementable", ``jointly measurable", or ``compatible", we mean that there exists another choice of a single measurement setting in the theory such that this measurement setting can reproduce the statistics of all the measurement settings in the set by coarse-graining.\footnote{The reader may recall from Section \ref{compatibilitysection} the general definition of compatibility. Also, see Ref.~\cite{notesonjm} for an overview of joint measurability in quantum theory.} As such, in its application to measurement settings, Specker's principle is a constraint on the measurements allowed in a physical theory that respects it, e.g., measurement settings that correspond to PVMs (projection valued measures) in quantum theory. This is, for example, the reading adopted in Ref.~\cite{almostquantumsharp}, where the failure of Specker's principle in any almost quantum theory was demonstrated. On the other hand, we will often also refer to the ``joint measurability" of a set of {\em measurement events}, by which we mean that this set of measurement events is a subset of the set of measurement outcomes for some choice of measurement setting. At the level of {\em measurement events},\footnote{Recall that a {\em measurement event} is a measurement outcome given a choice of measurement setting, e.g., a projector that appears in a particular PVM in quantum theory.} then, there are two distinct ways to read Specker's principle that one needs to keep in mind which we distinguish as {\em structural Specker's principle} vs.~{\em statistical Specker's principle}. We define these two readings below:
\begin{itemize}
	\item {\em Structural Specker's principle} imposes a {\em structural} constraint on a contextuality scenario $\Gamma$.
	
	This (strong) reading of Specker's principle applies to any set of measurement events, say $\frak{M}\subseteq V(\Gamma)$, where every pair of measurement events can arise as outcomes of a single measurement: that is, for each pair $\{v,v'\}\subseteq\frak{M}$, there exists some $e\in E(\Gamma)$ such that $\{v,v'\}\subseteq e$. The principle then states: 
	
	{\em Given a set $\frak{M}$ of pairwise jointly measurable measurement events in some contextuality scenario $\Gamma$, all the measurement events in $\frak{M}$ are jointly measurable, i.e., all the measurement events in the set can arise as outcomes of a single measurement: $\frak{M}\subseteq e$ for some $e\in E(\Gamma)$.}
	
	Alternatively, the constraint of structural Specker's principle can be restated as: 
	
	{\em Every clique in the orthogonality graph of $\Gamma$, $O(\Gamma)$, is a subset of some hyperedge in $\Gamma$.}
	
	Note that we haven't said anything directly about {\em probabilities} here: any $\Gamma$ satisfying the above property is said to satisfy structural Specker's principle. 
			
	\item {\em Statistical Specker's principle (or consistent exclusivity)} imposes a {\em statistical} constraint on probabilistic models on any contextuality scenario $\Gamma$ representing measurement events in an operational theory. 
	
	This (weak) reading of Specker's principle imposes an additional constraint on a probabilistic model $p\in\mathcal{G}(\Gamma)$ (thus defining $\mathcal{CE}^1(\Gamma)\subseteq\mathcal{G}(\Gamma)$), namely: 
	
	{\em Given a set $\frak{M}$ of pairwise jointly measurable measurement events, $p$ satisfies 
	$\sum_{v\in\frak{M}}p(v)\leq 1$.}
	
	This can also be expressed as:
	
	{\em A probabilistic model $p\in\mathcal{G}(\Gamma)$ is said to satisfy statistical Specker's principle if the sum of probabilities it assigns to the vertices of every clique in the orthogonality graph of $\Gamma$, $O(\Gamma)$, does not exceed 1, i.e., $\sum_{v\in c}p(v)\leq 1$ for all cliques $c$ in $O(\Gamma)$.}
	
	All probabilistic models that satisfy this constraint define the set of probabilistic models $\mathcal{CE}^1(\Gamma)$ (or $E1$) for any contextuality scenario $\Gamma$ {\em regardless} of whether $\Gamma$ satisfies structural Specker's principle. Clearly, $\mathcal{CE}^1(\Gamma)\subseteq \mathcal{G}(\Gamma)$.
	
	Any probabilistic model $p$ on $\Gamma$ such that $p\in\mathcal{CE}^1(\Gamma)$ is said to satisfy statistical Specker's principle or, equivalently, {\em consistent exclusivity} \cite{AFLS}.
\end{itemize}

Probabilistic models on any hypergraph $\Gamma$ which satisfies the (strong) structural Specker's principle obviously satisfy the (weak) statistical Specker's principle. This holds simply on account of the structure of such $\Gamma$: that is, for all $\Gamma$ satisfying structural Specker's principle, we have $\mathcal{CE}^1(\Gamma)=\mathcal{G}(\Gamma)$. To see this, note that every clique $c$ in $O(\Gamma)$ is a subset of some hyperedge in $\Gamma$, hence for every clique $c$, $\sum_{v\in c}p(v)\leq 1$ for all $p\in\mathcal{G}(\Gamma)$, i.e., $p\in\mathcal{CE}^1(\Gamma)$.\footnote{This partially answers the open Problem 7.2.3 of Ref.~\cite{AFLS}.} On the other hand, it remains an {\em open question} whether the converse is true:

{\em That is, given that $\mathcal{CE}^1(\Gamma)=\mathcal{G}(\Gamma)$ for some $\Gamma$, is it the case that $\Gamma$ must then necessarily satisfy structural Specker's principle, namely, that every clique in $O(\Gamma)$ is a subset of some hyperedge in $\Gamma$?}

A positive answer to this question would answer Problem 7.2.3 of Ref.~\cite{AFLS} asking for a characterization of $\Gamma$ for which $\mathcal{CE}^1(\Gamma)=\mathcal{G}(\Gamma)$.

\subsubsection{What does it mean for an operational theory to satisfy structural/statistical Specker's principle?}\label{Speck2}

We have so far defined structural Specker's principle as a constraint on $\Gamma$ and statistical Specker's principle as a constraint on a probabilistic model on any $\Gamma$. Any operational theory would typically allow many possible $\Gamma$ to be realized by its measurement events as well as many possible probabilistic models to be realized on any $\Gamma$ representing its measurement events. Note that when we say that a particular $\Gamma$ is ``realizable" or ``allowed" by an operational theory, we mean that there exist measurement events in the operational theory that satisfy the operational equivalences required by $\Gamma$.\footnote{Realizability of a particular $\Gamma$ in an operational theory depends on the (in)compatibility relations that the operational theory allows between its measurements (cf.~Section \ref{compatibilitysection}). Recall that incompatibility of measurements is {\em necessary} for KS-contextuality to be witnessed and the structure of $\Gamma$ depends on this incompatibility.} Further, given such a $\Gamma$, the realizability of a probabilistic model on it by the operational theory means that there exists a source event in the operational theory that assigns probabilities to the measurement events in $\Gamma$ according to the probabilistic model. It will be useful for our discussion to define what it means for an operational theory, say $\mathbb{T}$, to  satisfy structural or statistical Specker's principle. But before we do that, let us formally specify what it means for $\mathbb{T}$ to satisfy Specker's principle:

{\bf $\mathbb{T}$ satisfies Specker's principle}: {\em An operational theory $\mathbb{T}$ is said to satisfy Specker's principle if, for any set of measurement settings in $\mathbb{T}$ that are pairwise jointly implementable, it follows that they are all jointly implementable in $\mathbb{T}$.}\footnote{Recall from Section \ref{compatibilitysection} the definition of joint implementability (or joint measurability) of some set of measurement settings.} 

We denote by $\mathcal{T}(\Gamma)$ the set of probabilistic models achievable on $\Gamma$ by an operational theory $\mathbb{T}$, i.e., for any $p\in\mathcal{T}(\Gamma)$, we have that $\forall v\in V(\Gamma): p(v)=p(v|S,s)$ for some source event $[s|S]$ possible in the operational theory $\mathbb{T}$.\footnote{Note that if the operational theory does not admit measurement events (represented by vertices) exhibiting the operational equivalences represented by $\Gamma$ (that is, $\mathbb{T}$ does not allow $\Gamma$), then we have that $\mathcal{T}(\Gamma)$ is an empty set.} Since an operational theory can only put further constraints on probabilistic models in $\mathcal{G}(\Gamma)$, we obviously have: $\mathcal{T}(\Gamma)\subseteq\mathcal{G}(\Gamma)$.

\begin{enumerate}
	\item 
{\bf $\mathbb{T}$ satisfies statistical Specker's principle:} {\em We say an operational theory $\mathbb{T}$ satisfies statistical Specker's principle if $\mathcal{T}(\Gamma)\subseteq\mathcal{CE}^1(\Gamma)\subseteq\mathcal{G}(\Gamma)$ for all $\Gamma$.}\footnote{That is, instead of considering only a particular probabilistic model on a particular $\Gamma$, we now consider the satisfaction of statistical Specker's principle by a whole set of probabilistic models, namely, $\mathcal{T}(\Gamma)$, for all $\Gamma$.}

Since the satisfaction of statistical Specker's principle is a constraint on the statistical predictions of $\mathbb{T}$, there must be some {\em fact} about the structure of theory $\mathbb{T}$ that leads to this constraint. This fact  enforcing statistical Specker's principle could be some restriction arising from the structure of allowed measurement events and/or even the structure of allowed preparations in the operational theory $\mathbb{T}$. For instance, this is the case for quantum theory when one only considers projective measurements implemented on an arbitrary quantum state, i.e., $\mathcal{Q}(\Gamma)\subseteq\mathcal{CE}^1(\Gamma)\subseteq\mathcal{G}(\Gamma)$, where $\mathcal{Q}(\Gamma)$ denotes the set of probabilistic models that can be obtained in this way. More generally, one could relax the no-restriction hypothesis \cite{CDP} in some particular way in $\mathbb{T}$ so that not all probabilistic models in $\mathcal{G}(\Gamma)$ are allowed in $\mathcal{T}(\Gamma)$. In the case of quantum theory, restricting attention to only projective measurements (as we just pointed out) rather than the more general case allowing arbitrary POVMs is one way of restricting the set of possible probabilistic models realizable with quantum states and measurements to a strict subset of $\mathcal{G}(\Gamma)$. Allowing arbitrary POVMs would lead to a violation of statistical Specker's principle by probabilistic models arising from quantum theory.\footnote{See Appendices \ref{allegedstrawman} (specifically \ref{trivialpovms}) and \ref{trivial} for other consequences of allowing arbitrary POVMs, in particular the trivial `classical' ones.}

Let us now define what it means for an operational theory $\mathbb{T}$ to satisfy structural Specker's principle.

\item {\bf $\mathbb{T}$ satisfies structural Specker's principle:} {\em An operational theory $\mathbb{T}$ is said to satisfy structural Specker's principle if for any set of measurement events that are pairwise jointly measurable, i.e, measurement events in each pair arise as outcomes of some measurement in the theory, it is the case that all the measurement events in the set are jointly measurable, i.e., all the measurement events in the set arise as outcomes of a single measurement in the theory.}
\end{enumerate}

We now show that a theory $\mathbb{T}$ that satisfies Specker's principle also satisfies structural Specker's principle.

\begin{theorem}\label{SPimpliesStrSP}
	If an operational theory $\mathbb{T}$ satisfies Specker's principle, then it also satisfies structural Specker's principle.
\end{theorem}
\begin{proof}
	The argument here relies on the fact that the operational theory $\mathbb{T}$ is such that measurement settings can be coarse-grained to yield new measurement settings with fewer outcomes. Operationally, this just corresponds to binning some subsets of outcomes together in a measurement procedure. The operational theories we consider in this paper satisfy this property, as outlined in Section \ref{coarsegrainings} on coarse-graining.
	 
	The argument proceeds, for any $\Gamma$ realizable in $\mathbb{T}$, by constructing a set of binary-outcome measurement settings for any given set of pairwise jointly measurable vertices in $\Gamma$. These measurement settings are, by construction, pairwise jointly measurable, so Specker's principle applied to them implies that they are all jointly measurable. This in turn means that the pairwise jointly measurable vertices in the given set are also all realizable as outcomes of a single measurement setting. Hence, the theory $\mathbb{T}$ satisfies structural Specker's principle. We detail the argument below.
	
	Consider a contextuality scenario $\Gamma$ realizable in $\mathbb{T}$. To each vertex $v\in V(\Gamma)$, we can associate a measurement setting $M_v$ with two possible outcomes labelled $\{0,1\}$ such that $[1|M_v]$ denotes the occurrence of $v$ and $[0|M_v]$ denotes the non-occurrence of $v$, i.e., $p(v|S,s)=p(1|M_v,S,s)$ and $1-p(v|S,s)=p(0|M_v,S,s)$ for any probabilistic model on $\Gamma$ induced by some source event $[s|S]$. The measurement setting $M_v$ can be obtained in various (operationally equivalent) ways from the hyperedges that $v\in V(\Gamma)$ appears in: for each hyperedge $e\in E(\Gamma)$ such that $v\in e$, we have that the binary-outcome measurement setting consisting of the vertices $\{v, e\backslash v\}$ --- where $e\backslash v$ denotes a coarse-graining over all the measurement outcomes of $e$ except $v$ ---  is operationally equivalent to $M_v$.
	
	Now, for any pair of vertices $\{v,v'\}$ that appear in a common hyperedge of $\Gamma$, consider the two corresponding measurement settings $\{M_v, M_{v'}\}$ such that they are jointly measurable and their outcomes are mutually exclusive. The measurement events that can possibly occur in their joint measurement, denoted $M_{vv'}$, are $[10|M_{vv'}], [01|M_{vv'}]$ and $[00|M_{vv'}]$. The probability of $[11|M_{vv'}]$ is always zero, reflecting the fact that $v$ and $v'$ are mutually exclusive. Here, the coarse-graining relations are: $[1|M_v]\equiv [10|M_{vv'}], [1|M_{v'}]\equiv [01|M_{vv'}], [0|M_v]\equiv [00|M_{vv'}]+[01|M_{vv'}],[0|M_{v'}]\equiv [00|M_{vv'}]+[10|M_{vv'}]$.
	
	The joint measurement $M_{vv'}$ can be constructed from any hyperedge that $v$ and $v'$ appear in: for any $e\in E(\Gamma)$ such that $\{v,v'\}\subseteq e$, we have that $[10|M_{vv'}]$ is a measurement event corresponding to $v$,\footnote{Recall that every vertex $v\in V(\Gamma)$ is an {\em equivalence class} of measurement events $[v|e]\simeq [v|e']$ for all $e,e'$ such that $v\in e$ and $v\in e'$.} $[01|M_{vv'}]$ corresponds to $v'$, $[00|M_{vv'}]$ corresponds to $e\backslash\{v,v'\}$ (the coarse-graining of all measurement outcomes in $e$ except $v$ and $v'$), and $[11|M_{vv'}]$ denotes the null event $\varnothing\subseteq e$. This means $p(10|M_{vv'},S,s)+p(01|M_{vv'},S,s)+p(00|M_{vv'},S,s)=p(v|S,s)+p(v'|S,s)+p(e\backslash\{v,v'\}|S,s)=1$ and $p(11|M_{vv'},S,s)=0$ for any probabilistic model (induced by some source event $[s|S]$) on $\Gamma$.
	
	Consider now any set of vertices in $\Gamma$ that is pairwise jointly measurable, denoted $V_{\rm 2JM}\subseteq V(\Gamma)$. We need to show that any such set of vertices $V_{\rm 2JM}$ is jointly measurable, i.e., the theory $\mathbb{T}$ realizing $\Gamma$ admits a single measurement such that all the vertices in $V_{\rm 2JM}$ arise as outcomes of this measurement. 
	
	Now, the two-outcome measurement settings $\{M_{v}|v\in V_{\rm 2JM}\}$ we have defined are pairwise jointly measurable and as such, following Specker's principle, they should all be jointly measurable in theory $\mathbb{T}$. The joint measurement corresponding to them can be defined as
	\begin{equation}
	M_{V_{\rm 2JM}}\equiv\{[\vec{b}|M_{V_{\rm 2JM}}]\big|
	\vec{b}
	\in \{0,1\}^{V_{\rm 2JM}}\},
	\end{equation}
	where each event $[\vec{b}|M_{V_{\rm 2JM}}]$ in the joint measurement $M_{V_{\rm 2JM}}$ represents a particular set of outcomes for measurements in the set $\{M_v|v\in V_{\rm 2JM}\}$.
	
	Denoting $V_{\rm 2JM}\equiv\{v_1,v_2,\dots, v_{|V_{\rm 2JM}|}\}$, we have that 
	\begin{align}
	&[(10\dots0)|M_{V_{\rm 2JM}}]\equiv [1|M_{v_1}],\nonumber\\
	&[(01\dots0)|M_{V_{\rm 2JM}}]\equiv [1|M_{v_2}],\nonumber\\
	&\vdots\nonumber\\ 
	&[(00\dots1)|M_{V_{\rm 2JM}}]\equiv [1|M_{v_{|V_{\rm 2JM}|}}],\nonumber\\
	&[(00\dots0)|M_{V_{\rm 2JM}}]\nonumber\\
	&\equiv [0|M_{v_1}]+[0|M_{v_2}]+\dots+[0|M_{v_{|V_{\rm 2JM}|}}],
	\end{align}
	where $[0|M_{v_1}]+[0|M_{v_2}]+\dots+[0|M_{v_{|V_{\rm 2JM}|}}]$ denotes the measurement event obtained by coarse-graining the  measurement events in $\{[0|M_v]|v\in|V_{\rm 2JM}|\}$. All the other measurement events of $M_{V_{\rm 2JM}}$ are null events that never occur, i.e., they are assigned probability zero by every source event. 
	Thus, using Specker's principle applied to the binary-outcome measurement settings defined for the vertices in $V_{\rm 2JM}$, we have that the pairwise jointly measurable vertices in $V_{\rm 2JM}$ are all jointly measurable, appearing as outcomes of a single measurement $M_{V_{\rm 2JM}}$. 
	
\end{proof}

Having established Theorem \ref{SPimpliesStrSP}, we now proceed to show that a theory which satisfies structural Specker's principle also satisfies statistical Specker's principle.  To do this, we consider a contextuality scenario $\Gamma$ which may not satisfy structural Specker's principle and from it construct a contextuality scenario $\Gamma'$ which does satisfy the principle. The construction proceeds as follows:

\begin{enumerate}
	\item Construct $O(\Gamma)$.
	
	\item Turn each clique in $O(\Gamma)$ that {\em is} a hyperedge in $\Gamma$ to  a hyperedge in a new hypergraph $\Gamma'$. That is, $\Gamma'$ is such that $V(\Gamma)\subseteq V(\Gamma')$ and $E(\Gamma)\subseteq E(\Gamma')$. 
	
	\item Turn each maximal clique $c$ in $O(\Gamma)$ that is {\em not} a hyperedge in $\Gamma$ to a hyperedge in $\Gamma'$ and include an additional vertex $v_c$ in this hyperedge. Here, a {\em maximal clique} in a graph is a clique that is not a strict subset of another clique, i.e., there is no vertex outside the clique that shares an edge with each vertex in the clique.
	
	We then have for the hyperedges of $\Gamma'$, 
	\begin{equation}
	E(\Gamma')=E(\Gamma)\cup\{c\cup \{v_c\}\}_{c\in C},
	\end{equation}
	where $C$ is the set of maximal cliques in $O(\Gamma)$ that are not hyperedges in $\Gamma$.
	
	Note that as long as a theory $\mathbb{T}$ satisfies structural Specker's principle, converting maximal cliques in $O(\Gamma)$ that are not hyperedges in $\Gamma$ to hyperedges in $\Gamma'$ is a valid move within the theory since the resulting hyperedge would indeed constitute a valid measurement in the theory.
	
	If $C=\varnothing$ (i.e., $\Gamma$ satisfies structural Specker's principle), then we just have $E(\Gamma')=E(\Gamma)$.

	\item The resulting contextuality scenario $\Gamma'$ is thus given by: $V(\Gamma')=V(\Gamma)\cup \{v_c\}_{c\in C}$ and $E(\Gamma')=E(\Gamma)\cup\{c\cup \{v_c\}\}_{c\in C}$.
	
	If $C=\varnothing$ we just have $V(\Gamma')=V(\Gamma)$ and $E(\Gamma')=E(\Gamma)$ so that $\Gamma'=\Gamma$ (i.e., the two hypergraphs are isomorphic).
	
\end{enumerate}

Our construction of $\Gamma'$ leads to the following properties:
\begin{itemize}
	\item $\Gamma'$ satisfies structural Specker's principle (by construction) since every clique in $O(\Gamma')$ is a subset of some hyperedge in $\Gamma'$. Hence, it's also the case that statistical Specker's principle holds for probabilistic models on $\Gamma'$ as $\mathcal{CE}^1(\Gamma')=\mathcal{G}(\Gamma')$.

Note that the construction of $\Gamma'$ relied on the fact that the theory we are considering satisfies structural Specker's principle. If the theory doesn't satisfy this principle, but one goes ahead with the construction of $\Gamma'$, then the new hyperedges in $\Gamma'$ may not constitute valid measurements in the theory.

\item Probabilistic models in $\mathcal{G}(\Gamma')$ are in bijective  correspondence with probabilistic models in $\mathcal{CE}^1(\Gamma)$: for any probabilistic model $p_{\Gamma}\in \mathcal{CE}^1(\Gamma)$, there exists a unique probabilistic model $p_{\Gamma'}\equiv f(p_{\Gamma})\in\mathcal{G}(\Gamma')$, where the function $f$ is given by $p_{\Gamma'}(v)\equiv f(p_{\Gamma})(v)=p_{\Gamma}(v)$ for all $v\in V(\Gamma)$ and $p_{\Gamma'}(v_c)\equiv f(p_{\Gamma})(v_c)=1-\sum_{v\in c}p_{\Gamma}(v)$ for all $c\in C$.\footnote{Recall that $\{v_c\}_{c\in C}=V(\Gamma')\backslash V(\Gamma)$.} Similarly, for any $p_{\Gamma'}\in\mathcal{G}(\Gamma')$, there exists a unique probabilistic model $p_{\Gamma}\equiv g(p_{\Gamma'})\in\mathcal{CE}^1(\Gamma)$ given by $p_{\Gamma}(v)\equiv g(p_{\Gamma'})(v)=p_{\Gamma'}(v)$ for all $v\in V(\Gamma)$, i.e., we simply ignore the probabilities assigned to the vertices $v_c\in V(\Gamma')\backslash V(\Gamma)$ which do not appear in $\Gamma$. Now note that the functions $f$ and $g$ are inverses of each other: $g(f(p_{\Gamma}))=g(p_{\Gamma'})=p_{\Gamma}$ and $f(g(p_{\Gamma'}))=f(p_{\Gamma})=p_{\Gamma'}$. Hence, there is a bijective correspondence between $\mathcal{G}(\Gamma')$ and $\mathcal{CE}^1(\Gamma)$.

\item Hence, the set of probabilistic models on $\Gamma$ that satisfy statistical Specker's principle, i.e., $\mathcal{CE}^1(\Gamma)$, are in one-to-one correspondence with the set of probabilistic models on $\Gamma'$ which (by construction) satisfies structural Specker's principle so that $\mathcal{CE}^1(\Gamma')=\mathcal{G}(\Gamma')$. 

We therefore have that $\mathcal{CE}^1(\Gamma)=\mathcal{CE}^1(\Gamma')|_{V(\Gamma)}$, where $\mathcal{CE}^1(\Gamma')|_{V(\Gamma)}$ denotes the probabilistic models induced on $\Gamma$ by those on $\Gamma'$ (ignoring the probabilities assigned to vertices in $V(\Gamma')\backslash V(\Gamma)$).

\end{itemize}

It is conceivable that a particular $\Gamma$ may not admit probabilistic models from an operational theory $\mathbb{T}$, i.e., $\mathcal{T}(\Gamma)=\varnothing$. On the other hand, if $\Gamma$ admits a representation in terms of measurement events admissible in $\mathbb{T}$, so that $\mathcal{T}(\Gamma)\neq\varnothing$, then two possibilities arise: $\Gamma$ satisfies structural Specker's principle or it doesn't. If $\Gamma$ satisfies structural Specker's principle then any probabilistic model in $\mathcal{T}(\Gamma)$ will satisfy statistical Specker's principle and we have $\Gamma'=\Gamma$. If $\Gamma$ does not satisfy structural Specker's principle, we consider its relation with the contextuality scenario $\Gamma'$ constructed from it that does satisfy structural Specker's principle. Such a $\Gamma'$ admits a representation in a theory $\mathbb{T}$ satisfying structural Specker's principle (that is, $\mathcal{T}(\Gamma')\neq\varnothing$) as long as $\Gamma$ admits such a representation (that is, $\mathcal{T}(\Gamma)\neq\varnothing$). Indeed, it's the satisfaction of structural Specker's principle in $\mathbb{T}$ that renders the construction of $\Gamma'$ from $\Gamma$ physically allowed in $\mathbb{T}$.

Thus, in a theory $\mathbb{T}$ that satisfies structural Specker's principle, the following holds: for every probabilistic model $p_{\Gamma}\in\mathcal{T}(\Gamma)$  $(\subseteq\mathcal{CE}^1(\Gamma))$, $\Gamma'$ admits a corresponding probabilistic model $p_{\Gamma'}\in\mathcal{T}(\Gamma')$ satisfying $p_{\Gamma'}(v)=p_{\Gamma}(v)$ for all $v\in V(\Gamma)$ and $p_{\Gamma'}(v_c)=1-\sum_{v\in c}p_{\Gamma}(v)$ for all $c\in C$, where $C$ is the set of maximal cliques in $O(\Gamma)$ such that none of them is a hyperedge in $\Gamma$. Similarly, given $p_{\Gamma'}\in\mathcal{T}(\Gamma')$  $(\subseteq\mathcal{CE}^1(\Gamma'))$, $p_{\Gamma}\in\mathcal{T}(\Gamma)$ is uniquely fixed: it's obtained by just neglecting the probabilities assigned by $p_{\Gamma'}$ to the vertices in $V(\Gamma')\backslash V(\Gamma)$.

We must therefore have $\mathcal{T}(\Gamma)=\mathcal{T}(\Gamma')\big|_{V(\Gamma)}$ for any $\Gamma$, where $\mathcal{T}(\Gamma')\big|_{V(\Gamma)}$ denotes the set of probabilistic models induced on $\Gamma$ by the set of probabilistic models in $\mathcal{T}(\Gamma')$ under the correspondence we have already established above.
We can now state and prove the following theorem:
\begin{theorem}\label{strimpliesstat}
If an operational theory $\mathbb{T}$ satisfies structural Specker's principle, then it also satisfies statistical Specker's principle.
\end{theorem}
\begin{proof}
For any $\Gamma$ that does not admit a probabilistic model in $\mathbb{T}$, i.e., $\mathcal{T}(\Gamma)=\varnothing$, statistical Specker's principle is trivially satisfied since $\mathcal{T}(\Gamma)=\varnothing\subseteq\mathcal{CE}^1(\Gamma)\subseteq\mathcal{G}(\Gamma)$.

For any $\Gamma$ that does admit a probabilistic model in $\mathbb{T}$, i.e., $\mathcal{T}(\Gamma)\neq\varnothing$, we can have one of two possibilities: either it satisfies structural Specker's principle, in which case $\mathcal{T}(\Gamma)\subseteq\mathcal{CE}^1(\Gamma)=\mathcal{G}(\Gamma)$, or it doesn't, in which case we consider the $\Gamma'$ constructed from it following the recipe we have already outlined so that we have: 

$\mathcal{T}(\Gamma')\big|_{V(\Gamma)}\subseteq\mathcal{G}(\Gamma')\big|_{V(\Gamma)}=\mathcal{CE}^1(\Gamma')\big|_{V(\Gamma)}=\mathcal{CE}^1(\Gamma)$.

Since $\mathbb{T}$ satisfies structural Specker's principle, we have $\mathcal{T}(\Gamma)=\mathcal{T}(\Gamma')\big|_{V(\Gamma)}$, which immediately implies that $\mathcal{T}(\Gamma)\subseteq \mathcal{CE}^1(\Gamma)$. That is, the theory $\mathbb{T}$ satisfies statistical Specker's principle on $\Gamma$: $\mathcal{T}(\Gamma)\subseteq \mathcal{CE}^1(\Gamma)\subseteq\mathcal{G}(\Gamma)$.

Overall, we have the desired result: $\mathbb{T}$ satisfies structural Specker's principle $\Rightarrow$ $\mathcal{T}(\Gamma)\subseteq \mathcal{CE}^1(\Gamma)\subseteq\mathcal{G}(\Gamma)$ for all $\Gamma$, i.e., $\mathbb{T}$ satisfies statistical Specker's principle.
\end{proof}

Thus, {\em one} way of enforcing that a particular operational theory $\mathbb{T}$ satisfies statistical Specker's principle --- that is, $\mathcal{T}(\Gamma)\subseteq\mathcal{CE}^1(\Gamma)\subseteq\mathcal{G}(\Gamma)$ for all $\Gamma$ --- is to require that it satisfies structural Specker's principle, a constraint on the structure of measurement events in $\mathbb{T}$. This is, for example, what is achieved in Ref.~\cite{giuliosharp} by invoking a notion of ``sharpness" for measurement events in an operational theory such that any set of sharp measurement events that are pairwise jointly measurable are all jointly measurable. That is, structural Specker's principle is satisfied in a theory with such sharp measurement events and, consequently, statistical Specker's principle, or what is more conventionally called {\em consistent exclusivity} \cite{AFLS}, is also satisfied. But it's conceivable that there may be {\em other} ways to ensure that only a subset of $\mathcal{CE}^1(\Gamma)$ probabilistic models are allowed in $\mathcal{T}(\Gamma)$ for any $\Gamma$. What we wish to emphasize here is that it is by no means obvious (or at least, it needs to be proven) that the {\em only} way to restrict the set of probabilistic models $\mathcal{T}(\Gamma)$ to a subset of $\mathcal{CE}^1(\Gamma)$ for any $\Gamma$ is to require that the theory $\mathbb{T}$ satisfy structural Specker's principle.\footnote{Indeed, any putative theory yielding the set of almost quantum correlations (which satisfy statistical Specker's principle) \cite{almostquantum} {\em cannot} satisfy Specker's principle --- that pairwise joint implementable measurement {\em settings} are all jointly implementable --- for any notion of sharp measurements \cite{almostquantumsharp}. Whether structural Specker's principle, which is defined at the level of measurement {\em events}, can be upheld for an almost quantum theory --- so that it falls in the category of operational theories with sharp measurements envisaged in Ref.~\cite{giuliosharp} --- remains an open question.}

\begin{corollary}\label{SPstrSPCE}
	For any operational theory $\mathbb{T}$, the following implications hold:
	
	\begin{align}
	&\mathbb{T} \textrm{ satisfies Specker's principle}\nonumber\\
	\Rightarrow&\mathbb{T} \textrm{ satisfies structural Specker's principle}\\
	\Rightarrow& \mathbb{T} \textrm{ satisfies statistical Specker's principle,}\nonumber\\
	&\textrm{i.e., consistent exclusivity}.
	\end{align}
	
\end{corollary}	
\begin{proof}
	This follows from combining Theorems \ref{SPimpliesStrSP} and \ref{strimpliesstat}.
\end{proof}

Note that statistical Specker's principle (or consistent exclusivity) is so intrinsic to the CSW approach \cite{CSW} that they do not consider probabilistic models that do not satisfy this principle.\footnote{As we have already noted, a noise-robust noncontextuality inequality of the type in Ref.~\cite{KunjSpek} that is based on a logical proof of the KS theorem is not even obtainable if one restricted attention to probabilistic models satisfying ${\rm CE}^1$. The upper bound on that inequality comes from a probabilistic model that {\em does not} satisfy ${\rm CE}^1$.} This will become important when we consider the fact that nonprojective measurements in quantum theory {\em do not} satisfy Specker's principle, structural or statistical (at the level of measurement {\em events}), and thus also fail to satisfy the stronger statement of Specker's principle for measurement {\em settings} (cf.~Ref.~\cite{almostquantumsharp}). Indeed, such measurements admit contextuality scenarios 
$\Gamma$ that are not possible with projective measurements, such as the one from three binary-outcome POVMs that are pairwise jointly measurable but not triplewise so \cite{LSW, KG, KHF}, and the probabilistic models they give rise to can only be accommodated in the most general set of probabilistic models, $\mathcal{G}(\Gamma)$, since trivial POVMs can realize any probabilistic model. Specker's principle, structural Specker's principle, and statistical Specker's principle were all motivated by the fact that projective measurements in quantum theory satisfy them. In particular, consistent exclusivity (or statistical Specker's principle) would be obeyed in any theory where measurement events satisfy structural Specker's principle, and indeed, the more recent approach \cite{cabellotalk} is to restrict attention to ``sharp" measurements in such theories \cite{giuliosharp,giuliosharp2}, where the definition of ``sharp" ensures the property of pairwise jointly measurable events being globally jointly measurable. This property forms the motivational basis (and is {\em sufficient}) for statistical Specker's principle to hold (cf.~Theorem \ref{strimpliesstat}). That is, this approach \cite{giuliosharp, cabellotalk} regards statistical Specker's principle as grounded in (and physically justified by) structural Specker's principle. Theorem \ref{strimpliesstat} is a precise statement of this intuition in the hypergraph formalism \`a la AFLS \cite{AFLS}. The work of Refs.~\cite{giuliosharp, giuliosharp2} can be understood as bridging the gap between structural Specker's principle and statistical Specker's principle by formally defining a notion of sharp measurements in an operational theory such that structural Specker's principle holds for these sharp measurements.

On the other hand, and this is the key point for our purposes, if one wants to make no commitment about the representation of measurements in the operational theory (in particular, not requiring a notion of ``sharpness"), then Specker's principle is not a natural constraint to impose on probabilistic models and, indeed, one must deal with the full set of probabilistic models $\mathcal{G}(\Gamma)$ on any contextuality scenario $\Gamma$ rather than restrict oneself to the set of probabilistic models $\mathcal{CE}^1(\Gamma)$. It is for this reason that we are translating the notions from CSW \cite{CSW} to the notational conventions of AFLS \cite{AFLS}, the latter being a more natural choice for our purposes, allowing the language needed to articulate the difference between $\mathcal{CE}^1(\Gamma)$ and $\mathcal{G}(\Gamma)$ rather than excluding the latter by fiat or, perhaps, by an appeal to structural Specker's principle holding for sharp measurements in the landscape of operational theories under consideration (cf.~Theorem \ref{strimpliesstat}). It is for all these reasons that the ``exclusivity principle" \`a la CSW \cite{CSW} is not enough to make sense of Spekkens contextuality applied to Kochen-Specker type scenarios. The framework we propose in this paper addresses this gap between the notions Spekkens contextuality (which applies to arbitrary measurements) requires in a hypergraph framework and those that the CSW framework \cite{CSW} (which applies to ``sharp" measurements) can provide in its graph-theoretic formulation.

\subsubsection{Remark on the classification of probabilistic models: why we haven't defined ``quantum models" as those obtained from projective measurements}
The reader may note that we haven't tried to define any notion of a ``quantum model" so far, having only adopted the definitions of Ref.~\cite{AFLS} for KS-noncontextual models ($\mathcal{C}(\Gamma)$), for models satisfying consistent exclusivity ($\mathcal{CE}^1(\Gamma)$), and for general probabilistic models ($\mathcal{G}(\Gamma)$). The reason for this is that we do not wish to restrict ourselves to projective measurements in defining a ``quantum model", unlike the traditional Kochen-Specker approaches \cite{CSW, AFLS}. In Ref.~\cite{AFLS}, a {\em quantum model} is defined as a probabilistic model that can be realized in the following manner: assign projectors $\{\Pi_v\}_{v\in V(\Gamma)}$ (defined on any Hilbert space) to all the vertices of $\Gamma$ such that $\sum_{v\in e}\Pi_v=\mathbb{I}$ for all $e\in E(\Gamma)$, and we have $p(v)=\Tr(\rho\Pi_v)$, for some density operator $\rho$ on the Hilbert space, $\mathbb{I}$ being the identity operator.

On the other hand, allowing arbitrary positive operator-valued measures (POVMs) in a definition of a quantum model (as we would rather prefer) means that, in fact, quantum models on a hypergraph $\Gamma$ are as general as the general probabilistic models $\mathcal{G}(\Gamma)$, rendering such a definition {\em redundant}. This can be seen by noting that for any probabilistic model $p\in\mathcal{G}(\Gamma)$, one can associate positive operators to the vertices of $\Gamma$ given by $p(v)\mathbb{I}$ such that for any quantum state $\rho$ on some Hilbert space, we have $p(v)=\Tr(\rho p(v)\mathbb{I})$, where $\mathbb{I}$ is the identity operator.

Our focus in this paper is not on quantum theory, in particular, even though the need to be able to handle noisy measurements and preparations (particularly, trivial POVMs) in quantum theory can be taken as a motivation for this work. Rather, our focus is on delineating the boundary between operational theories that admit noncontextual ontological models (for Kochen-Specker type experiments, suitably augmented with multiple preparation procedures, as outlined in this paper) and those that don't by obtaining noise-robust noncontextuality inequalities. In particular, we want these inequalities to indicate the noise thresholds beyond which an experiment cannot rule out the existence of a noncontextual ontological model with respect to the quantities of interest. This also means that making sense of quantum correlations in this approach requires one to pay attention not only to the measurements involved in an experiment but {\em also} the preparations; indeed, this shift of focus from measurements alone, to include multiple preparations (or source settings), is a fundamental conceptual difference between our approach and that of traditional Kochen-Specker contextuality frameworks \cite{CSW,AFLS,sheaf}.

\subsubsection{Scope of this framework}
Note that whenever we refer to the ``CSW framework", we mean the framework of Ref.~\cite{CSW}, which often differs from the framework of Ref.~\cite{CSW10} in some respects, e.g., the normalization of probabilities in a given hyperedge, assumed in \cite{CSW}, but not in \cite{CSW10}. In Ref.~\cite{CSW10}, the authors write: 
\begin{quotation}
	Notice that in all of the above we never require that any particular context should be associated to a complete measurement: the conditions only make sure that each context is a subset of outcomes of a measurement and that they are mutually exclusive.
	Thus, unlike the original KS theorem, it is clear that every context hypergraph $\Gamma$ has always a classical noncontextual model, besides possibly quantum and generalized models.
\end{quotation}
On the other hand, in Ref.~\cite{CSW}, they write: 
\begin{quotation}
	The fact that the sum of probabilities of outcomes of a test is 1 can be used to express these correlations as a positive linear combination of probabilities of events, $S=\sum_iw_iP(e_i)$, with $w_i>0$.
\end{quotation}
The latter presentation \cite{CSW} is more in line with the ``original KS theorem" \cite{KochenSpecker}, as well as the presentation in Ref.~\cite{AFLS}. Since normalization of probabilities is thus presumed in Ref.~\cite{CSW}, in keeping with the definition of a probabilistic model we have presented (following \cite{AFLS}), the graph invariants of CSW \cite{CSW} refer, specifically, to subgraphs $G$ of {\em those} hypergraphs $\Gamma$  on which the set of KS-noncontextual probabilistic models is non-empty. In particular, our generalization of the CSW framework \cite{CSW} in this paper says nothing about noise-robust noncontextuality inequalities from logical proofs of the Kochen-Specker theorem \cite{KochenSpecker}, which rely on hypergraphs $\Gamma$ that admit no KS-noncontextual probabilistic models, i.e., KS-uncolourable hypergraphs. It also says nothing for the hypergraphs $\Gamma$ that do not satisfy the property $\mathcal{CE}^1(\Gamma)=\mathcal{G}(\Gamma)$. An example of such a hypergraph, which is not covered by our generalization of the CSW framework on both counts, is the 18 ray hypergraph first presented in Ref.~\cite{CEGA}, denoted $\Gamma_{18}$ (see Fig.~\ref{18ray} and Appendix \ref{scopeviacega}).
\begin{figure}
	\centering
	\includegraphics[scale=0.3]{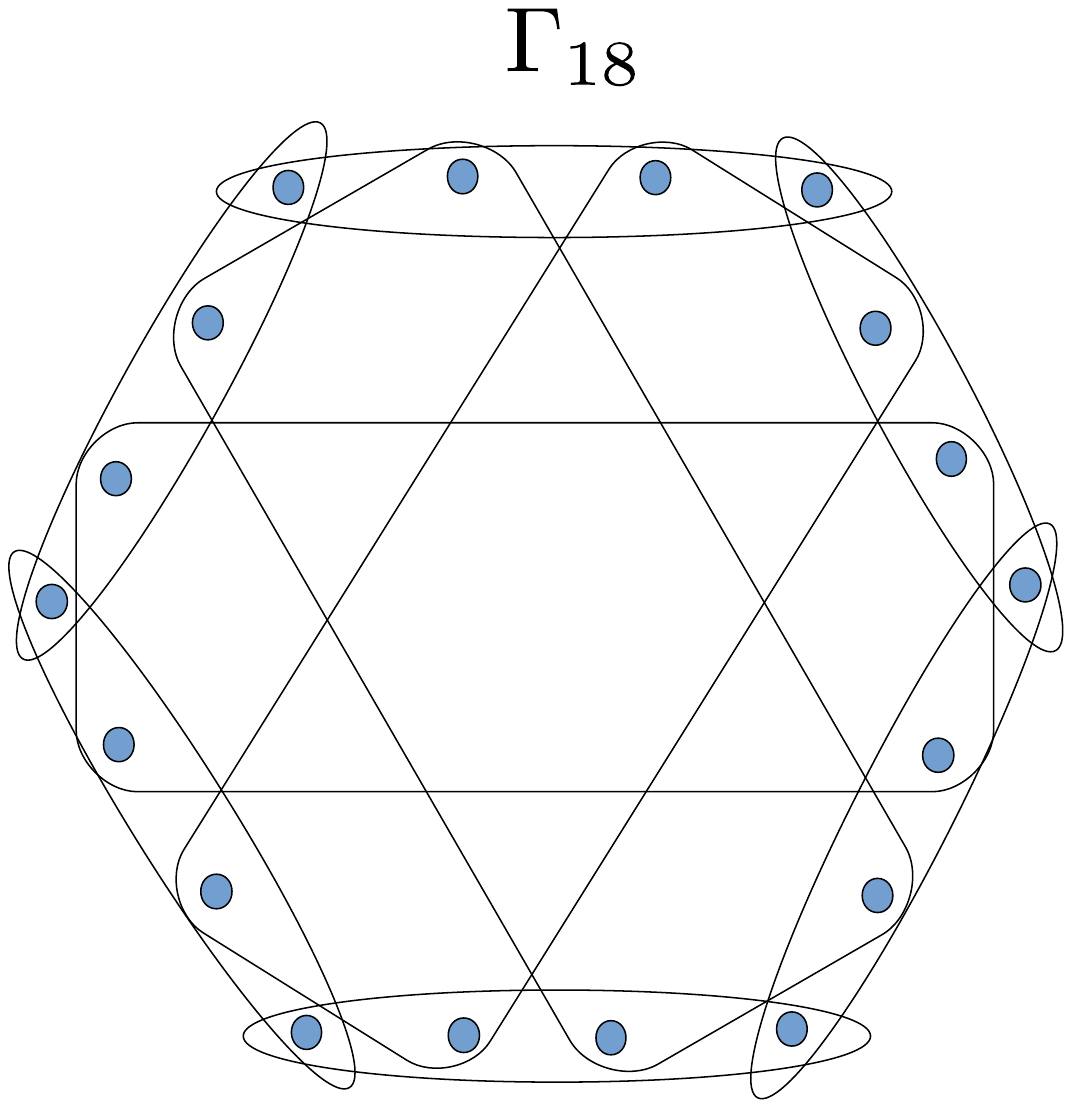}
	\caption{The KS-uncolourable hypergraph from Ref.~\cite{CEGA} that is not covered by our generalization of the CSW framework. We denote this hypergraph as $\Gamma_{18}$.}
	\label{18ray}
\end{figure}
  Indeed, the study of noise-robust noncontextuality inequalities from such KS-uncolourable hypergraphs was initiated in Ref.~\cite{KunjSpek}, and a more exhaustive hypergraph-theoretic treatment of it is presented in Ref.~\cite{kunjunc}. In this paper, we will restrict ourselves to KS-colourable hypergraphs, the study of which was initiated in Ref.~\cite{KunjSpek17}, and, of these, {\em only} those KS-colourable hypergraphs $\Gamma$ which satisfy $\mathcal{CE}^1(\Gamma)=\mathcal{G}(\Gamma)$. Note that this is not a limitation of our general approach, which is based on Ref.~\cite{KunjSpek17} and applies to {\em any} KS-colourable hypergraph, but rather a limitation we inherit from the CSW framework \cite{CSW}\footnote{Ref.~\cite{CSW} takes Specker's principle to be fundamental and identifies $\mathcal{CE}^1(\Gamma_{18})$ as the most general set of probabilistic models, which is {\em not} the case for $\Gamma_{18}$ (for example). See Appendix \ref{scopeviacega} for a detailed discussion of this point.} since we want to leverage their graph invariants in obtaining our noise-robust noncontextuality inequalities. The study of other KS-colourable hypergraphs, in particular those which arise {\em only} with nonprojective measurements in quantum theory \cite{LSW,KG,KHF} and are outside the scope of traditional frameworks \cite{CSW, AFLS, sheaf}, will be taken up in future work.
  
  To summarize, the measurement events hypergraphs $\Gamma$ where the present framework (and the CSW framework \cite{CSW}) applies {\em must} satisfy two properties: 
  $\mathcal{C}(\Gamma)\neq\varnothing$ (that is, {\em KS-colourability}) and $\mathcal{CE}^1(\Gamma)=\mathcal{G}(\Gamma)$.\footnote{As we have shown, when the operational theory $\mathbb{T}$ under consideration satisfies structural Specker's principle, we can always turn a hypergraph $\Gamma$ that doesn't satisfy structural Specker's principle into a hypergraph $\Gamma'$ that satisfies it and for which, therefore, $\mathcal{CE}^1(\Gamma')=\mathcal{G}(\Gamma')$ holds. This can be seen as justification for restricting oneself to probabilistic models satisfying consistent exclusivity in the CSW framework \cite{CSW}: such a restriction is not really a restriction if the theory satisfies structural Specker's principle. On the other hand, we restrict ourselves to hypergraphs for which $\mathcal{CE}^1(\Gamma)=\mathcal{G}(\Gamma)$ without assuming that $\mathbb{T}$ satisfies structural Specker's principle. The justification for this seemingly {\em ad hoc} restriction is simply that it is necessary in order to meaningfully leverage the graph invariants of CSW \cite{CSW} -- in particular, the fractional packing number -- in our noise-robust noncontextuality inequalities. This will become clear when we obtain our noise-robust noncontextuality inequalities.}

In the next subsection, we define additional notions necessary to obtain noise-robust noncontextuality inequalities that make use of graph invariants from the CSW framework. These notions correspond to {\em source events} that are an integral part of our framework.

\subsection{Sources}\label{sources}

Having introduced the (hyper)graph-theoretic elements that we need to talk about measurement events, we are now in a position to introduce features of source events that are relevant in the Spekkens framework. This part of our framework has no precedent in the literature on KS-noncontextuality, in particular the CSW framework \cite{CSW}. We introduce these source events in order to benchmark the measurement events against them, i.e., for every measurement event, we seek to identify in the operational theory a corresponding source event that makes this measurement event as likely as possible. This helps us deal with cases where a measurement device may be implementing very noisy measurements by explicitly accounting for this noise in our noise-robust noncontextuality inequalities. Further, while we do not assume outcome determinism (which is essential to KS-noncontextuality), we will invoke preparation noncontextuality with respect to these source events in the Spekkens framework \cite{Spe05}. As an example of what we mean by ``benchmarking" a measurement event against a  source event, consider the case of quantum theory, where any measurement event represented by a projector occurs with probability $1$ for any source event that is represented by an eigenstate of this projector; on the other hand, a positive operator that isn't projective cannot occur with a probability greater than its largest eigenvalue ($<1$) for any source event. We now proceed to describe the necessary hypergraph-theoretic ingredients we need to accommodate source events in our framework.

As we have argued previously, we require the measurement events hypergraph $\Gamma$ to be such that $\mathcal{C}(\Gamma)\neq\varnothing$ and $\mathcal{CE}^1(\Gamma)=\mathcal{G}(\Gamma)$ to be able to obtain noise-robust noncontextuality inequalities that use graph invariants from the CSW framework \cite{CSW}. Hence, we will restrict ourselves to experiments that realize the operational equivalences represented by this class of $\Gamma$. Now, in the CSW framework \cite{CSW}, every Bell-KS expression picks out a particular subgraph $G$ of the orthogonality graph $O(\Gamma)$ of the contextuality scenario $\Gamma$ of interest. This amounts to focussing on a restricted set of probabilities (for the vertices of $G$) rather than probabilities for all the measurement events (represented by vertices of $\Gamma$) in the experiment. Hence, the vertices of $G$ denote the measurement events of interest in a given Bell-KS expression and we have the following:
\begin{itemize}
	\item A general probabilistic model $p\in\mathcal{G}(\Gamma)$ will assign probabilities to vertices in $G$ such that: $p(v)\geq 0$ for all $v\in V(G)$ and $p(v)+p(v')\leq 1$ for every edge $\{v,v'\}\in E(G)$.
	
	\item A probabilistic model $p\in\mathcal{CE}^1(\Gamma)$ will assign probabilities to vertices in $G$ such that: $p(v)\geq 0$ for all $v\in V(G)$ and
	\begin{equation}
	\sum_{v\in c}p(v)\leq 1,
	\end{equation}
	for every clique $c\subseteq V(G)$.
	
	\item A probabilistic model $p\in\mathcal{C}(\Gamma)$ will assign probabilities to vertices in $G$ such that: $p(v)=\sum_k{\rm Pr}(k)p_k(v)$, where ${\rm Pr}(k)\geq0, \sum_k{\rm Pr}(k)=1$, and for each $k$, $p_k$ is a deterministic assignment $p_k(v)\in\{0,1\}$ for all $v\in V(G)$, and $p_k(v)+p_k(v')\leq 1$ for every edge $\{v,v'\}\in E(G)$.	
\end{itemize} 
Since $\Gamma$ is such that $\mathcal{CE}^1(\Gamma)=\mathcal{G}(\Gamma)$, the condition
\begin{equation}
\sum_{v\in c}p(v)\leq 1 \textrm{ for every clique }c\subseteq V(G)\nonumber
\end{equation}
on the probabilities assigned to vertices in $G$ is redundant. We now obtain a simplified hypergraph, $\Gamma_G$, from $G$ as follows: convert all maximal cliques in $G$ to hyperedges and add an extra (no-detection) vertex to each such hyperedge.\footnote{Physically, a ``no-detection" vertex denotes the case when none of the measurement events of interest (here, the events in $G$) for a given measurement setting occur.}

This $\Gamma_G$, for any $G$, will satisfy the property that $\mathcal{CE}^1(\Gamma_G)=\mathcal{G}(\Gamma_G)$ and any probabilistic model on $\Gamma$ assigning probabilities to measurement events in $G$ will correspond to a probabilistic model on $\Gamma_G$ which also assigns the same probabilities to measurement events in $G$. Formally: $V(\Gamma_G)\equiv V(G) \bigsqcup \{v_c|c \textrm{ is a maximal clique in } G\}$, and $E(\Gamma_G)\equiv\{c\sqcup\{v_c\}|c \textrm{ is a maximal clique in } G\}$, where $v_c$ is the extra no-detection vertex added to the hyperedge corresponding to maximal clique $c$ in $G$. 

We have the following probabilistic model on $\Gamma_G$, given a probabilistic model $p\in\mathcal{G}(\Gamma)$: the probabilities assigned to the vertices in $V(G)\subseteq V(\Gamma_G)$ are the same as specified by $p\in\mathcal{G}(\Gamma)$ and the probabilities assigned to the remaining vertices in $V(\Gamma_G)\backslash V(G)$ are given by $p(v_c)=1-\sum_{v\in c}p(v)$, for every maximal clique $c$ in $G$.
\begin{figure}
	\centering
	\includegraphics[scale=0.3]{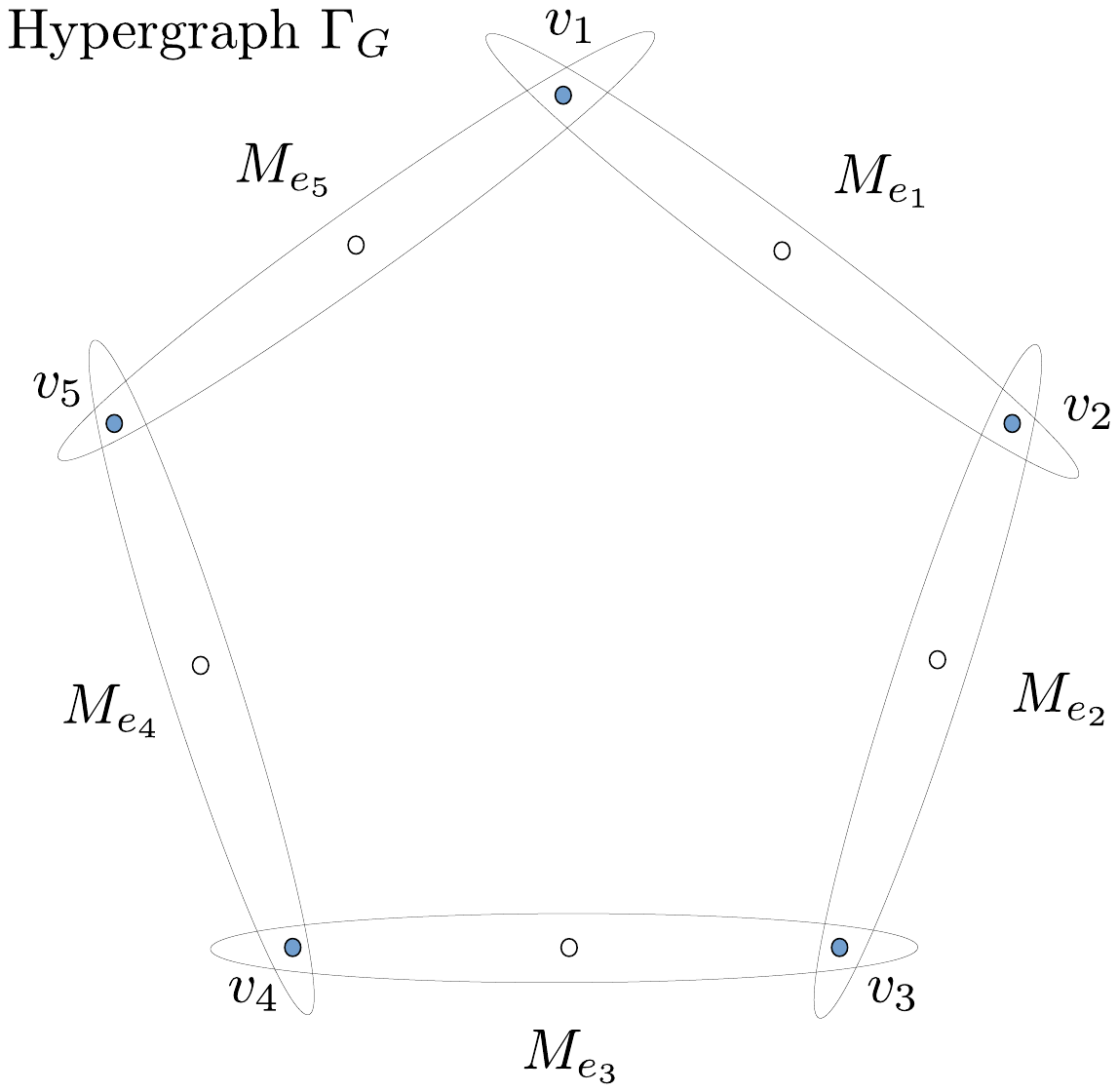}
	
	\caption{The hypergraph $\Gamma_G$ obtained from $G$ by adding a no-detection vertex (represented by a hollow circle) to every maximal clique in $G$.}
	\label{fig4}
\end{figure}
Consider, for example, the KCBS scenario \cite{KCBS, CSW, KunjSpek17}: the 20-vertex $\Gamma$ representing measurement events from five 4-outcome joint measurements (Fig.~\ref{fig2}), its 5 vertices $G$ involved in the KCBS inequality (Fig.~\ref{fig3}), and 10-vertex hypergraph $\Gamma_G$ constructed from $G$ (Fig.~\ref{fig4}).

Given $\Gamma_G$, constructed from $G$, we now require that the operational theory that realizes measurement events in $\Gamma_G$ also admits preparations that can be represented by a hypergraph $\Sigma_G$ of source events as follows: for every hyperedge $e\in E(\Gamma_G)$, corresponding to the choice of measurement setting $M_e$, we define a hyperedge $e\in E(\Sigma_G)$ denoting a corresponding choice of source setting $S_e$. And for every vertex $v\in e (\in E(\Gamma_G))$, we define a vertex $v_e\in e (\in E(\Sigma_G))$.\footnote{Recall from the discussion at the beginning of Section \ref{sources} that we seek to benchmark the measurement events against those source events in the operational theory that (ideally) make them as predictable as possible. The source setting against which the predictability of a particular measurement setting is tested -- that is the predictability of each measurement event (e.g., $v\in e (\in E(\Gamma_G))$) for this measurement setting (e.g., $M_e$) is benchmarked against some source event (e.g., $v_e\in e (\in E(\Sigma_G))$) for the source setting (e.g., $S_e$) --  is the ``corresponding choice of source setting $S_e$". In Section \ref{definingCorr} we will see how these pairs of source and measurement settings are used to compute an operational quantity relevant for our noise-robust noncontextuality inequalities.} Hence, every measurement event $[v|e]$ in $\Gamma_G$ corresponds to a vertex $v_e$ of $\Sigma_G$, and the number of such vertices in $V(\Sigma_G)$ is $|V(\Gamma_G)||E(\Gamma_G)|$. This means that the operational equivalences between the measurement events that are implicit in $\Gamma_G$ --- such as $[v|e]$ is operationally equivalent to $[v|e']$, where $e,e'\in E(\Gamma_G)$ are distinct hyperedges that share the vertex (representing an equivalence class of measurement events) $v\in V(\Gamma_G)$ --- are not carried over to the source events, where none is presumed to be operationally equivalent to any other, hence $v_e\in V(\Sigma_G)$ is a different vertex from $v_{e'}\in V(\Sigma_G)$. Here $v_e$ ($v_{e'}$) represents a source event $[s_e|S_e]$ ($[s_{e'}|S_{e'}]$), rather than an equivalence class of source events.

Besides these $|V(\Gamma_G)||E(\Gamma_G)|$ vertices in $V(\Sigma_G)$ and the associated hyperedges $e\in E(\Sigma_G)$, we require that the operational theory admits an additional hyperedge $e_*\in E(\Sigma_G)$, representing a source setting $S_{e_*}$, containing two new vertices $v^0_{e_*}, v^1_{e_*}\in V(\Sigma_G)$. Here $v^0_{e_*}$ represents the source event $[s_{e_*}=0|S_{e_*}]$ and $v^1_{{e_*}}$ represents the source event $[s_{e_*}=1|S_{e_*}]$. Hence, we have $|V(\Sigma_G)|=|V(\Gamma_G)||E(\Gamma_G)|+2$ and $|E(\Sigma_G)|=|E(\Gamma_G)|+1$.

The operational equivalence we {\em do} require for $\Sigma_G$ (in any operational theory that admits source events represented by $\Sigma_G$) applies to the source settings: all source settings, each represented by coarse-graining the source events in a hyperedge $e\in E(\Sigma_G)$, are operationally equivalent, i.e., $[\top|S_{e_{\top}}]\simeq [\top|S_{{e'}_{\top}}]$ for all $e,e'\in E(\Sigma_G)$, i.e., $\forall [m|M]: \sum_{s_e}p(m,s_e|M,S_e)=\sum_{s_{e'}}p(m,s_{e'}|M,S_{e'})$, for all $e,e'\in E(\Sigma_G)$.

An example of such a source events hypergraph was considered in Ref.~\cite{KunjSpek}, albeit without the additional source labelled by $e_*$ here \cite{KunjSpek17}. We illustrate it here in Fig.~\ref{source} for the KCBS scenario.

\begin{figure}
	\includegraphics[scale=0.3]{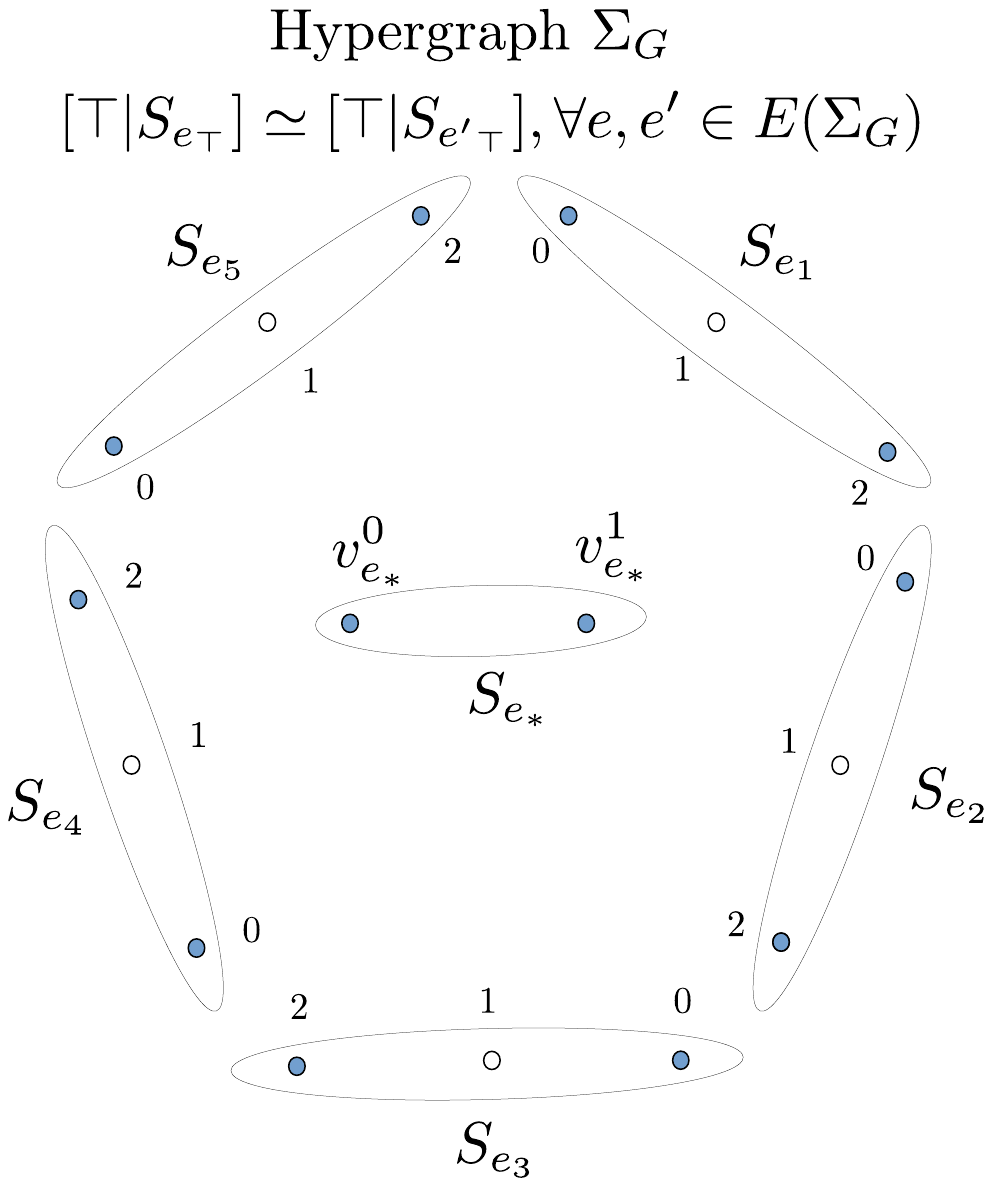}
	\caption{The source events hypergraph with the operational equivalences between the source settings separately specified.}
	\label{source}
\end{figure}

\section{A key hypergraph invariant: the weighted max-predictability}\label{betainvariant}
We now define a hypergraph invariant that will be relevant for our noise-robust noncontextuality inequalities:
\begin{equation}\label{betadefn}
\beta(\Gamma_G,q)\equiv \max_{p\in\mathcal{G}(\Gamma_G)|_{\rm ind}}\sum_{e\in E(\Gamma_G)}q_e\zeta(M_e,p),
\end{equation}
where $q_e\geq0$ for all $e\in E(\Gamma_G)$, $\sum_{e\in E(\Gamma_G)}q_e=1$, and
\begin{equation}
\zeta(M_e,p)\equiv\max_{v\in e} p(v)\nonumber
\end{equation}
is the maximum probability assigned to a vertex in $e\in E(\Gamma_G)$ by an extremal indeterministic probabilistic model $p\in\mathcal{G}(\Gamma_G)|_{\rm ind}$.\footnote{An extremal {\em indeterministic} probabilistic model refers to those extremal $p\in\mathcal{G}(\Gamma_G)$ for which $\zeta(M_e,p)<1$ for some $e\in E(\Gamma_G)$.}

We call $\beta(\Gamma_G,q)$ the {\em weighted max-predictability} of the measurement settings (i.e., hyperedges) in $\Gamma_G$, where the hyperedges $e\in E(\Gamma_G)$ are weighted according to the probability distribution $q\equiv\{q_e\}_{e\in E(\Gamma_G)}$.

We now outline how this quantity is related to properties of an operational theory $\mathbb{T}$ admitting a measurement noncontextual ontological model. $\Gamma_G$ represents a particular configuration of operational equivalences that a set of measurement events in $\mathbb{T}$ may realize. The probabilistic models on $\Gamma_G$ that can be realized by $\mathbb{T}$ are, as earlier, denoted by  $\mathcal{T}(\Gamma_G)$. Since $\mathbb{T}$ admits a measurement noncontextual ontological model,\footnote{This will always be the case for any operational theory we consider: the assumption of measurement noncontextuality on its own can always be satisfied by a trivial ontological model of the type we outlined in Section \ref{noncontextualitysection}. Indeed, quantum theory satisfies it, the Beltrametti-Bugajski model \cite{BB} that was discussed in Ref.~\cite{Spe05} being an example of a measurement noncontextual ontological model of quantum theory. It is only when this assumption is supplemented with something else -- outcome determinism in the case of KS-noncontextuality and preparation noncontextuality in the case of generalized noncontextuality \cite{Spe05} -- that it can produce a contradiction with the predictions of an operational theory.} its predictions for the specific case of $\Gamma_G$ can be reproduced by such a model. But since, in keeping with the CSW approach \cite{CSW}, we will look at witnesses of contextuality tailored to particular experiments ($\Gamma_G$ representing features of one such experiment), we do not need an ontological model for the full theory $\mathbb{T}$ to reproduce its predictions for a particular experiment. Indeed, to construct a measurement noncontextual ontological model for the set of probabilistic models $\mathcal{T}(\Gamma_G)$, it suffices to assume (without loss of generality) that the extremal probabilistic models on $\Gamma_G$ -- given by $\mathcal{G}(\Gamma_G)|_{\rm det}\sqcup\mathcal{G}(\Gamma_G)|_{\rm ind}$ -- are in bijective correspondence with the ontic states ($\Lambda$) of the physical system on which the measurements are carried out. This is because, firstly, {\em any} probabilistic model in $\mathcal{G}(\Gamma_G)$ can be expressed as a convex mixture of extremal probabilistic models in $\mathcal{G}(\Gamma_G)|_{\rm det}\sqcup\mathcal{G}(\Gamma_G)|_{\rm ind}$, and, secondly, associating each ontic state in the ontological model with an extremal probabilistic model\footnote{Representing response functions for the ontic state, i.e., $p(v)=\xi(v|\lambda), \forall v\in V(\Gamma_G)$} in $\mathcal{G}(\Gamma_G)|_{\rm det}\sqcup\mathcal{G}(\Gamma_G)|_{\rm ind}$ means that {\em any} probabilistic model in $\mathcal{G}(\Gamma_G)$ corresponding to predictions of an operational theory (in particular, any $p\in \mathcal{T}(\Gamma_G)\subseteq\mathcal{G}(\Gamma_G)$) can be obtained by an appropriate probability distribution over this set of ontic states. Denoting the set of ontic states corresponding to $\mathcal{G}(\Gamma_G)|_{\rm det}$ by $\Lambda_{\rm det}$ and the set of ontic states corresponding to $\mathcal{G}(\Gamma_G)|_{\rm ind}$ by $\Lambda_{\rm ind}$, we have that the measurement noncontextual ontological model given by $\Lambda\equiv\Lambda_{\rm det}\sqcup\Lambda_{\rm ind}$ reproduces the predictions $\mathcal{T}(\Gamma_G)$ of any operational theory $\mathbb{T}$ that admits a measurement noncontextual ontological model: that is, for every $p\in\mathcal{G}(\Gamma_G)$ (and therefore also $p\in\mathcal{T}(\Gamma_G)$), $$p(v)=\sum_{\lambda\in\Lambda}\xi(v|\lambda)\mu(\lambda)$$ for all $v\in V(\Gamma_G)$,  for some probability distribution $\mu:\Lambda\rightarrow [0,1]$ such that $\sum_{\lambda\in\Lambda}\mu(\lambda)=1$.\footnote{As a corollary, note that as long as the polytope $\mathcal{G}(\Gamma_G)$ has a finite number of extreme points, we can take the ontic state space to consist of a finite number of ontic states (as we have done) without any loss of generality. The hypergraphs $\Gamma_G$ we study -- representing the measurement events of interest in a contextuality experiment -- have this property because of their finiteness.} We can also then rewrite $\beta(\Gamma_G,q)$ as 
\begin{equation}\label{betaontl}
\beta(\Gamma_G,q)=\max_{\lambda\in\Lambda_{\rm ind}}\sum_{e\in E(\Gamma_G)}q_e\zeta(M_e,\lambda),
\end{equation}
where $\zeta(M_e,\lambda)\equiv\max_{m_e}\xi(m_e|M_e,\lambda)$.
\section{Noise-robust noncontextuality inequalities}\label{derivingineqs}
We will now proceed to obtain our noise-robust noncontextuality inequalities following the ideas outlined in Ref.~\cite{KunjSpek17}. 
\subsection{Key notions from CSW}
We first recall some key notions from the CSW framework \cite{CSW} before obtaining our inequalities.

Consider the positive linear combination of the probabilities of measurement events,
\begin{equation}
R([s|S])\equiv\sum_{v\in V(G)}w_v p(v|S,s),
\end{equation}
where $w_v>0$ for all $v\in V(G)$.

The fundamental result of CSW is that this quantity is bounded for different sets of correlations --- KS-noncontextual, those realizable by projective quantum measurements, and those satisfying consistent exclusivity --- by graph-theoretic invariants as follows:
\begin{equation}
\forall [s|S]: R([s|S])\overset{\rm KS}{\leq} \alpha(G,w)\overset{\rm Q}{\leq} \theta(G,w)\overset{{\rm CE}^1}{\leq} \alpha^*(G,w),
\end{equation}
where ${\rm KS}$ denotes operational theories that admit KS-noncontextual ontological models and thus realize probabilistic models on $\Gamma_G$ that fall in the set $\mathcal{C}(\Gamma_G)$, ${\rm Q}$ denotes quantum theory with projective measurements which assigns probabilistic models on $\Gamma_G$ denoted by $\mathcal{Q}(\Gamma_G)$, and ${\rm CE}^1$ denotes operational theories satisfying consistent exclusivity and thus realizing the set of probabilistic models $\mathcal{CE}^1(\Gamma_G)$ on $\Gamma_G$. The graph invariants of the weighted graph $(G,w)$, namely, $\alpha(G,w)$, $\theta(G,w)$, and $\alpha^*(G,w)$ are defined as follows:
\begin{enumerate}
	\item Independence number $\alpha(G,w)$:
	\begin{equation}
	\alpha(G,w)\equiv\max_{\mathcal{I}}\sum_{v\in\mathcal{I}}w_v,
	\end{equation}
	where $\mathcal{I}\subseteq V(G)$ is an {\em independent set} of vertices of $G$, i.e., a set of nonadjacent vertices of $G$, so that none of the vertices in this set shares an edge with any other vertex in the set.
	
	\item Lovasz theta number $\theta(G,w)$:
	\begin{equation}
	\theta(G,w)\equiv\max_{\{|u_v\rangle\}_{v\in V(G)},|\psi\rangle}\sum_{v\in V(G)}w_v|\langle\psi|u_v\rangle|^2,
	\end{equation}
	where $\{|u_v\rangle\}_{v\in V(G)}=\{|u_v\rangle\}_{v\in V(\bar{G})}$ (each $|u_v\rangle$ a unit vector in $\mathbb{R}^d$) is an {\em orthonormal representation} (OR) of the complement of $G$, namely, $\bar{G}$, and the unit vector $|\psi\rangle\in \mathbb{R}^d$ is called a {\em handle}.	
	
	Here $V(\bar{G})\equiv V(G)$ and $E(\bar{G})\equiv\{(v,v')|v,v'\in V(G), (v,v')\notin E(G)\}$, and we have in an orthonormal representation that $\langle u_{v''}|u_{v'''}\rangle=0$
	for all pairs of nonadjacent vertices, $(v'',v''')$, in $\bar{G}$, or equivalently, for all $(v'',v''')\in E(G)$.
	
	\item Fractional packing number $\alpha^*(G,w)$:
	\begin{equation}
	\alpha^*(G,w)\equiv\max_{\{p_v\}_{v\in V(G)}}\sum_{v\in V(G)}w_vp_v,
	\end{equation}
	where $\{p_v\}_{v\in V(G)}$ is such that $p_v\geq 0$ for all $v\in V(G)$ and $\sum_{v\in c}p_v\leq 1$ for all cliques $c$ in $G$.
	
\end{enumerate}

Note that since we are always considering $\Gamma_G$ such that $\mathcal{CE}^1(\Gamma_G)=\mathcal{G}(\Gamma_G)$, we, in fact, have the bounds

\begin{equation}
\forall [s|S]: R([s|S])\overset{\rm KS}{\leq} \alpha(G,w)\overset{\rm Q}{\leq} \theta(G,w)\overset{{\rm GPT}}{\leq} \alpha^*(G,w),
\end{equation}
where ``GPT" denotes the full set of probabilistic models on $\Gamma_G$, i.e., $\mathcal{G}(\Gamma_G)$.

In terms of the notation we have already introduced, where $R([s|S])\leq R_{\rm KS}$ was a Bell-KS inequality, we now have --- from CSW \cite{CSW} --- that $R_{\rm KS}=\alpha(G,w)$.
\subsection{Key notion {\em not} from CSW:\\source-measurement correlation, ${\rm Corr}$}\label{definingCorr}
We need to define a new quantity not in the CSW framework, namely,
\begin{equation}\label{eq:Corr}
{\rm Corr}\equiv\sum_{e\in E(\Gamma_G)}q_e\sum_{m_e,s_e}\delta_{m_e,s_e}p(m_e,s_e|M_e,S_e),
\end{equation}
where $\{q_e\}_{e\in E(\Gamma_G)}$ is a probability distribution, i.e., $q_e\geq0$ for all $e\in E(\Gamma_G)$ and $\sum_{e\in E(\Gamma_G)}q_e=1$, such that $\beta(\Gamma_G,q)<1$ holds.\footnote{\label{strongestpossible}Indeed, for the strongest possible constraint on ${\rm Corr}$, one must pick $q$ such that $\beta(\Gamma_G,q)$ is minimized.} In previous work \cite{KunjSpek,KunjSpek17}, we have taken $q$ to be the uniform distribution $q_e=\frac{1}{|E(\Gamma_G)|}$, but the derivation of the noncontextuality inequalities is independent of that choice (as we'll see here). Also, note that we have chosen the following labelling convention for outcomes of source setting $S_e$ (namely, $s_e$) and measurement setting $M_e$ (namely, $m_e$): the source outcomes $s_e$ for source setting $S_e$ take values in the same set as measurement outcomes $m_e$ for measurement setting $M_e$, i.e., $V_{S_e}=V_{M_e}$ (recalling notation from Section \ref{spekframework}). In particular, outcomes corresponding to the measurement event $[v|e]$ (representing $[m_e|M_e]$) and its corresponding source event $v_e$ (representing $[s_e|S_e]$) are both denoted by the same label, so that $m_e=s_e$ for them. An example of this from Figs.~\ref{fig4} and \ref{source} would be to, say, denote the outcomes of a particular $e\in E(\Gamma_G)$ (measurement setting $M_e$) by $m_e\in V_{M_e}\equiv\{0,1,2\}$ and corresponding outcomes of $e\in E(\Sigma_G)$ (source setting $S_e$) by $s_e\in V_{S_e}\equiv\{0,1,2\}$; so if $[v|e]$ denotes $[m_e=0|M_e]$, then $v_e$ will denote $[s_e=0|S_e]$, etc. 

\subsection{Obtaining the noise-robust noncontextuality inequalities}
\subsubsection{Expressing operational quantities in ontological terms}
We begin with expressing the operational quantities of interest in terms of a noncontextual ontological model. In an ontological model, $R([s|S])$ is given by
\begin{equation}
R([s|S])=\sum_{\lambda\in\Lambda}\sum_{v\in V(G)}w_vp(v|\lambda)\mu(\lambda|S,s).
\end{equation}
Defining $R(\lambda)\equiv\sum_{v\in V(G)}w_vp(v|\lambda)$, we have that
\begin{equation}
R([s|S])=\sum_{\lambda\in\Lambda}R(\lambda)\mu(\lambda|S,s).
\end{equation}
Similarly, ${\rm Corr}$ is given by 
\begin{eqnarray}
&&{\rm Corr}\nonumber\\
&=&\sum_{\lambda\in\Lambda}\sum_{e\in E(\Gamma_G)}q_e\sum_{m_e,s_e}\delta_{m_e,s_e}\xi(m_e|M_e,\lambda)\mu(\lambda,s_e|S_e)\nonumber\\
&=&\sum_{\lambda\in\Lambda}\sum_{e\in E(\Gamma_G)}q_e\nonumber\\&&\sum_{m_e,s_e}
\delta_{m_e,s_e}\xi(m_e|M_e,\lambda)\mu(s_e|S_e,\lambda)\mu(\lambda|S_e).\nonumber\\
\end{eqnarray}
Here, we have used the fact that $$\mu(\lambda,s_e|S_e)=\mu(s_e|S_e,\lambda)\mu(\lambda|S_e)$$ 
to express ${\rm Corr}$ in a way that treats sources and measurements similarly.

Using {\em preparation noncontextuality} (cf.~Eq.~\eqref{ncprep}), we have that
\begin{eqnarray}
\forall e,e'&\in& E(\Sigma_G):[\top|S_{e_{\top}}] \simeq[\top|S_{{e'}_{\top}}]\nonumber\\
\Rightarrow \mu(\lambda|S_e)&=&\mu(\lambda|S_{e'})\equiv \nu(\lambda), \forall\lambda\in\Lambda.
\end{eqnarray}
Then we can rewrite ${\rm Corr}$ as
\begin{eqnarray}
&&{\rm Corr}\nonumber\\
&=&\sum_{\lambda\in\Lambda}\sum_{e\in E(\Gamma_G)}q_e\sum_{m_e,s_e}\delta_{m_e,s_e}\xi(m_e|M_e,\lambda)\mu(s_e|S_e,\lambda)\nu(\lambda).\nonumber\\
\end{eqnarray}
Note that the only $\lambda$ that contribute to ${\rm Corr}$ are those for which $\nu(\lambda)>0$. Also, $\mu(s_e|S_e,\lambda)$ and $\mu(\lambda|S_e,s_e)$ satisfy the condition $\mu(s_e|S_e,\lambda)\nu(\lambda)=\mu(\lambda|S_e,s_e)p(s_e|S_e)$, so that $\mu(s_e|S_e,\lambda)$ is well-defined whenever $\nu(\lambda)>0$.

Defining 
\begin{equation}
{\rm Corr}(\lambda)\equiv\sum_{e\in E(\Gamma_G)}q_e\sum_{m_e,s_e}\delta_{m_e,s_e}\xi(m_e|M_e,\lambda)\mu(s_e|S_e,\lambda),
\end{equation}
we have that 
\begin{equation}
{\rm Corr}=\sum_{\lambda\in\Lambda}{\rm Corr}(\lambda)\nu(\lambda),
\end{equation}

Recalling that $\zeta(M_e,\lambda)=\max_{m_e}\xi(m_e|M_e,\lambda)$, note that ${\rm Corr}(\lambda)$ is upper bounded as follows (for any $\lambda\in\Lambda$):
\begin{eqnarray}
&&{\rm Corr}(\lambda)\nonumber\\
&\equiv&\sum_{e\in E(\Gamma_G)}q_e\sum_{m_e,s_e}\delta_{m_e,s_e}\xi(m_e|M_e,\lambda)\mu(s_e|S_e,\lambda)\nonumber\\
&\leq&\sum_{e\in E(\Gamma_G)}q_e\zeta(M_e,\lambda)\sum_{s_e}\mu(s_e|S_e,\lambda)\nonumber\\
&=&\sum_{e\in E(\Gamma_G)}q_e\zeta(M_e,\lambda).
\end{eqnarray}
 If $\lambda\in\Lambda_{\rm det}$, then this upper bound is trivial, i.e., ${\rm Corr}(\lambda)\leq 1$, since every measurement has deterministic response functions. On the other hand, for all $\lambda\in\Lambda_{\rm ind}$, we have (from Eq.~\eqref{betaontl})
\begin{equation}\label{Corrupperbound}
{\rm Corr}(\lambda)\leq\beta(\Gamma_G,q).
\end{equation}
Similarly, for $\lambda\in\Lambda_{\rm det}$ we have $R(\lambda)\leq \alpha(G,w)$, while for $\lambda\in\Lambda_{\rm ind}$ we have $R(\lambda)\leq\alpha^*(G,w)$.

Using the fact that $$\nu(\lambda)=\mu(\lambda|S)=\sum_{s}\mu(\lambda|S,s)p(s|S),$$
for any $S\equiv S_e, e\in E(\Sigma_G)$, 
we have 
\begin{eqnarray}
&&{\rm Corr}\nonumber\\
&=&\sum_{s}\left(\sum_{\lambda}{\rm Corr}(\lambda)\mu(\lambda|S,s)\right)p(s|S)\nonumber\\
&=&\sum_{s}{\rm Corr}_{s}p(s|S).
\end{eqnarray}
where we have defined ${\rm Corr}_{s}\equiv \sum_{\lambda}{\rm Corr}(\lambda)\mu(\lambda|S,s)$.

\subsubsection{Derivation of the noncontextual tradeoff for any graph $G$}
We are now in a position to express our general noise-robust noncontextuality inequality as a tradeoff between three operational quantities: ${\rm Corr}$, $R([s_{e_*}=0|S_{e_*}])$, and $p(s_{e_*}=0|S_{e_*})$.

First, note that KS-contextuality is witnessed when for some choice of $[s|S]$, here given by $[s_{e_*}=0|S_{e_*}]$, we have $$R([s_{e_*}=0|S_{e_*}])>\alpha(G,w).$$ This means that for some set of ontic states in the support of $[s_{e_*}=0|S_{e_*}]$, i.e., 
\begin{align}
\lambda&\in{\rm Supp}\{\mu(.|S_{e_*},s_{e_*}=0)\}\nonumber\\
&\equiv\{\lambda\in\Lambda:\mu(\lambda|S_{e_*},s_{e_*}=0)>0\},
\end{align} 
we have $R(\lambda)>\alpha(G,w)$. For such a set of ontic states one must then have ${\rm Corr}(\lambda)<1$ (because these $\lambda\in\Lambda_{\rm ind}$ and we have Eq.~\eqref{Corrupperbound}), which in turn implies that ${\rm Corr}_{s_{e_*}=0}<1$.
On the other hand, for $s_{e_*}=1$, we have no constraints: ${\rm Corr}_{s_{e_*}=1}\leq 1$. Thus,
\begin{eqnarray}
&&{\rm Corr}\nonumber\\
&=&{\rm Corr}_{s_{e_*}=0}p(s_{e_*}=0|S_{e_*})+{\rm Corr}_{s_{e_*}=1}p(s_{e_*}=1|S_{e_*})\nonumber\\
&\leq& p_0 {\rm Corr}_{s_{e_*}=0}+1-p_0,
\end{eqnarray}
where $p_0\equiv p(s_{e_*}=0|S_{e_*})$. 

Defining $\mu_{\rm det}\equiv\sum_{\lambda\in\Lambda_{\rm det}}\mu(\lambda|S_{e_*},s_{e_*}=0)$ and $\mu_{\rm ind}\equiv\sum_{\lambda\in\Lambda_{\rm ind}}\mu(\lambda|S_{e_*},s_{e_*}=0)$, we now have
\begin{equation}\label{cons1}
\mu_{\rm det}+\mu_{\rm ind}=1,
\end{equation}
\begin{equation}\label{cons2}
{\rm Corr}_{s_{e_*}=0}\leq \mu_{\rm det}+\beta(\Gamma_G,q)\mu_{\rm ind},
\end{equation}
\begin{equation}\label{cons3}
R\leq \alpha(G,w)\mu_{\rm det}+\alpha^*(G,w)\mu_{\rm ind}.
\end{equation}
Note that assuming $\mu_{\rm det}=1$ would reduce these constraints to a standard Bell-KS inequality, $R\leq \alpha(G,w)$. However, since we are not assuming this, simply eliminating $\mu_{\rm det}$ and $\mu_{\rm ind}$ from these constraints leads us to\footnote{To see this explicitly, just use Eq.~\eqref{cons1} to make the substitution $\mu_{\rm ind}=1-\mu_{\rm det}$ in Eqs.~\eqref{cons2} and \eqref{cons3}, then eliminate $\mu_{\rm det}$ from Eq.~\eqref{cons2} by using the upper bound on it from Eq.~\eqref{cons3}.}
\begin{eqnarray}
&&{\rm Corr}_{s_{e_*}=0}\nonumber\\
&\leq&1-(1-\beta(\Gamma_G,q))\frac{R-\alpha(G,w)}{\alpha^*(G,w)-\alpha(G,w)}\nonumber\\
\end{eqnarray}
where  the upper bound is nontrivial if and only if $\beta(\Gamma_G,q)<1$ and $R-\alpha(G,w)>0$. 

If we are given that $\beta(\Gamma_G,q)<1$, then we have a trivial upper bound on ${\rm Corr}_{s_{e_*}=0}$ for the remaining cases: the upper bound is $1$ for $R=\alpha(G,w)$ and greater than $1$ for $R<\alpha(G,w)$.

Thus, our noise-robust noncontextuality inequality now reads:
\begin{equation}\label{NCI1}
{\rm Corr}\leq 1-p_0(1-\beta(\Gamma_G,q))\frac{R-\alpha(G,w)}{\alpha^*(G,w)-\alpha(G,w)},
\end{equation}
which can be rewritten as
\begin{equation}\label{NCI2}
R\leq \alpha(G,w)+\frac{\alpha^*(G,w)-\alpha(G,w)}{p_0}\frac{1-{\rm Corr}}{1-\beta(\Gamma_G,q)}.
\end{equation}

Note that Eq.~\eqref{NCI1} expresses the constraint from noncontextuality as an upper bound on the source-measurement correlations ${\rm Corr}$, reminiscent of the noise-robust noncontextuality inequality first derived in Ref.~\cite{KunjSpek} (and later treated in hypergraph-theoretic terms in Ref.~\cite{kunjunc}), except here the upper bound on ${\rm Corr}$ depends not only on the hypergraph invariant $\beta(\Gamma_G,q)$ but also two of the graph invariants from the CSW framework \cite{CSW}, namely, $\alpha(G,w)$ and $\alpha^*(G,w)$, besides also the operational quantity $R$, which is the figure-of-merit for KS-contextuality ($R>\alpha(G,w)$ witnesses KS-contextuality) in the CSW framework. Eq.~\eqref{NCI1} indicates that the source-measurement correlations would fail to be perfect (i.e., ${\rm Corr}<1$) in an operational theory admitting a noncontextual ontological model if and only if $R>\alpha(G,w)$ and $\beta(\Gamma_G,q)<1$. Contextuality would be witnessed when the source-measurement correlations are stronger than the constraint from Eq.~\eqref{NCI1}. For $R\leq\alpha(G,w)$, in particular, there is no constraint from noncontextuality on ${\rm Corr}$.

On the other hand, rewriting the constraint from noncontextuality as Eq.~\eqref{NCI2}, one is reminded of the CSW framework \cite{CSW}, where $R$ is taken to be the quantity that is upper bounded by KS-noncontextuality. Here, instead, we have that $R$ is upper bounded by a term that includes the source-measurement correlations ${\rm Corr}$ that can be achieved for the measurements and thus penalizes for measurements that cannot be made highly predictable with respect to some preparations, i.e., ${\rm Corr}<1$ makes it harder to violate the upper bound on $R$. When the upper bound reaches $\alpha^*(G,w)$, it becomes trivial and $R$ is no longer constrained by noncontextuality on account of noise in the measurements. Indeed, trivial POVMs (cf.~Appendices \ref{trivialpovms} and \ref{trivial}) never violate such a noncontextuality inequality because of the penalty incurred via ${\rm Corr}$, as we later show in Section \ref{trivialnoviol}.
\subsubsection{When is the noncontextual tradeoff violated?}

The inequality of Eq.~\eqref{NCI2} can be rewritten as the following tradeoff between ${\rm Corr}$, $p_0$, and $R$:

\begin{equation}\label{NCI3}
{\rm Corr}+p_0(1-\beta(\Gamma_G,q))\frac{R-\alpha(G,w)}{\alpha^*(G,w)-\alpha(G,w)}\leq 1.
\end{equation}
Writing the constraint from noncontextuality in the form of Eq.~\eqref{NCI3} (in contrast to Eqs.~\eqref{NCI1} and \eqref{NCI2}) makes it more even-handed in its treatment of the two operational quantities $R$ (which is key in the CSW framework \cite{CSW}) and ${\rm Corr}$ (which is key in noise-robust noncontextuality inequalities inspired by logical proofs of the KS theorem \cite{KunjSpek,kunjunc}) and emphasizes that noise-robust noncontextuality inequalities inspired by statistical proofs of the KS theorem \cite{KunjSpek17} are  tradeoffs between $R$ (which is about the strength of correlations {\em between} measurements) and ${\rm Corr}$ (which is about the predictability of measurements) that must be satisfied by any operational theory admitting a noncontextual ontological model. Roughly speaking, a high degree of predictability for measurements (e.g., ${\rm Corr}=1$) cannot coexist with very strong correlations {\em between} the measurements (e.g., $R=\alpha^*(G,w)$) when the operational theory admits a noncontextual ontological model.

For a nontrivial constraint -- and hence, the possibility of witnessing contextuality via violation of this inequality (Eq.~\eqref{NCI3}) -- the upper bound on ${\rm Corr}$ (the right-hand-side of Eq.~\eqref{NCI1}) should be strictly bounded above by 1, and the upper bound on $R$ (the right-hand-side of Eq.~\eqref{NCI2}) should be strictly bounded above by $\alpha^*(G,w)$ (the algebraic upper bound on $R$), that is
\begin{eqnarray}\label{conditions}
p_0&>&0\textrm{ and }\beta(\Gamma_G,q)<1,\nonumber\\
R&>&\alpha(G,w),\nonumber\\
{\rm Corr}&>&1-p_0(1-\beta(\Gamma_G,q)).
\end{eqnarray}

These are the minimal benchmarks necessary --- besides the requirement of tomographic completeness of a finite set of procedures and the possibility of inferring secondary procedures with exact operational equivalences using convexity of the operational theory \cite{exptlpaper} --- to witness contextuality in a Kochen-Specker type experiment adapted to our framework following Spekkens \cite{Spe05}.

Suppose one achieves, by some means, a value of $R=\theta(G,w)$, the upper bound on the quantum value with projective measurements. When would this value be an evidence of contextuality? For this to be the case, we must have:
\begin{equation}
{\rm Corr}>1-p_0(1-\beta(\Gamma_G,q))\frac{\theta(G,w)-\alpha(G,w)}{\alpha^*(G,w)-\alpha(G,w)}.
\end{equation}
Now, for the ideal quantum realization where measurement events are projectors, and the corresponding source events are eigenstates, it is always the case that ${\rm Corr}=1$, hence contextuality is witnessed. However, it's possible to witness contextuality even if ${\rm Corr}<1$, as long as it exceeds the lower bound specified above. In a sense, for quantum theory, this allows for a quantitative accounting of the effect of nonprojectiveness in the measurements (or mixedness in preparations) on the possibility of witnessing contextuality, a feature that is absent in traditional Kochen-Specker approaches \cite{CSW10, CSW, AFLS, sheaf}. Indeed, as long as one achieves any value of $R>\alpha(G,w)$, it is possible to witness contextuality for a sufficiently high value of ${\rm Corr}$ (see Eq.~\eqref{NCI1}).
\subsection{Example: KCBS scenario}
We will now illustrate our hypergraph framework by applying it to the KCBS scenario to make differences with respect to the CSW graph-theoretic framework \cite{CSW} explicit.

The graph $G$ for the KCBS scenario is given in Fig.~\ref{fig3}, the measurement events hypergraph $\Gamma_G$ is given in Fig.~\ref{fig4}, and the source events hypergraph $\Sigma_G$ is given in Fig.~\ref{source}. We then have
\begin{equation}
R([s|S])=\sum_{v\in V(G)}p(v|S,s),
\end{equation}
where the (vertex) weights $w_v=1$ for all $v\in V(G)$, i.e., it's an unweighted graph and we will use $\alpha(G)$ and $\alpha^*(G)$ to denote its independence number and the fractional packing number, respectively. These are given by 
\begin{equation}
\alpha(G)=2\quad{\rm and}\quad \alpha^*(G)=5/2.
\end{equation}
The source-measurement correlation term is given by 
\begin{equation}
{\rm Corr}=\sum_{e\in E(\Gamma_G)}q_e\sum_{m_e,s_e}\delta_{m_e,s_e}p(m_e,s_e|M_e,S_e)
\end{equation}
for any choice of probability distribution $q\equiv\{q_e\}_{e\in E(\Gamma_G)}$. For simplicity, we will just take this probability distribution to be uniform, i.e., $q_e=\frac{1}{5}$ for all $e\in E(\Gamma_G)$. Note that the only extremal probabilistic model on $\Gamma_G$ corresponding to an indeterministic assignment (in $\Lambda_{\rm ind}$) assigns $\xi(v|\lambda)=\frac{1}{2}$ for all $v\in V(G)$. This means 
\begin{equation}
\beta(\Gamma_G,q)=\frac{1}{2} \quad\forall q.
\end{equation}

The noncontextuality inequality of Eq.~\eqref{NCI2} 
\begin{equation}
R\leq \alpha(G,w)+\frac{\alpha^*(G,w)-\alpha(G,w)}{p_0}\frac{1-{\rm Corr}}{1-\beta(\Gamma_G,q)}
\end{equation}
then becomes (in the KCBS scenario)
\begin{equation}
R\leq 2+\frac{1/2}{p_0}\frac{1-{\rm Corr}}{1/2},
\end{equation}
or 
\begin{equation}
R\leq 2+\frac{1-{\rm Corr}}{p_0}.
\end{equation}
Recall that the KCBS inequality \cite{KCBS, CSW} reads $R\leq 2$ and it would be a valid noncontextuality inequality in our framework if and only if one can find measurements and preparations such that ${\rm Corr}=1$. In the standard KCBS construction \cite{KCBS} that violates the inequality $R\leq 2$, we have the five vertices in $G$ (say $v_i$, $i\in\{1,2,3,4,5\}$, labelled cyclically) associated with five projectors $\Pi_i=|l_i\rangle\langle l_i|$, $i\in\{1,2,3,4,5\},$ on a qutrit Hilbert space, given by the vectors 
$|l_i\rangle=(\sin\theta \cos\phi_i,\sin\theta \sin\phi_i,\cos\theta)$, $\phi_i=\frac{4\pi i}{5}$, and $\cos\theta=\frac{1}{\sqrt[4]{5}}$. The special source event $[s_{e_*}=0|S_{e_*}]$ is associated with the quantum state $|\psi\rangle=(0,0,1)$, so that
\begin{equation}
R=\sum_{i=1}^5 |\langle l_i|\psi\rangle|^2=\sqrt{5}>2.
\end{equation}
\begin{figure}
	\includegraphics[scale=0.3]{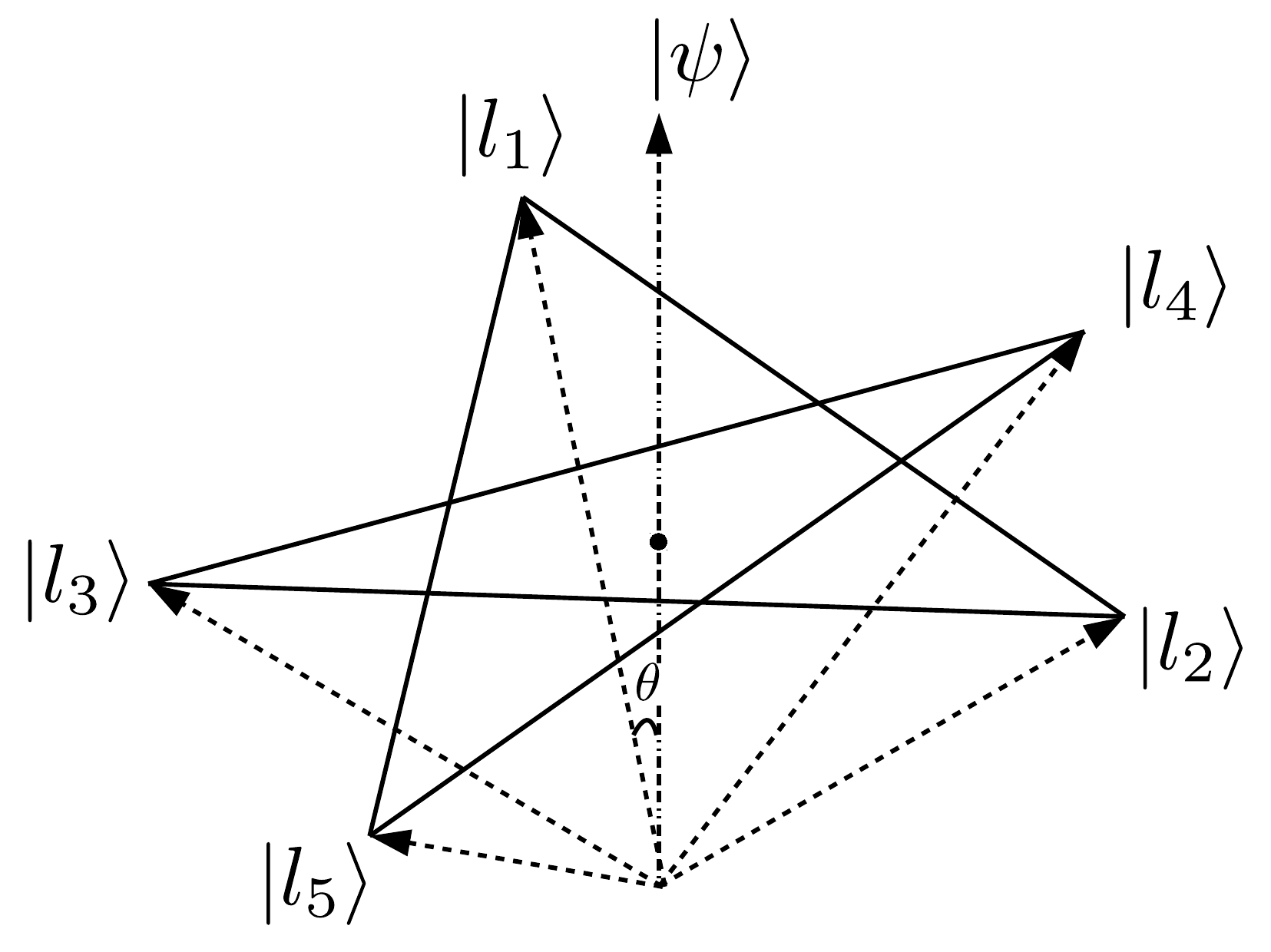}
	\caption{Geometric configuration of the vectors appearing in the KCBS construction \cite{KCBS}.}
	\label{kcbsfig}
\end{figure}
See Fig.~\ref{kcbsfig} for a depiction of the geometric configuration of these vectors. 

To turn this KCBS construction into an argument against noncontextuality in our approach, we need additional ingredients beyond the graph $G$. Firstly, for both the measurement events hypergraph $\Gamma_G$ and the source events hypergraph $\Sigma_G$, we denote the hyperedges by $e_i$, $i\in\{1,2,3,4,5\}$. In $\Gamma_G$, the measurement events for the setting $M_{e_i}$ are given by $\{[m_{e_i}=0|M_{e_i}]=|l_i\rangle\langle l_i|, [m_{e_i}=1|M_{e_i}]=\mathbb{I}-|l_i\rangle\langle l_i|-|l_{i+1}\rangle\langle l_{i+1}|, [m_{e_i}=2|M_{e_i}]=|l_{i+1}\rangle\langle l_{i+1}|\}$, where for $i=5$, $i+1=1$ (addition modulo 5). Similarly, in $\Sigma_G$, the source events corresponding to source setting $S_{e_i}$ are given by $\{[s_{e_i}=0|S_{e_i}]=|l_i\rangle\langle l_i|$, $[s_{e_i}=2|S_{e_i}]=|l_{i+1}\rangle\langle l_{i+1}|$, and $[s_{e_i}=1|S_{e_i}]=\mathbb{I}-|l_i\rangle\langle l_i|-|l_{i+1}\rangle\langle l_{i+1}|\}$, where $p(s_{e_i}=b|S_{e_i})=\frac{1}{3}$ for all $b\in\{0,1,2\}$. The special source setting $S_{e_*}$ consists of source events $\{[s_{e_*}=0|S_{e_*}]=|\psi\rangle\langle\psi|,[s_{e_*}=1|S_{e_*}]=\frac{\mathbb{I}-|\psi\rangle\langle\psi|}{2}\}$, where $p(s_{e_*}=0|S_{e_*})=\frac{1}{3}$ and $p(s_{e_*}=1|S_{e_*})=\frac{2}{3}$. We thus have the operational equivalences we need between the source settings: 
\begin{equation}
[\top|S_{e_{\top}}]\simeq[\top|S_{{e'}_{\top}}]=\frac{\mathbb{I}}{3},\quad\forall e,e'\in E(\Sigma_G).
\end{equation}
This choice of representation for $\Gamma_G$ and $\Sigma_G$ yields $p_0=\frac{1}{3}$, ${\rm Corr}=1$, and $R([s_{e_*}=0|S_{e_*}])=\sqrt{5}$, so that the inequality 
\begin{equation}
R\leq 2+\frac{1-{\rm Corr}}{p_0}
\end{equation}
is violated. However, note that this is an idealization (under which ${\rm Corr}=1$) and, typically, the source events and measurement events will not be perfectly correlated (${\rm Corr}<1$) and the operational equivalences between the source settings need not correspond to the maximally mixed state. All that is required for a test of noncontextuality using this inequality is that the operational equivalences hold for {\em some} choice of preparations and measurements which need not be the same as that in the ideal KCBS construction. 

To illustrate what happens when ${\rm Corr}<1$, we consider the effect of a depolarizing channel on the states and measurements in the ideal KCBS construction. The channel is given by 
	\begin{align}
	\mathcal{D}_r(\cdot)=rI(\cdot)\mathbb{I}+(1-r)\frac{1}{3}\mathbb{I}\Tr(\cdot), r\in[0,1].
	\end{align}
The action of this channel -- with parameter $r_1\in[0,1]$, say -- on the pure states $\{\{|l_i\rangle\langle l_i|\}_{i=1}^5, |\psi\rangle\langle\psi|\}$ yields the noisy states given by
\begin{align}
&\mathcal{D}_{r_1}(|l_i\rangle\langle l_i|)={r_1}|l_i\rangle\langle l_i|+(1-{r_1})\frac{\mathbb{I}}{3}, \forall i\in[5],\\
&\mathcal{D}_{r_1}(|\psi\rangle\langle \psi|)={r_1}|\psi\rangle\langle \psi|+(1-{r_1})\frac{\mathbb{I}}{3},
\end{align}
and the action of its adjoint -- with parameter $r_2\in[0,1]$, say -- on the ideal projectors, $\{|l_i\rangle\langle l_i|\}_{i=1}^5$, involved in the measurements correspondingly yields the POVM elements given by 
\begin{align}
\mathcal{D}^{\dagger}_{r_2}(|l_i\rangle\langle l_i|)={r_2}|l_i\rangle\langle l_i|+(1-{r_2})\frac{\mathbb{I}}{3}, \forall i\in[5].
\end{align}
Hence, we are imagining a situation where the preparation procedures are affected by depolarizing noise with parameter $r_1$ and measurement procedures are affected by depolarizing noise with parameter $r_2$, similar to the situation considered previously in Section II of the Supplemental material of Ref.~\cite{KunjSpek}. The operational equivalences required for our argument from preparation and measurement noncontextuality are satisfied by these noisy preparations and measurements. That is, in $\Gamma_G$, we can represent the measurement events for the setting $M_{e_i}$ (where $i\in[5]$, $i+1=1$, i.e., addition modulo 5) by
\begin{align}
&[m_{e_i}=0|M_{e_i}]=\mathcal{D}^{\dagger}_{r_2}(|l_i\rangle\langle l_i|),\\ &[m_{e_i}=1|M_{e_i}]=\mathcal{D}^{\dagger}_{r_2}(\mathbb{I}-|l_i\rangle\langle l_i|-|l_{i+1}\rangle\langle l_{i+1}|),\\
&[m_{e_i}=2|M_{e_i}]=\mathcal{D}^{\dagger}_{r_2}(|l_{i+1}\rangle\langle l_{i+1}|).
\end{align}
It is easy to verify that these form elements of a valid POVM denoted by the measurement setting $M_{e_i}$ and that the operational equivalences between the measurement events (represented by $\Gamma_G$) are indeed respected.
On the other hand, in $\Sigma_G$, the source events corresponding to source setting $S_{e_i}$ can be represented by 
\begin{align}
&[s_{e_i}=0|S_{e_i}]=\mathcal{D}_{r_1}(|l_i\rangle\langle l_i|),\\ &[s_{e_i}=1|S_{e_i}]=\mathcal{D}_{r_1}(\mathbb{I}-|l_i\rangle\langle l_i|-|l_{i+1}\rangle\langle l_{i+1}|\})\\
&[s_{e_i}=2|S_{e_i}]=\mathcal{D}_{r_1}(|l_{i+1}\rangle\langle l_{i+1}|),
\end{align}
where $p(s_{e_i}=b|S_{e_i})=\frac{1}{3}$ for all $b\in\{0,1,2\}$, while the source events for source setting $S_{e_*}$ can be represented by 
\begin{align}
&[s_{e_*}=0|S_{e_*}]=\mathcal{D}_{r_1}(|\psi\rangle\langle\psi|),\\
&[s_{e_*}=1|S_{e_*}]=\mathcal{D}_{r_1}\left(\frac{\mathbb{I}-|\psi\rangle\langle\psi|}{2}\right),
\end{align}
where $p(s_{e_*}=0|S_{e_*})=\frac{1}{3}$ and $p(s_{e_*}=1|S_{e_*})=\frac{2}{3}$. These satisfy the operational equivalences
\begin{equation}
[\top|S_{e_{\top}}]\simeq[\top|S_{{e'}_{\top}}]=\frac{\mathbb{I}}{3},\quad\forall e,e'\in E(\Sigma_G).
\end{equation}
We then have 
\begin{align}
&{\rm Corr}=\frac{1}{5}\sum_{e\in E(\Gamma_G)}\frac{1}{3}\sum_{b\in\{0,1,2\}}p(m_e=b|M_e,S_e,s_e=b).
\end{align}
Noting that for any qutrit pure state $|\phi\rangle$ and its corresponding projector $|\phi\rangle\langle \phi|$, each affected by depolarizing noise with parameters $r_1$ and $r_2$, respectively, we have
\begin{align}\label{noisyoverlap}
&\Tr(\mathcal{D}_{r_1}(|\phi\rangle\langle \phi|)\mathcal{D}^{\dagger}_{r_2}(|\phi\rangle\langle \phi|))\nonumber\\
=&\frac{1}{3}+\frac{2}{3}r_1r_2.
\end{align}
Now, each term in the summation defining ${\rm Corr}$, namely, $p(m_e=b|M_e,S_e,s_e=b)$, is obtained from a calculation of the type in Eq.~\eqref{noisyoverlap}. Hence, we have for each such term, 
\begin{align}
&p(m_e=b|M_e,S_e,s_e=b)\nonumber\\
=&\frac{1}{3}+\frac{2}{3}r_1r_2,
\end{align}
so that
\begin{align}
{\rm Corr}=\frac{1}{3}+\frac{2}{3}r_1r_2.
\end{align}
In the noiseless regime, i.e., $r_1=r_2=1$, this reduces to the ideal KCBS scenario. On the other hand, we have 
\begin{widetext}
\begin{align}
&R([s_{e_*}=0|S_{e_*}])\nonumber\\
=&\sum_{v\in V(G)}p(v|S_{e_*},s_{e_*}=0)\nonumber\\
=&\sum_{i=1}^5 \Tr\left(\mathcal{D}^{\dagger}_{r_2}(|l_i\rangle\langle l_i|)\mathcal{D}_{r_1}(|\psi\rangle\langle\psi|)\right)\nonumber\\
=&\sum_{i=1}^5\left(r_1r_2|\langle l_i|\psi\rangle|^2+\frac{r_2(1-r_1)}{3}+\frac{r_1(1-r_2)}{3}+\frac{(1-r_1)(1-r_2)}{3}\right)\nonumber\\
=&r_1r_2\sum_{i=1}^5|\langle l_i|\psi\rangle|^2+\frac{5}{3}(1-r_1r_2).
\end{align}
\end{widetext}
Recall that violation of the noncontextuality inequality requires that
\begin{equation}
R>2+\frac{1-{\rm Corr}}{p_0}.
\end{equation}
That is,
\begin{align}
&r_1r_2\sum_{i=1}^5|\langle l_i|\psi\rangle|^2+\frac{5}{3}(1-r_1r_2)\nonumber\\
>&2+3\left(1-\left(\frac{1}{3}+\frac{2}{3}r_1r_2\right)\right).
\end{align}
Given that $\sum_{i=1}^5|\langle l_i|\psi\rangle|^2=\sqrt{5}$, this becomes
\begin{align}
&r_1r_2\sqrt{5}+\frac{5}{3}(1-r_1r_2)\nonumber\\
>&2+2(1-r_1r_2).
\end{align}
Rewriting this, we obtain
\begin{align}
r_1r_2>1-\frac{\sqrt{5}-2}{\sqrt{5}+\frac{1}{3}}\approx 0.908,
\end{align}
that is, the noncontextuality inequality can be violated only when the depolarizing noise is below a certain threshold given by $r_1r_2>0.908$. In terms of ${\rm Corr}$, this requires ${\rm Corr}>0.939$. The noiseless case $r_1=r_2=1$ takes us back to the ${\rm Corr}=1$ regime that we previously discussed.

\section{Discussion}\label{discussion}

\subsection{Measurement-measurement correlations vs.~source-measurement correlations}
Note that the usual Kochen-Specker experiment, as conceptualized in Refs.~\cite{CSW10,CSW,AFLS, sheaf}, for example, involves only the quantity $R([s|S])$, representing correlations between various measurement events when all the measurements are implemented on a system prepared according to the same preparation procedure, denoted by the source event $[s|S]$. Thus, $R$ represents measurement-measurement correlations on a system prepared according to a fixed choice of preparation procedure.

On the other hand, the experiment we have conceptualized in this paper involves, besides the quantity $R$, a quantity ${\rm Corr}$ representing source-measurement correlations, characterizing the quality of the measurements in terms of their response to corresponding preparations.

Our noncontextuality inequalities represent a trade-off relation that must hold between $R$ and ${\rm Corr}$ in an operational theory that admits a noncontextual ontological model. Here we note that the first example of such a tradeoff relation, albeit only for the case of operational quantum theory with unsharp measurements, appeared in Ref.~\cite{LSW} as the Liang-Spekkens-Wiseman (LSW) inequality \cite{KG} which has been shown to be experimentally violated in Ref.~\cite{LSWexpt}.\footnote{This experiment, however, is not in a position to make claims about contextuality without presuming the operational theory is quantum theory simply because the LSW inequality presumes operational quantum theory. The noncontextuality inequalities in this paper do not require the operational theory to be quantum theory and can therefore be experimentally tested using techniques from Refs.~\cite{exptlpaper,Kunjthesis,gpttompaper}.} And, indeed, the developments reported in Ref.~\cite{KunjSpek17} and the present paper have their origins in the idea of such a trade-off relation that first appeared in Ref.~\cite{LSW}.

\subsection{Can our noise-robust noncontextuality inequalities be saturated by a noncontextual ontological model?}\label{saturateineqs}
A natural question concerns the tightness of these noncontextuality inequalities, i.e., can Eq.~\eqref{NCI3} be saturated by a noncontextual ontological model? This requires one to specify a noncontextual ontological model reproducing the operational equivalences between the measurement events and between the source settings, such that 

\begin{equation}\label{saturateNCI}
{\rm Corr}+p_0(1-\beta(\Gamma_G,q))\frac{R-\alpha(G,w)}{\alpha^*(G,w)-\alpha(G,w)}=1.
\end{equation}

The assumption of measurement noncontextuality is already implicit in our characterization of the response functions $\xi(m_e|M_e,\lambda)$, and for this reason it is, indeed, trivial to satisfy measurement noncontextuality while saturating these noncontextuality inequalities. Measurement noncontextuality, alone, in fact even allows a violation of the inequality (when no preparation noncontextuality is imposed), the extreme case being $R=\alpha^*(G,w)$ and $1\geq{\rm Corr}>1-p_0(1-\beta(\Gamma_G,q))$. It's the assumption of preparation noncontextuality that is nontrivial to satisfy and we do not know if there exists a general construction of a noncontextual ontological model saturating our noncontextuality inequalities. We outline the general situation below.

\subsubsection{The special case of facet-defining Bell-KS inequalities: {\rm Corr}=1}
If outcome determinism is presumed (as in traditional Bell-KS type treatments), then we know that there exists a necessary and sufficient set of Bell-KS inequalities (each corresponding to a particular choice of $R([s|S])$) that are satisfied by {\em any} operational theory admitting a KS-noncontextual ontological model. In particular, each such (facet) Bell-KS inequality can be saturated by KS-noncontextual ontological models that yield probabilities (from $\mathcal{G}(\Gamma_G)$) corresponding to the facet-defining Bell-KS inequality, i.e., which satisfy $R([s|S])=\alpha(G,w)$ for such a Bell-KS inequality. Indeed, our noise-robust noncontextuality inequalities corresponding to these choices of $R([s|S])$ (i.e., facet-defining Bell-KS inequalities of the Bell-KS polytope which is given by the convex hull of points in $\mathcal{G}(\Gamma_G)|_{\rm det}$) can always be saturated when ${\rm Corr}=1$, because in that case outcome determinism is justified by preparation noncontextuality (cf.~Ref.~\cite{KunjSpek17}) and our inequalities are identical to the Bell-KS inequalities (saturated by $R=\alpha(G,w)$).

\subsubsection{The general case: ${\rm Corr}<1$}
Since we do not want to assume outcome determinism, nor necessarily the idealization of ${\rm Corr}=1$, what is at stake here is the assumption of preparation noncontextuality. This assumption must be satisfied while saturating the noise-robust noncontextuality inequality in order for a measurement noncontextual ontological model to be universally noncontextual. Constructing such a noncontextual ontological model amounts to specifying the distributions $\mu(s_e|S_e,\lambda)$ and $\nu(\lambda)$ such that 
\begin{equation}
\forall\lambda\in\Lambda: \mu(\lambda|S_e)=\nu(\lambda), \forall e\in E(\Sigma_G),
\end{equation}
i.e., preparation noncontextuality holds, and we have (rewriting the saturation condition from Eq.~\eqref{saturateNCI})
\begin{eqnarray}\label{eq:saturation}
&&(\alpha^*(G,w)-\alpha(G,w)){\rm Corr}+p_0(1-\beta(\Gamma_G,q))R\nonumber\\
&=&(\alpha^*(G,w)-\alpha(G,w))+p_0\alpha(G,w)(1-\beta(\Gamma_G,q)),\nonumber\\
\end{eqnarray}
where 
\begin{equation}
{\rm Corr}=\sum_{s_{e_*}}p(s_{e_*}|S_{e_*}){\rm Corr}_{s_{e_*}},
\end{equation}
\begin{equation}
{\rm Corr}_{s_{e_*}}=\sum_{\lambda\in\Lambda}{\rm Corr}(\lambda)\mu(\lambda|S_{e_*},s_{e_*}),
\end{equation}
\begin{align}
&{\rm Corr}(\lambda)\nonumber\\
\equiv&\sum_{e\in E(\Gamma_G)}q_e\sum_{m_e,s_e}\delta_{m_e,s_e}\xi(m_e|M_e,\lambda)\mu(s_e|S_e,\lambda),
\end{align}
and
\begin{equation}
R=\sum_{\lambda\in\Lambda}R(\lambda)\mu(\lambda|S_{e_*},s_{e_*}=0).
\end{equation}
Unfortunately, we do not have a general construction that can show this to be possible for any noise-robust noncontextuality inequality obtained according to the approach we have outlined. We therefore leave it as an open question whether such an inequality can (always?) be saturated by a noncontextual ontological model.

\subsection{Can trivial POVMs ever violate these noncontextuality inequalities?}\label{trivialnoviol}

No.

Recall that a trivial POVM is defined as an assignment of positive operators $p(v)\mathbb{I}$ to the vertices of $\Gamma_G$, where $\mathbb{I}$ is the identity operator on some Hilbert space and $p:V(\Gamma_G)\rightarrow[0,1]$, such that $\sum_{v\in e}p(v)=1$ for all $e\in E(\Gamma_G)$, is a probabilistic model on $\Gamma_G$.

\subsubsection{The case $p\in\mathcal{C}(\Gamma_G)$}
Consider trivial POVMs corresponding to any KS-noncontextual probabilistic model, i..e., $p\in\mathcal{C}(\Gamma_G)$ is a convex mixture of deterministic vertices, $\mathcal{G}(\Gamma_G)|_{\rm det}$, or equivalently, of ontic states in $\Lambda_{\rm det}$. In other words, $\mathcal{C}(\Gamma_G)\equiv{\rm ConvHull}(\mathcal{G}(\Gamma_G)|_{\rm det})$, the convex hull of points in $\mathcal{G}(\Gamma_G)|_{\rm det}$. The largest value ${\rm Corr}$ can take in this case is less than or equal to 1. This means that the upper bound on $R$ from our noncontextuality inequality, Eq.~\eqref{NCI2}, will be greater than or equal to $\alpha(G,w)$, whereas we know that for a KS-noncontextual probabilistic model, $R\leq \alpha(G,w)$. Hence, there is no violation of our noncontextuality inequality for such trivial POVMs.

\subsubsection{The case $p\in{\rm ConvHull}(\mathcal{G}(\Gamma_G)|_{\rm ind})$}
Now consider trivial POVMs that correspond to the indeterministic vertices, $\mathcal{G}(\Gamma_G)_{\rm ind}$ (correspondingly, $\Lambda_{\rm ind}$), or their convex mixtures. We know that for these trivial POVMs, ${\rm Corr}\leq\beta(\Gamma_G,q)$.
For any $R\leq\alpha^*(G,w)$ that is achieved by these trival POVMs, our noncontextuality inequality reads
 \begin{equation}
 {\rm Corr}\leq 1-p_0(1-\beta(\Gamma_G,q))\frac{R-\alpha(G,w)}{\alpha^*(G,w)-\alpha(G,w)},
 \end{equation}
A sufficient condition for this inequality to be satisfied is that 
\begin{equation}\label{eq:suffcond}
\beta(\Gamma_G,q)\leq 1-p_0(1-\beta(\Gamma_G,q))\frac{R-\alpha(G,w)}{\alpha^*(G,w)-\alpha(G,w)},
\end{equation}
which reduces, for $R>\alpha(G,w)$, to
\begin{equation}
p_0\leq\frac{\alpha^*(G,w)-\alpha(G,w)}{R-\alpha(G,w)},
\end{equation}
where the upper bound is greater than or equal to 1, since $\alpha(G,w)<R\leq\alpha^*(G,w)$. This is trivially satisfied since $p_0\leq 1$. 

For $R<\alpha(G,w)$, the sufficient condition of Eq.~\eqref{eq:suffcond} is again trivially satisfied since it reduces to 
\begin{equation}
p_0\geq-\frac{\alpha^*(G,w)-\alpha(G,w)}{\alpha(G,w)-R},
\end{equation}
and we must anyway have $p_0\geq 0$. 

For $R=\alpha(G,w)$, the sufficient condition reduces to $\beta(\Gamma_G,q)\leq 1$, which is again trivially satisfied since $\beta(\Gamma_G,q)<1$ by definition.

\subsubsection{The general case $p\in\mathcal{G}(\Gamma_G)$}
In general, a probabilistic model achieved by trivial POVMs can be in the convex hull of both deterministic ($\Lambda_{\rm det}$) and indeterministic ($\Lambda_{\rm ind}$) ontic states, with the total weight on deterministic ontic states denoted by ${\rm Pr}(\Lambda_{\rm det})$ and that on indeterministic ontic states by ${\rm Pr}(\Lambda_{\rm ind})$, so that ${\rm Pr}(\Lambda_{\rm det})+{\rm Pr}(\Lambda_{\rm ind})=1$. We then have 
\begin{eqnarray}
{\rm Corr}&\leq& {\rm Pr}(\Lambda_{\rm det})+{\rm Pr}(\Lambda_{\rm ind})\beta(\Gamma_G,q),\nonumber\\
R&\leq& {\rm Pr}(\Lambda_{\rm det})\alpha(G,w)+{\rm Pr}(\Lambda_{\rm ind})\alpha^*(G,w).\nonumber\\
\end{eqnarray}
A sufficient condition for satisfaction of the noncontextuality inequality is then 
\begin{eqnarray}
&&1-{\rm Pr}(\Lambda_{\rm ind})(1-\beta(\Gamma_G,q))\nonumber\\
&\leq&1-p_0(1-\beta(\Gamma_G,q))\frac{R-\alpha(G,w)}{\alpha^*(G,w)-\alpha(G,w)},\nonumber\\
\end{eqnarray}
which becomes
\begin{equation}
p_0\leq \frac{\alpha^*(G,w)-\alpha(G,w)}{R-\alpha(G,w)}{\rm Pr}(\Lambda_{\rm ind})
\end{equation}
when $R>\alpha(G,w)$. Noting that $$R\leq \alpha(G,w)+{\rm Pr}(\Lambda_{\rm ind})(\alpha^*(G,w)-\alpha(G,w)),$$
we have
\begin{equation}
{\rm Pr}(\Lambda_{\rm ind})\geq \frac{R-\alpha(G,w)}{\alpha^*(G,w)-\alpha(G,w)},
\end{equation}
so that the sufficient condition for satisfaction of the noncontextuality inequality becomes $p_0\leq 1$, which is trivially satisfied.

When $R=\alpha(G,w)$, the sufficient condition becomes $\beta(\Gamma_G,q)\leq 1$, which is again trivially satisfied.

Finally, when $R<\alpha(G,w)$, the sufficient condition becomes 
\begin{equation}
p_0\geq-\frac{\alpha^*(G,w)-\alpha(G,w)}{\alpha(G,w)-R}{\rm Pr}(\Lambda_{\rm ind}),
\end{equation}
which is again trivially satisfied since $p_0\geq0$.

Hence trivial POVMs cannot yield a violation of our noncontextuality inequalities. This is the sense in which trivial POVMs cannot lead to nonclassicality in our approach, unlike the case of traditional Kochen-Specker approaches \cite{CSW10,CSW,AFLS, sheaf} applied to the case of POVMs \cite{hensonsainz}.  To violate our noncontextuality inequalities, the POVMs must necessarily have some nontrivial projective component (that is not the identity operator or zero) but they {\em need not} be projectors. Further, we do not rely on restricting the notion of joint measurability \cite{notesonjm} (cf.~Section \ref{compatibilitysection}) to commutativity for POVMs. Taking joint measurability to be just commutativity is the  approach adopted in, for example, Ref.~\cite{sheaf}. We refer to Appendix \ref{allegedstrawman} and Appendix \ref{trivial} for more discussion on these issues, in particular Appendix \ref{trivial} for the role of commutativity vs.~joint measurability.

\section{Conclusions}\label{conclusions}

We have obtained a hypergraph framework for obtaining noise-robust noncontextuality inequalities corresponding to KS-colourable scenarios, suitably augmented with preparation procedures in the spirit of Spekkens contextuality \cite{Spe05}. The inequalities take the form of a noncontextual tradeoff between the three operational quantities ${\rm Corr}$, $R$, and $p_0$, cf.~Eq.~\eqref{NCI3}. This framework leverages the graph invariants from the graph-theoretic framework of CSW for doing this, in addition to a new hypergraph invariant (Eq.~\eqref{betadefn}) that we call the {\em weighted max-predictability}. Our approach is general enough to be applicable to any situation involving noisy preparations and measurements that arises from a KS-colourable contextuality scenario. 
 
We conclude with a list of open questions raised in this paper and other directions for future research:

\begin{enumerate}
	\item {\em Characterizing structural Specker's principle from probabilistic models on a hypergraph $\Gamma$:} 
	
	Given that $\mathcal{CE}^1(\Gamma)=\mathcal{G}(\Gamma)$ for some $\Gamma$, is it the case that $\Gamma$ must then necessarily satisfy structural Specker's principle, namely, that every clique in $O(\Gamma)$ is a subset of some hyperedge in $\Gamma$? Or is it the case that there exists a hypergraph $\Gamma'$ for which $\mathcal{CE}^1(\Gamma')=\mathcal{G}(\Gamma')$ but structural Specker's principle fails?
	
	More generally, is there {\em any} characterization of a hypergraph satisfying structural Specker's principle entirely in terms of the probabilistic models on it?
	
	As already pointed out earlier, this open question relates to the open Problem 7.2.3 of Ref.~\cite{AFLS} of characterizing $\Gamma$ for which $\mathcal{CE}^1(\Gamma)=\mathcal{G}(\Gamma)$. It is known that $\Gamma$ representing bipartite Bell scenarios \cite{LO} satisfy the property $\mathcal{CE}^1(\Gamma)=\mathcal{G}(\Gamma)$ and we have provided a generic recipe for converting any $\Gamma$ that does not satisfy structural Specker's principle to a $\Gamma'$ that does satisfy it so that $\mathcal{CE}^1(\Gamma')=\mathcal{G}(\Gamma')$. The question is if there are any {\em other} $\Gamma$ that also satisfy $\mathcal{CE}^1(\Gamma)=\mathcal{G}(\Gamma)$.
	
	\item {\em Almost quantum theory:} We know that an almost quantum theory cannot satisfy Specker's principle \cite{almostquantumsharp} but it satisfies statistical Specker's principle (or consistent exclusivity). An open question that remains is: 
	
	Can an almost quantum theory satisfy structural Specker's principle?
	
	If not, this would render the satisfaction of consistent exclusivity by an almost quantum theory unexplained by a natural structural feature of measurements in the theory, namely, the satisfaction of structural Specker's principle, i.e., almost quantum theory would not fall in the category of operational theories envisaged in Ref.~\cite{giuliosharp}.

	\item {\em Conditions for saturating the noise-robust noncontextuality inequalities:}
	
	As mentioned in Section \ref{saturateineqs}, it is an open question whether the noise-robust noncontextuality inequalities of Eq.~\eqref{NCI3} based on our generalization of the CSW framework \cite{CSW} can be saturated by a noncontextual ontological model.
	
More generally, the status of these noise-robust noncontextuality inequalities vis-\`a-vis the algorithmic approach of Ref.~\cite{algorithmic} for finding necessary and sufficient conditions for noncontextuality in a general prepare-and-measure scenario remains to be explored. One would suspect that the algorithmic approach of Ref.~\cite{algorithmic} when adapted to the kind of situation considered in this paper would yield nontrivial noncontextuality inequalities that aren't merely generalizations of the ones obtained in the CSW framework \cite{CSW}. It would be interesting to investigate the full structure of this set of inequalities and compare it with the facet-defining Bell-KS inequalities of the CSW framework.
	
	\item {\em Properties of the weighted max-predictability, $\beta(\Gamma_G,q)$:}
	
	Since the crucial new hypergraph-theoretic ingredient in our inequalities is the weighted max-predictability, it would be interesting to understand properties of this hypergraph invariant on both counts: as a new mathematical object in its own right, one we haven't been able to find a reference to in the hypergraph theory literature, as well as an important parameter of a hypergraph relevant for noise-robustness of a noise-robust noncontextuality inequality. Indeed, as we point out in Footnote \ref{strongestpossible}, identifying a distribution $q$ (in the definition of ${\rm Corr}$, Eq.~\eqref{eq:Corr}) that minimizes $\beta(\Gamma_G,q)$ for a given $\Gamma_G$ would lead to better noise-robustness in the inequalities of Eqs.~\eqref{NCI1} or \eqref{NCI2}.
	
	\item {\em Noise-robust applications of quantum protocols based on KS-contextuality:}
	
	A general research direction is to construct noise-robust versions of applications that have previously been suggested for KS-contextuality. Our approach provides a recipe for doing this for any Bell-KS inequality appearing in such applications. Besides serving as a witness for strong nonclassicality \cite{robweakstrongtalk} (i.e., Spekkens contextuality),\footnote{As opposed to weak nonclassicality that can arise in epistemically restricted classical theories \cite{quasiquant}. See also the talk at Ref.~\cite{robweakstrongtalk}, 41:43 minutes, for a short discussion.} noise-robust versions of these applications can help benchmark the experiments in terms of the noise that can be tolerated while still witnessing nonclassicality. Examples of such applications include those from Refs.~\cite{wehnervidick, mbqc, magic, qcsi, qubitqc, qkd}.
	
\end{enumerate}

\section*{Acknowledgments}
I would like to thank Andreas Winter for his comments on an earlier version of some of these ideas, Tobias Fritz for the ping-pong and the sing-song in which we often talked about hypergraphs, Rob Spekkens for the often argumentative -- but always productive -- conversations over lunch, and participants at the Contextuality conference (CCIOSA) at Perimeter Institute, during July 24 - 28, 2017, for very stimulating discussions that fed into the narrative of this paper. I would also like to thank David Schmid, Ana Bel{\'e}n Sainz, Elie Wolfe, and Tom\'a\v s Gonda for helping me better articulate the difference between structural vs.~statistical readings of Specker's principle, and Eric Cavalcanti for comments on the manuscript. Theorem \ref{SPimpliesStrSP} owes its origin to a discussion with Tom\'a\v s Gonda. I would also like to thank anonymous referees for suggestions that immensely improved the presentation of these results. Research at Perimeter Institute is supported by the Government of Canada through the Department of Innovation, Science and Economic Development Canada, and by the Province of Ontario through the Ministry of Research, Innovation and Science.

\begin{appendix}
\section{Status of KS-contextuality as an experimentally testable notion of nonclassicality for POVMs in quantum theory}\label{allegedstrawman}
The purpose of this section is to emphasize how the progression from KS-contextuality to Spekkens contextuality for KS-type contextuality experiments is a natural one rather than an {\em ad hoc} move from one framework to another. That is, Spekkens contextuality is not just another notion of nonclassicality that is incomparable with KS-contextuality, but is indeed intimately connected in its motivations to the limitations of KS-contextuality \cite{Spe05}. In particular, we will focus on the role of KS-contextuality with respect to POVMs and why allowing arbitrary POVMs poses a difficulty for KS-contextuality as a notion of nonclassicality that is experimentally testable, i.e., a notion that applies to noisy measurements (POVMs) typically implemented in a laboratory experiment.\footnote{And how a rather compelling way to arrive at a notion that {\em is} experimentally testable is Spekkens contextuality.} While one may be tempted to reject this premise for assessing the suitability of KS-contextuality as a notion of nonclassicality -- claiming instead that KS-contextuality was never meant for POVMs and applies only to ``purified" experiments (namely, ones with only PVMs and pure states) -- the reasons for doing so are rooted in the literature on KS-contextuality where POVMs have indeed been considered and (at least) two kinds of conclusions drawn: one, that there exists a Kochen-Specker contradiction for POVMs, even on a qubit, so KS-contextuality for POVMs is interesting \cite{CabelloPOVM} and two, that allowing arbitrary POVMs in assessing nonclassicality would make the research program of identifying device-independent principles for quantum correlations in KS-contextuality experiments ill-defined, so quantum correlations allowing arbitrary POVMs are ``pathological" \cite{hensonsainz}. We will look at these arguments in turn and use the latter, in particular, to segue into our motivations for the framework proposed in this paper.

\subsection{Limitations of KS-contextuality vis-\`a-vis POVMs}

\subsubsection{KS-contextuality for POVMs in the literature}
The first paper that applied KS-contextuality to the case of POVMs was by Cabello \cite{CabelloPOVM} where a KS-uncolourability argument for POVMs on a single qubit was proposed. This was motivated by the Gleason-type derivation of the Born rule starting with the structure of POVMs due to Busch \cite{busch} and Caves {\em et al.}~\cite{cavesetal}, analogous to the case of the Kochen-Specker theorem \cite{KochenSpecker} which can be seen as motivated by Gleason's theorem \cite{gleason}. Insofar as there exists a Gleason-type theorem for POVMs \cite{busch, cavesetal}, one could motivate KS-contextuality as a reasonable notion of nonclassicality for POVMs, as was presumably the case in Ref.~\cite{CabelloPOVM}. The role of this notion of nonclassicality is then just to argue -- using a finite set of POVM elements -- that no KS-noncontextual assignment of outcomes is possible for certain finite sets of POVMs in quantum theory. Should we, however, assume that it is reasonable to demand deterministic assignment of outcomes to POVM elements in an ontological model, just as we do for PVM elements? The argument of Ref.~\cite{CabelloPOVM} was later criticized on various counts \cite{Spe05, grudka, ODUM} and we refer the reader to Ref.~\cite{ODUM} for criticisms pertinent to this paper, namely, that {\em outcome determinism for all unsharp measurements} (ODUM in Ref.~\cite{ODUM}) in quantum theory is untenable.\footnote{Ref.~\cite{ODUM} is also a good resource for a detailed analysis of arguments concerning dilations of POVMs, which we will not get into here. Besides, it also provides a principled recipe for assigning response functions to POVMs.} Other works in the literature where KS-contextuality for POVMs has been explored include Refs.~\cite{aravind,methot,povmexpt,mancinska}. 

Besides, doubts about the experimental testability of the KS theorem were raised in the late `90s in a series of papers by Meyer, Clifton, and Kent \cite{meyer,kent,cliftonkent}. A review can be found in Ref.~\cite{barrettkent}. These doubts were premised on the idea that the set of KS-colourable projectors (or PVMs) on any given Hilbert space is dense in the set of all projectors (or PVMs) on that Hilbert space. That is, for any given set of PVMs yielding a KS contradiction, it is always possible to find PVMs which are arbitrarily ``close" to the PVMs required for a KS contradiction (for any finite precision) but which do not themselves lead to a KS contradiction. The property of denseness of KS-colourable sets of measurements in the set of all measurements in fact extends to even the most general case when the measurements are POVMs on any Hilbert space. So, even a KS contradiction for POVMs (such as the one in Ref.~\cite{CabelloPOVM}) falls prey to the Meyer-Clifton-Kent argument \cite{barrettkent}. As Ref.~\cite{barrettkent} notes:
\begin{quotation}
	Dealing with projective measurements is arguably not enough. One quite popular view of quantum theory holds that a correct version of the measurement rules would take POV measurements as fundamental, with projective measurements either as special cases or as idealisations which are never precisely realised in practice. In order to define an NCHV theory catering for this line of thought, Kent constructed a KS-colourable dense set of positive operators in a complex	Hilbert space of arbitrary dimension, with the feature that it gives rise to a dense set of POV decompositions of the identity (Kent, 1999). Clifton and Kent constructed a dense set of positive operators in complex Hilbert space of arbitrary dimension with the special feature that no positive operator in the set belongs to more than one decomposition of the identity (Clifton \& Kent, 2000).
	Again, the resulting set of POV decompositions is dense, and the special feature	ensures that one can average over hidden states to recover quantum predictions.
\end{quotation}

 Hence, in any finite precision experiment it would be impossible to test the Kochen-Specker theorem, i.e., such an experimental test would require an infinitely precise measurement and measurements in a real-world laboratory are never infinitely precise. Although there was a lively debate along these lines (see the references in \cite{barrettkent}), the resolutions that were proposed all involved modifying the notion of KS-noncontextuality by adding auxiliary assumptions that seek to exclude the Meyer-Clifton-Kent type arguments. A recent attempt in this direction can be found in Ref.~\cite{Winter} where a notion of ``ontological faithfulness" is proposed. As such, it was already recognized -- for reasons independent of Spekkens contextuality \cite{Spe05} -- that the notion of KS-noncontextuality needs to be revised if one is to make it experimentally testable.\footnote{Of course, this takes nothing away from the importance of the Kochen-Specker theorem \cite{KochenSpecker} as a no-go theorem concerning the logical structure of quantum theory and the constraints it places on the ontological models possible for the theory.} What Spekkens brought to the fore \cite{Spe05}, besides generalizing the notion of contextuality to all experimental procedures rather than measurements alone, was the idea that an experimental test of noncontextuality should not rely on inequalities that presume outcome determinism, just as a test of local causality does not require the assumption of outcome determinism. Indeed, the assumption of outcome determinism for sharp measurements in quantum theory is {\em derived} in the Spekkens framework from the assumption of preparation noncontextuality rather than being assumed independently.

We will now consider the more modern approach to KS-contextuality along the lines of the frameworks in Refs.\cite{CSW, AFLS, sheaf} to segue into our framework for Spekkens contextuality which we develop in this paper.

\subsubsection{Classifying probabilistic models: restriction of quantum models to PVMs}\label{trivialpovms}
Research on KS-contextuality took a different turn with the advent of the graph-theoretic framework of Cabello, Severini and Winter in 2010 \cite{CSW10} (revised slightly in 2014 \cite{CSW}), the sheaf-theoretic framework of Abramsky and Brandenburger in 2011 \cite{sheaf}, and the hypergraph based formalism of Ac\`in, Fritz, Leverrier, and Sainz in  2012 \cite{AFLS}. The unifying theme of these contributions was that they took the key mathematical idea underlying KS-noncontextuality and Bell-locality --- namely, that both are instances of the classical marginal problem \cite{fine, ChavesFritz1, ChavesFritz2} --- and built frameworks that sought to distinguish between classical theories (namely, those admitting KS-noncontextual ontological models), quantum theory, and post-quantum general probabilistic theories by classifying their empirical predictions relative to a Kochen-Specker experiment into these categories. All these frameworks, motivated by the device-independence paradigm, eschewed the erstwhile restriction of the notion of KS-noncontextuality to quantum theory and sought to make their analysis theory-independent, relying only on empirical predictions relative to a KS experiment to classify theories. They separated the assumption of KS-noncontextuality from the operational theory -- namely, quantum theory -- to which it was originally meant to apply, allowing arbitrary operational theories in their analysis. However, there was a key distinction between Bell scenarios and KS-contextuality scenarios that was lost in this formal unification: namely, that while the definition of a quantum probabilistic model in a Bell scenario need not be restricted to  (local) PVMs (and arbitrary local POVMs can be allowed without changing the set of quantum models), the same is not true of a KS-contextuality scenario. Indeed, as Henson and Sainz note in their work \cite{hensonsainz},\footnote{Proposing a principle bounding the KS-contextuality possible in quantum theory, namely, ``Macroscopic Noncontextuality".} reflecting on the question of allowing arbitrary POVMs in the definition of a quantum probabilistic model:

\begin{quotation}
...if we allow general POVMs rather than projective
measurements then no principle that places a non-trivial restriction on correlations will be respected. Thus, this kind of ``quantum model" is clearly pathological.
\end{quotation}

One way to motivate the present work is as a response to the pathology that Henson and Sainz allude to: that trivial POVMs can realize any probabilistic model, hence allowing arbitrary POVMs makes the problem of finding principles to identify quantum correlations in KS-contextuality scenarios trivial, i.e., all probabilistic models are quantum and there is nothing to be learnt about post-quantum probabilistic models. This is because any set of probabilities satisfying the ``no-disturbance" or ``no-signalling" condition (of which the $E1$  correlations of CSW \cite{CSW} are a subset, in general) can be achieved by (trivial) POVMs by simply multiplying an identity operator with every probability in such an assignment of probabilities.\footnote{Trivial POVMs are, therefore, trivial resolutions of the identity, where every POVM element is proportional to identity, i.e., $\{a\mathbb{I}\}_a$, such that $a\in[0,1]$ and $\sum_a a=1$.} By the lights of KS-noncontextuality as one's notion of classicality, then, trivial POVMs saturating the general probabilistic bound on the correlations would seem to be maximally nonclassical (i.e., maximally KS-contextual).
To avoid such ``pathological" quantum models, they restrict the definition of a quantum model to allow only projective measurements. Indeed, with recent work on a sensible notion of ``sharp" measurement in a general probabilistic theory \cite{giuliosharp, giuliosharp2}, an appeal to the ``fundamental sharpness" of all measurements (see, e.g., \cite{cabellotalk}) is made to restrict attention to sharp measurements in both quantum theory and general probabilistic theories. 

On the other hand, the approach in this paper is different. In particular, we want our approach to capture the intuition that trivial POVMs are ``classical" (and not pathological), so we must go beyond KS-noncontextuality. A simple operational sense in which trivial POVMs are ``classical" is that they reveal nothing about the quantum state on which they are measured, being incapable of distinguishing any pair of states whatsoever.\footnote{Indeed, any trivial POVM can be realized in the following operational manner: take the quantum system prepared in some state, throw it in the garbage, and then sample from the classical probability distribution corresponding to the trivial POVM.} The correlations (denoted by $R([s|S])$) usually examined in a KS-contextuality experiment do not allow such experiments to witness the ``triviality" of trivial POVMs, i.e., the fact that they correspond to a fixed probability distribution that doesn't vary even as the choice of preparation is varied. Moreover, since all nonprojective measurements are excluded {\em by fiat} in traditional Kochen-Specker type approaches \cite{CSW,AFLS} for reasons alluded to by Henson and Sainz \cite{hensonsainz}, one loses out on the potential to explore the possibilities that nontrivial and nonprojective measurements offer with respect to contextuality.\footnote{All trivial POVMs are nonprojective, but not all nonprojective POVMs are trivial. Indeed, see Refs.~\cite{LSW,KG,KHF,jmctx} for examples of generalized contextuality \cite{Spe05} with nonprojective measurements, albeit assuming operational quantum theory.} Our approach, therefore, 
 is to allow arbitrary POVMs when considering probabilistic models arising from quantum theory (and not restricting to any notion of ``sharp measurements" in general probabilistic theories) but examine more quantities than are examined in traditional approaches, i.e., besides the quantity $R$ typical in a KS-contextuality scenario, we invoke the quantity ${\rm Corr}$ to account for noise in the measurements.

 If one restricts attention to operational theories that can always achieve ${\rm Corr}=1$ for any KS-contextuality scenario, then the usual classification of probabilistic models following Refs.~\cite{CSW,AFLS} holds (Eq.~\eqref{NCI3}). What is of interest in our framework, however, is the tradeoff between $R$ and ${\rm Corr}$: how large can both $R$ and ${\rm Corr}$ be in an operational theory? (See Eq.~\eqref{NCI3}.)
 
\subsection{Robustness of Bell nonlocality vis-\`a-vis POVMs}
Note that whenever we refer to ``Bell-KS" functionals or inequalities for Kochen-Specker type experiments, we are {\em not} thinking of experiments that are Bell experiments \cite{Bell64,Bell66,CHSH,Belltest1,Belltest2,Belltest3}, which have spacelike separation between multiple parties, each performing local measurements on a shared multipartite preparation. For the case of Bell experiments, trivial local POVMs assigned to each party in a Bell experiment do not lead to Bell violations for a simple reason: the trivial POVMs for each party are all compatible with each other, thereby admitting a joint probability distribution over their outcomes for each party; taking a product of these local joint probability distributions (one for each party) results in a joint distribution over all measurements of all parties, hence satisfying Bell inequalities. The fact that the POVMs are trivial ensures that the Bell inequalities are satisfied regardless of the choice of shared quantum state. On the other hand, forgetting the constraint of {\em local} POVMs, there always exist global trivial POVMs that can violate Bell inequalities: e.g., just take the Popescu-Rohrlich (PR) box distribution \cite{PR}, and multiply an identity operator (on the joint Hilbert space of Alice and Bob) with each probability in the PR-box; this results in four trivial POVMs, defined over the joint Hilbert space, that together violate the CHSH inequality maximally. But, of course, this violation is uninteresting because it doesn't obey the {\em locality} constraint on the measurements in a Bell experiment. This is mathematically reflected in the fact that the PR-box distribution cannot be written as a convex mixture of product distributions, one for each party, hence the corresponding trivial POVM cannot be understood in terms of trivial {\em local} POVMs. Hence, it is the {\em locality} of the trivial POVMs in a Bell experiment that prevents them from violating a Bell inequality and renders them non-pathological, unlike in the case of KS-contextuality. The fact that they are ``trivial" in the sense of being unable to distinguish two quantum states plays a role in the sense that, {\em regardless} of the shared quantum state, these POVMs yield fixed distributions over the measurement outcomes, thus always allowing the construction of a fixed (that is, independent of the quantum state) global joint probability distribution over all measurements in a Bell scenario. Since there are no such locality constraints on the form of the POVM elements in a Kochen-Specker experiment, they can easily violate any KS-noncontextuality inequality, e.g., the two-party CHSH experiment considered as a Kochen-Specker experiment with four observables in a 4-cycle where adjacent pairs are jointly measurable allows for trivial POVMs (like the PR-box trivial POVM above) violating the CHSH-type Bell-KS inequality in this scenario maximally. By the lights of KS-noncontextuality, this violation would indicate the maximum possible KS-contextuality with respect to this CHSH-type inequality.\footnote{See Appendix \ref{trivial} for more discussion.}
For all these reasons, our discussion of KS-noncontextuality as a notion of classicality --- in an experiment with no {\em locality} constraints on the measurements --- does not extend to the case of Bell-locality (or local causality) as a notion of classicality in a Bell experiment, where the experiment {\em must} respect {\em locality} constraints on the measurements for a Bell inequality violation to be meaningful. 

The unification of Bell nonlocality and KS-contextuality \`a la Refs.~\cite{CSW, AFLS, sheaf} forces a certain dichotomy in these approaches: while in Bell scenarios, one {\em need not} restrict to any notion of a ``sharp" measurement in the definition of probabilistic models (and thus claim ``theory independence"), in Kochen-Specker scenarios, one must make some statement about the nature of the measurements (concerning their presumed sharpness \cite{cabellotalk}, or that their joint measurability \cite{notesonjm,jmctx} is restricted to commutativity \cite{sheaf}), rendering any putative ``theory independence" claim (on a level at par with Bell nonlocality) unfounded.\footnote{See Ref.~\cite{finegen} for how this lack of locality of measurements in a Kochen-Specker type experiment translates, at the ontological level, to the unreasonableness of assuming factorizability in the ontological model; this factorizability (or the stronger condition of outcome determinism) is invoked to justify the resulting derivation of Bell-KS inequalities as constraints from a classical marginal problem.}

\section{Ontological models without respecting coarse-graining relations}\label{ontmodelwocoarsegraining}
Here we will construct explicit examples where the coarse-graining relations are not respected in an ontological model, in contrast to the requirement on the representation of coarse-grainings that we invoked in Section II.C of the main text. The goal is to emphasize that the requirement of Section II.C is {\em necessary} not only for the treatment of Spekkens contextuality but also for Kochen-Specker contextuality. Below, we first demonstrate how a ``KS-noncontextual" model can be constructed for any scenario that proves the KS theorem by using the example of the KCBS setup \cite{KCBS}. We then proceed to demonstrate how a ``preparation and measurement noncontextual" model can be constructed in a similar way whenc considering generalized noncontextuality \cite{Spe05}.

\subsection{How to construct a ``KS-noncontextual" ontological model of the KCBS experiment \cite{KCBS} without coarse-graining relations}\label{trivialKSNC}
Here we have that $\mathbb{M}$ contains at least the following measurement settings: $\{M_i\}_{i=1}^5$, each with three possible outcomes, $m_i\in\{0,1,2\}$. The measurement events for each measurement setting $M_i$ can be coarse-grained into two different ways, defining new measurement settings $M'_i$ (with outcomes $m'_i\in\{0,\bar{0}\}$) and $M''_i$ (with outcomes $m''_i\in\{2,\bar{2}\}$), where the coarse-graining relations are given by 
\begin{align}
&[0|M'_i]\equiv[0|M_i],\label{kcbscgrelns1}\\
&[\bar{0}|M'_i]\equiv[1|M_i]+[2|M_i],\\
&[2|M''_i]\equiv[2|M_i],\label{kcbscgrelns2}\\
&[\bar{2}|M'_i]\equiv[0|M_i]+[1|M_i].
\end{align}
In the operational theory, these coarse-graining relations are respected, i.e., for all $[s|S]$, $s\in V_S, S\in\mathbb{S}$,
\begin{align}
&p(0,s|M'_i,S)\equiv p(0,s|M_i,S),\\
&p(\bar{0},s|M'_i,S)\equiv p(1,s|M_i,S)+p(2,s|M_i,S),\\
&p(2,s|M''_i,S)\equiv p(2,s|M_i,S),\\
&p(\bar{2},s|M''_i,S)\equiv p(0,s|M_i,S)+p(1,s|M_i,S).
\end{align}
However, we do not require that these relations be respected in an ontological model. Now, the KCBS argument requires the following operational equivalences,
\begin{equation}\label{kcbsopequivs}
[2|M''_i]\simeq [0|M'_{i+1}],
\end{equation}
for all $i\in\{1,2,3,4,5\}$, where addition is modulo 5, so that $i+1=1$ for $i=5$. A KS-noncontextual ontological model for this experiment requires that 
\begin{equation}
\xi(2|M''_i,\lambda)=\xi(0|M'_{i+1},\lambda)\in\{0,1\},\quad \forall \lambda\in \Lambda.
\end{equation}
Constructing such a model requires one to specify response functions for the measurements $\{M_i,M'_i,M''_i\}_{i=1}^5$. However, since there are no constraints from coarse-graining relations on these response functions, there is no obstruction to the construction of a ``KS-noncontextual model" of this type for any set of operational statistics. In particular, since we do not require that $\forall\lambda\in\Lambda: \xi(0|M'_{i+1},\lambda)\equiv \xi(0|M_{i+1},\lambda)$, nor that $\forall\lambda\in\Lambda: \xi(2|M''_i,\lambda)\equiv \xi(2|M_i,\lambda)$, we can assign arbitrary response functions to $\{M'_i,M''_i\}_{i=1}^5$, subject only to the condition from KS-noncontextuality that $\forall\lambda\in\Lambda: \xi(2|M''_i,\lambda)=\xi(0|M'_{i+1},\lambda)\in\{0,1\}$.\footnote{This ``KS-noncontextual" ontological model will thus reproduce operational equivalences of the type $[2|M''_i]\simeq [0|M'_{i+1}]$ (cf.~Eq.~\eqref{kcbsopequivs}).} Note that, because coarse-graining relations are not respected, this does not imply that $\forall\lambda\in\Lambda: \xi(2|M_i,\lambda)=\xi(0|M_{i+1},\lambda)\in\{0,1\}$, which is the usual constraint we would have presumed from KS-noncontextuality when coarse-graining relations are respected in the ontological model. In the absence of any such constraints on the response functions for $\{M_i\}_{i=1}^5$, one can always reproduce their operational statistics, in particular the operational equivalences of the type $[2|M_i]\simeq [0|M_{i+1}]$, which follow from Eqs.~\eqref{kcbscgrelns1},\eqref{kcbscgrelns2}, and \eqref{kcbsopequivs}.

\subsection{How to construct a ``preparation and measurement noncontextual" ontological model without coarse-graining relations}
Just as for measurements in the case of KS-noncontextuality, abandoning the coarse-graining relations for preparations in the case of generalized noncontextuality \cite{Spe05} makes possible the existence of a ``preparation and measurement noncontextual" ontological model for any set of operational statistics. For the kinds of proofs of contextuality relevant to this article, the relevant notion of coarse-graining is that of complete coarse-graining: that is, consider two source settings $S$ and $S'$ with (respective) source events $\{[s|S]\}_{s\in V_S}$ and $\{[s'|S']\}_{s'\in S'}$, that can be completely coarse-grained to yield the operational equivalence $[\top|S_{\top}]\simeq [\top|S'_{\top}]$, cf.~Eq.~\eqref{opequiv2prep}. In the operational description, where we assume the coarse-graining relation is respected, this is represented by
\begin{align}
&\forall [m|M], m\in V_M, M\in\mathbb{M}:\nonumber\\
&\sum_sp(m,s|M,S)=\sum_{s'}p(m,s'|M,S').
\end{align}
In the ontological description, however, we do {\em not} impose the coarse-graining relations $\mu(\lambda,\top|S_{\top})\equiv \sum_s\mu(\lambda,s|S)$ and $\mu(\lambda,\top|S'_{\top})\equiv \sum_{s'}\mu(\lambda,s'|S')$, which makes it trivial to write down probability distributions $\mu(\lambda,\top|S_{\top})$ and $\mu(\lambda,\top|S'_{\top})$ such that $\mu(\lambda,\top|S_{\top})=\mu(\lambda,\top|S'_{\top})$ (as required by preparation noncontextuality applied to $[\top|S_{\top}]\simeq [\top|S'_{\top}]$) but where we do not require that $\sum_s\mu(\lambda,s|S)=\sum_s\mu(\lambda,s|S)$ (which is not required by preparation noncontextuality). Note how the refusal to respect the coarse-graining relations, i.e., identifying $\mu(\lambda,\top|S_{\top})$ with $\sum_s\mu(\lambda,s|S)$ and $\mu(\lambda,\top|S'_{\top})$ with $\sum_{s'}\mu(\lambda,s'|S')$, lifts the constraint from preparation noncontextuality that would have been in place if the coarse-graining relations were respected. The same refusal for the case of measurements lifts any constraints (just as in the case of KS-noncontextuality above) from measurement noncontextuality on the ontological model. It thus becomes trivial to construct a ``preparation and measurement noncontextual" ontological model without coarse-graining relations.

	\section{Trivial POVMs}\label{trivial}
	\subsection{Bell-CHSH scenario}
	We have the Hilbert space $\mathcal{H}_A\otimes\mathcal{H}_B$ for Alice ($\mathcal{H}_A$) and Bob ($\mathcal{H}_B$). Consider four binary-outcome POVMs, $\{A^{(0)},A^{(1)},B^{(0)},B^{(1)}\}$, where
	\begin{eqnarray}
	&&A^{(0)}\equiv\{A^{(0)}_0,A^{(0)}_1\},\nonumber\\
	&&A^{(1)}\equiv\{A^{(1)}_0,A^{(1)}_1\},\nonumber\\
	&&B^{(0)}\equiv\{B^{(0)}_0,B^{(0)}_1\},\nonumber\\
	&&B^{(0)}\equiv\{B^{(1)}_0,B^{(1)}_1\},
	\end{eqnarray}
	$0\leq A^{(0)}_0, A^{(1)}_0\leq \mathbb{I}_{\mathcal{H}_A}$, $0\leq B^{(0)}_0, B^{(1)}_0\leq \mathbb{I}_{\mathcal{H}_B}$, $A^{(0)}_0+A^{(0)}_1=A^{(1)}_0+A^{(1)}_1=\mathbb{I}_{\mathcal{H}_A}$, and 
	$B^{(0)}_0+B^{(0)}_1=B^{(1)}_0+B^{(1)}_1=\mathbb{I}_{\mathcal{H}_B}$. The quantum probability, given a shared quantum state $\rho_{AB}$ defined on $\mathcal{H}_A\otimes\mathcal{H}_B$, is given by 
	\begin{equation}
	p(a,b|x,y)=\Tr(\rho_{AB}A^{(x)}_a\otimes B^{(y)}_b),
	\end{equation}
	for $a,b,x,y\in\{0,1\}$. Here $A^{(x)}\otimes\mathbb{I}_{\mathcal{H}_B}$ is jointly measurable with $\mathbb{I}_{\mathcal{H}_A}\otimes B^{(y)}$, just because of the commutativity of their respective POVM elements. The joint observable being measured is $A^{(x)}\otimes B^{(y)}$. Now, consider the case when all the POVM elements are trivial, i.e., $A^{(x)}_a=q^{(x)}_a\mathbb{I}_{\mathcal{H}_A}$ and $B^{(y)}_b=r^{(y)}_b\mathbb{I}_{\mathcal{H}_B}$, for some $q^{(x)}_a,r^{(y)}_b\in[0,1]$ for all $a,b,x,y\in\{0,1\}$. We then have 
	\begin{equation}
	p(a,b|x,y)=q^{(x)}_ar^{(y)}_b, \forall a,b,x,y\in\{0,1\}.
	\end{equation}
	A global joint probability distribution which reproduces the above as marginals is simply given by their product:
	\begin{equation}
	p(a^{(0)},a^{(1)},b^{(0)},b^{(1)})\equiv q^{(0)}_{a^{(0)}}q^{(1)}_{a^{(1)}}r^{(0)}_{b^{(0)}}r^{(1)}_{b^{(1)}}.
	\end{equation}
	Hence, trivial POVMs never violate any Bell-CHSH inequality for this scenario.
	
	\subsection{CHSH-type contextuality scenario: 4-cycle}
	We now consider the Bell-CHSH scenario without the constraint of spacelike separation. What the lack of spacelike separation means from the quantum perspective is that one no longer needs to model this spacelike separation by requiring a tensor product structure, or (more generally) by requiring the commutativity of the observables that are jointly measured \cite{scholzwerner, fritztsirelson,sheaf}. That is, there is no physical justification for imposing the tensor product structure or the commutativity of jointly measured observables.\footnote{On the other hand, what this lack of spacelike separation means from the perspective of an ontological model is that one no longer has a justification for assuming factorizability \cite{sheaf} and, consequently, the generalization of Fine's theorem \cite{fine} fails to prove that there is no loss of generality in assuming outcome determinism in discussions of KS-contextuality (unlike the case of Bell scenarios, where factorizability is justified by spacelike separation); there is a definite loss of generality, in that measurement noncontextual and outcome-indeterministic ontological models that are non-factorizable are not empirically equivalent to measurement noncontextual and outcome-deterministic (or KS-noncontextual) ontological models. See Ref.~\cite{finegen} for a discussion of this aspect.}
	
		Thus, we have the Hilbert space $\mathcal{H}$ and we consider four binary-outcome POVMs, $\{A^{(0)},A^{(1)},B^{(0)},B^{(1)}\}$, on $\mathcal{H}$, where
	\begin{eqnarray}
	&&A^{(0)}\equiv\{A^{(0)}_0,A^{(0)}_1\},\nonumber\\
	&&A^{(1)}\equiv\{A^{(1)}_0,A^{(1)}_1\},\nonumber\\
	&&B^{(0)}\equiv\{B^{(0)}_0,B^{(0)}_1\},\nonumber\\
	&&B^{(0)}\equiv\{B^{(1)}_0,B^{(1)}_1\},
	\end{eqnarray}
	$0\leq A^{(0)}_0, A^{(1)}_0, B^{(0)}_0, B^{(1)}_0\leq \mathbb{I}_{\mathcal{H}}$, $A^{(0)}_0+A^{(0)}_1=A^{(1)}_0+A^{(1)}_1=B^{(0)}_0+B^{(0)}_1=B^{(1)}_0+B^{(1)}_1=\mathbb{I}_{\mathcal{H}}$. Further, the following sets of POVMs are jointly measurable: $\{A^{(0)},B^{(0)}\}$, $\{A^{(0)},B^{(1)}\}$, $\{A^{(1)},B^{(0)}\}$, $\{A^{(1)},B^{(1)}\}$. The most general joint observable for a pair of compatible POVMs $\{A^{(x)},B^{(y)}\}$ is given by a POVM $G^{(xy)}\equiv\{G^{(xy)}_{00},G^{(xy)}_{01},G^{(xy)}_{10},G^{(xy)}_{11}\}$ (that isn't necessarily unique \cite{jmctx}) such that: $G^{(xy)}_{00}+G^{(xy)}_{01}=A^{(x)}_0,G^{(xy)}_{10}+G^{(xy)}_{11}=A^{(x)}_1, G^{(xy)}_{00}+G^{(xy)}_{10}=B^{(y)}_0, G^{(xy)}_{01}+G^{(xy)}_{11}=B^{(y)}_1$. In particular, if (and only if) the POVMs $A^{(x)}$ and $B^{(y)}$ commute, we can construct the joint POVM as a product: $G^{(xy)}_{ab}=A^{(x)}_aB^{(y)}_b$ for all $a,b,x,y\in\{0,1\}$. In the absence of such commutativity, the joint POVM cannot be written as a product.
	
	 The quantum probability, given a quantum state $\rho$ on $\mathcal{H}$, is given by 
	\begin{equation}
	p(a,b|x,y)=\Tr(\rho G^{(xy)}_{ab}),
	\end{equation}
	for $a,b,x,y\in\{0,1\}$. Note that this probability depends on the joint measurement $G^{(xy)}$ implementing $A^{(x)}$ and $B^{(y)}$ together, and that, in general, there may be multiple choices of $G^{(xy)}$ possible. This is easy to see since there is one undetermined positive operator in the joint measurement that is not fixed by $A^{(x)}$ or $B^{(y)}$, i.e., we can write the POVM elements of $G^{(xy)}$ as: $G^{(xy)}_{01}=A^{(x)}_0-G^{(xy)}_{00}$, $G^{(xy)}_{10}=B^{(y)}_0-G^{(xy)}_{00}$, $G^{(xy)}_{11}=\mathbb{I}-A^{(x)}_0-B^{(y)}_0+G^{(xy)}_{00}$, where $G^{(xy)}_{00}$ is a positive semidefinite operator satisfying $A^{(x)}_0+B^{(y)}_0-\mathbb{I}_{\mathcal{H}}\leq G^{(xy)}_{00}\leq A^{(x)}_0,B^{(y)}_0$. Here $G^{(xy)}_{00}$ represents the freedom in the choice of how the joint measurement might be implemented within quantum theory. This freedom reflects the fact that since the jointly measured observables are no longer spacelike separated, it is possible to introduce correlations between them that are stronger than what is allowed in the corresponding Bell scenario in quantum theory. The strength of these correlations is only limited by the constraints on $G^{(xy)}_{00}$ imposed by the marginal observables $A^{(x)}$ and $B^{(y)}$. This is in contrast to the case where $A^{(x)}$ and $B^{(y)}$ are spacelike separated observables and the {\em only} choice of joint POVM consistent with spacelike separation is fixed by $G^{(xy)}_{00}=A^{(x)}_0B^{(y)}_0$, i.e., the strength of correlations between $A^{(x)}$ and $B^{(y)}$ is fixed entirely by them and there is no freedom in choosing $G^{(xy)}$.	
	
	Thus, we have that $A^{(x)}$ is jointly measurable with $B^{(y)}$ and $G^{(xy)}$ denotes a joint POVM of $A^{(x)}$ and $B^{(y)}$. Now, consider the case when all the POVM elements are trivial, i.e., $A^{(x)}_a=q^{(x)}_a\mathbb{I}_{\mathcal{H}}$ and $B^{(y)}_b=r^{(y)}_b\mathbb{I}_{\mathcal{H}}$, for some $q^{(x)}_a,r^{(y)}_b\in[0,1]$ for all $a,b,x,y\in\{0,1\}$. 
	
	In particular, consider the case where $q^{(x)}_a=r^{(y)}_b=\frac{1}{2}$ for all $a,b,x,y\in\{0,1\}$. A possible joint POVM for these trivial POVMs is then the product POVM:
	
	\begin{equation}
		G^{(xy)}_{ab}=A^{(x)}_aB^{(y)}_b=\frac{1}{4}\mathbb{I}_{\mathcal{H}}.
	\end{equation}
	
	If one restricted joint measurability of $A^{(x)}$ and $B^{(y)}$ to just {\em commutativity} --- a sufficient but not necessary condition for joint measurability\footnote{Particularly in the absence of spacelike separation. It is the need to model spacelike separation in a quantum Bell experiment that makes commutativity a necessary (and sufficient) condition for joint measurability of spacelike separated observables in a Bell scenario} \cite{notesonjm} --- we would take the above choice of the product POVM as a ``natural" one. Being a product of trivial POVMs, this choice will never lead to a violation of the CHSH-type inequality for this scenario. Indeed, the structure of a Bell scenario --- requiring the decomposition of the Hilbert space as $\mathcal{H}=\mathcal{H}_A\otimes\mathcal{H}_B$ ({\em tensor product paradigm}), or more generally, imposing the commutativity requirement $[A^{(x)}_a,B^{y}_b]=0$ ({\em commutativity paradigm}) --- is such that the only possible choice of joint measurement that can be implemented by spacelike separated parties is the one that corresponds to the product POVM, given by operators $G^{(xy)}_{ab}=A^{(x)}_aB^{(y)}_b$.
	
	However, this is not the only allowed joint measurement for these trivial POVMs, particularly when there is no locality constraint on the measurements from spacelike separation.\footnote{To incorporate such a constraint, spacelike separation needs to be modelled via either the tensor product paradigm or the commutativity paradigm. Both these ways of modelling spacelike separation lead to the same set of quantum correlations for any finite-dimensional Hilbert space $\mathcal{H}$ \cite{scholzwerner}. The question of whether the two paradigms lead to the same set of correlations in the case of infinite dimensional Hilbert spaces is the subject of Tsirelson's problem \cite{scholzwerner, fritztsirelson}. Most studies of Bell-nonlocality are primarily concerned with finite dimensional Hilbert spaces; should one encounter infinite dimensional Hilbert spaces, the commutativity paradigm is the proper way to model spacelike separation.} An extreme choice of joint POVM is the following: 
	
	\begin{equation}
	G^{PR(xy)}_{ab}= \frac{\mathbb{I}_{\mathcal{H}}}{2}\delta_{a\oplus b,xy},
	\end{equation}
	which leads to the probability distribution $p(a,b|x,y)=\frac{1}{2}\delta_{a\oplus b,xy}$ for any choice of quantum state. Hence,	this joint POVM $G^{PR(xy)}$ always yields statistics corresponding to the PR-box, maximally violating the CHSH-type inequality for this scenario, namely,
	
	\begin{equation}
	\sum_{a,b,x,y\big| a\oplus b=xy}\frac{1}{4}p(a,b|x,y)\leq \frac{3}{4}.
	\end{equation}
	
	 Physically, it's possible to implement this (without requiring any quantum resources) by providing a box that always produces these correlations between measurement settings denoted by $(xy)\in\{0,1\}^2$, regardless of the input state. Such a black-box would maximally violate the CHSH-type inequality (viewed as a Bell-KS inequality witnessing KS-contextuality), but that shouldn't be surprising in the absence of spacelike separation. Also, the trivial PR-box joint POVM $G^{PR(xy)}_{ab}$ is a perfectly valid way to implement the joint measurement of trivial POVMs $A^{(x)}$ and $B^{(y)}$ within the standard paradigm of operational quantum theory.\footnote{Note that the point of this demonstration is to show how, in the absence of spacelike separation justifying commutativity or a promise that the measurements are sharp, arbitrary correlations are achievable in quantum theory if unsharp measurements are allowed. All trivial POVMs are unsharp, but the converse is not true. That is, one can consider nontrivial POVMs that don't violate the CHSH-type inequality maximally, but which violate it (arbitrarily) more than is allowed by sharp measurements in quantum theory. One could construct them, for example, by just taking a convex combination of the PR-box trivial POVM with some sharp (and thus product) joint POVM.}
	
	To summarize, we note the following: 
	\begin{itemize}
		\item Within the traditional framework of KS-noncontextuality, if one wants to go beyond projective measurements to arbitrary POVMs in a contextuality scenario, then one must -- in order to avoid the pathology of trivial POVMs violating the Bell-KS inequalities maximally -- {\em restrict} by fiat the notion of joint measurability to merely commutativity. This is, for example, the attitude adopted in Ref.~\cite{sheaf}.
		
		\item However, if one is going beyond projective measurements, we know that commutativity is {\em only} a sufficient condition for joint measurability, not a necessary one \cite{notesonjm}.
		
		\item This brings us to our observation that the traditional notion of KS-noncontextuality is pathological once the most general situation in quantum theory is considered: arbitrary POVMs with the general notion of joint measurability (see, e.g., Ref.~\cite{notesonjm} for this notion and its relation to commutativity). In particular, in the absence of spacelike separation, there is no physical justification to restrict the notion of joint measurability to merely commutativity.
		
		\item A similar consideration applies at the level of a KS-noncontextual ontological model: there, factorizability is not justified in the absence of spacelike separation. So, on those grounds alone, one should go beyond KS-noncontextuality as one's notion of classicality; particularly, if one wants a notion of classicality that does not presume outcome determinism, just as {\em local causality} doesn't presume it. This was argued in Ref.~\cite{finegen}: imagine an adversarial setting where because of the absence of spacelike separation in a KS-contextuality experiment, two measurement settings on the same system can exhibit correlations that are independent of those induced by the system on which the measurements are being implemented, thus allowing them to exhibit stronger correlations than are possible in a KS-noncontextual model. We use trivial POVMs only to drive home that this can be done arbitrarily well (achieving PR-box type correlations, in fact) if there is no constraint on the strength of correlations the measurement settings can exhibit. The way such constraints on the correlations between the measurement settings show up in our analysis within the Spekkens framework is in terms of the quantity ${\rm Corr}$: if ${\rm Corr}$ is really high, the measurements in a noncontextual ontological model cannot be arbitrarily strongly correlated, i.e., $R$ cannot be arbitrarily high (cf.~Eq.~\eqref{NCI3}).
	\end{itemize}
	
	\section{The KS-uncolourable hypergraph $\Gamma_{18}$}\label{scopeviacega}
	It is instructive to consider the KS-uncolourable hypergraph $\Gamma_{18}$, originally appearing in Ref.~\cite{CEGA}, and studied in the light of Spekkens contextuality in Ref.~\cite{KunjSpek}. This hypergraph fails both criteria for the hypergraphs $\Gamma$ considered in this paper, namely, $\mathcal{C}(\Gamma)\neq\varnothing$ (KS-colourability) and $\mathcal{CE}^1(\Gamma)=\mathcal{G}(\Gamma)$.
	
	For probabilistic models on $\Gamma_{18}$, the following hold: $\mathcal{C}(\Gamma_{18})=\varnothing\subsetneq\mathcal{CE}^1(\Gamma_{18})\subsetneq\mathcal{G}(\Gamma_{18})$. This was considered in Ref.~\cite{KunjSpek}, where $\mathcal{CE}^1(\Gamma_{18})$ excludes the extremal probabilistic model in $\mathcal{G}(\Gamma_{18})$ that corresponds to the upper bound on the noise-robust noncontextuality inequality of Ref.~\cite{KunjSpek}. As argued in Ref.~\cite{KunjSpek}, this noise-robust noncontextuality inequality is the appropriate operational generalization (to possibly noisy measurements) of the Kochen-Specker contradiction first demonstrated in Ref.~\cite{CEGA}; this generalization cannot be accommodated in our generalization of the CSW framework \cite{CSW}. 
	
	If one extends the KS-uncolourable $\Gamma_{18}$ to a KS-colourable hypergraph $\Gamma_{27}$ with 9 ``no-detection" events, one for each hyperedge, then we have $\mathcal{C}(\Gamma_{27})\neq \varnothing$, but it's still the case that $\mathcal{C}(\Gamma_{27})\subsetneq\mathcal{CE}^1(\Gamma_{27})\subsetneq\mathcal{G}(\Gamma_{27})$ for this hypergraph.\footnote{This follows from noting that extremal probabilistic models on $\Gamma_{18}$ are still extremal probabilistic models on $\Gamma_{27}$: ones where the no-detection events are assigned zero probabilities. See Theorem 2.5.3 of Ref.~\cite{AFLS}.} Hence, $\Gamma_{27}$ cannot be understood in our generalization of the CSW framework either.\footnote{Note that adding these no-detection events is equivalent to allowing subnormalized probabilities (i.e., sum of probabilities assigned to measurement events in a hyperedge can be less than 1) on $\Gamma_{18}$. Hence, even allowing for subnormalization on $\Gamma_{18}$, which means that one is looking at probabilistic models on the hypergraph $\Gamma_{27}$, does not eliminate the gap between CE$^1$ probabilistic models and general probabilistic models, so that any upper bound on a Bell-KS expression given by probabilistic models in $\mathcal{CE}^1(\Gamma_{27})$ is not always the same as the general probabilistic upper bound from probabilistic models in $\mathcal{G}(\Gamma_{27})$. The CSW framework only considers the upper bound given by $\mathcal{CE}^1(\Gamma_{27})$ probabilistic models.}
	
Indeed, if one ``blindly" writes down a CSW classical bound for some Bell-KS expression defined on $O(\Gamma_{18})$, then such a bound is equivalently a bound for the same Bell-KS expression defined on $\Gamma_{27}$ (where normalization is restored). Further, the E1 bound on $\Gamma_{18}$ is a ${\rm CE}^1$ bound on $\Gamma_{27}$. The GPT bound {\em happens} to agree with the CE$^1$ bound for a particular Bell-KS expression (sum of all probabilities) but {\em differs} for some other Bell-KS expressions defined on this hypergraph. Consider, for example, the following three expressions (see Fig.~\ref{fignodetection}):

\begin{figure}
	\includegraphics[scale=0.3]{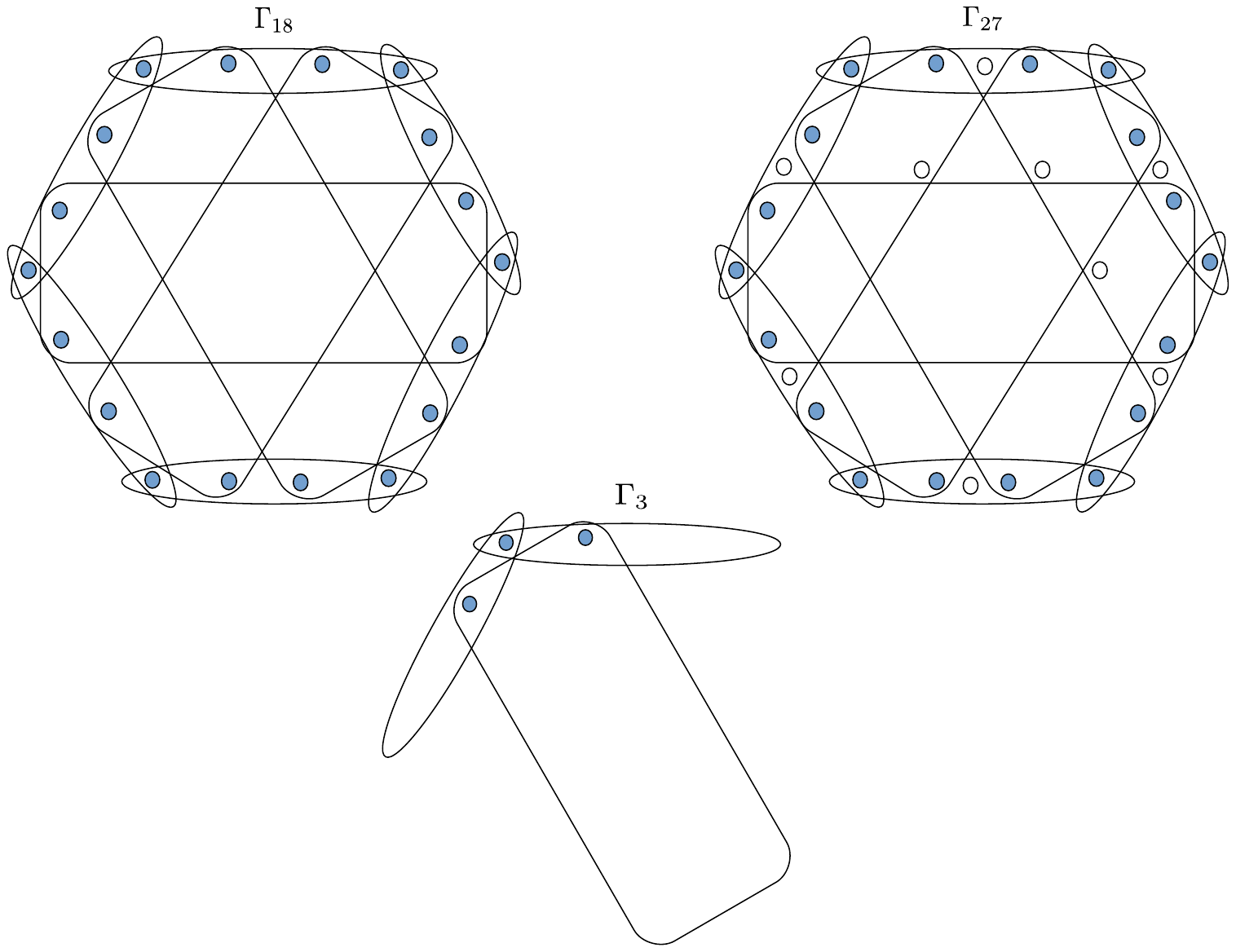}
	\caption{The hypergraph $\Gamma_{27}$ and its subhypergraphs, i.e., $\Gamma_{18}$ and $\Gamma_3$, appearing in the three Bell-KS expressions of Eq.~\eqref{bellksexpr}. The probabilistic model $p$ considered in Eq.~\eqref{bellksexpr} is a probabilistic model on $\Gamma_{27}$, and {\em not} on the subhypergraphs. We have illustrated the subhypergraphs separately only for clarity regarding the subsets of vertices to which the Bell-KS expressions refer: the probabilities assigned to these vertices are obtained from probabilistic models on $\Gamma_{27}$.}
	\label{fignodetection}
\end{figure}

\begin{figure}
	\centering
	\includegraphics[scale=0.3]{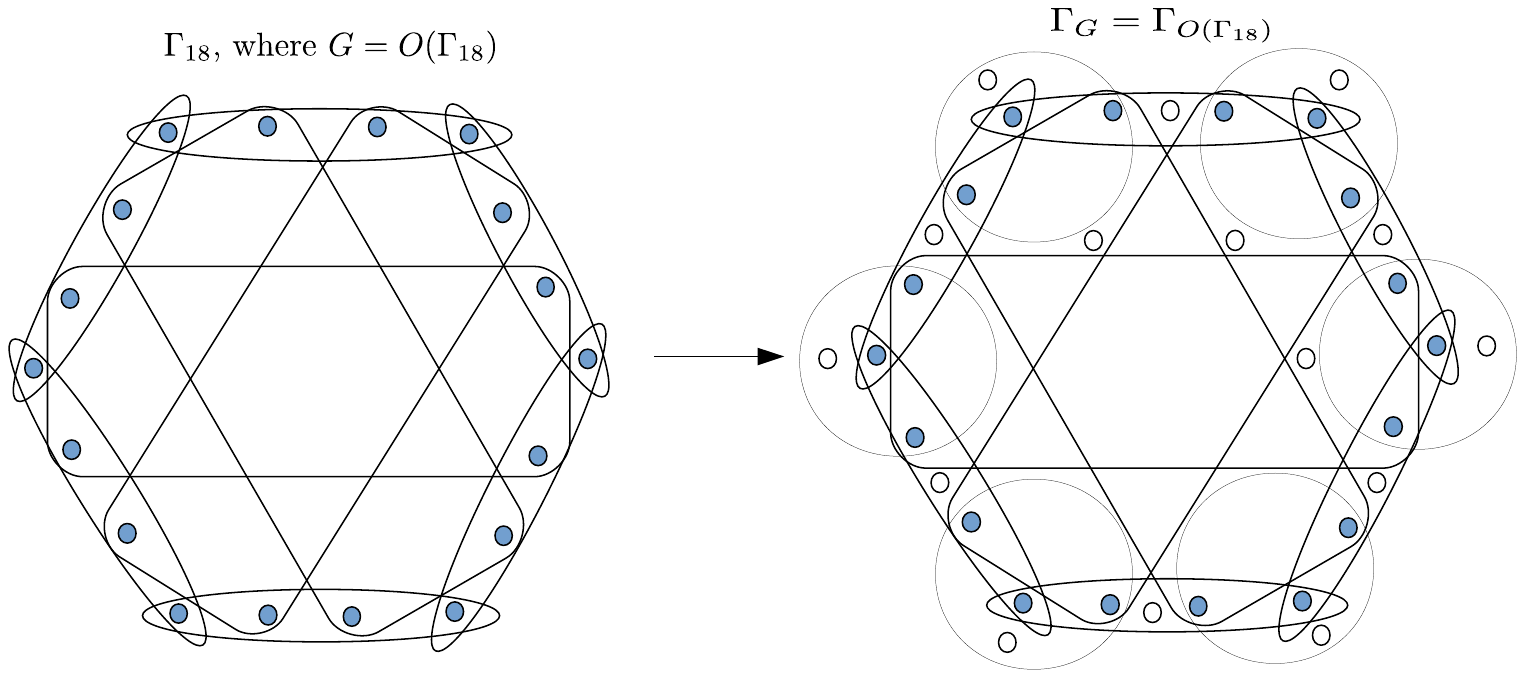}
	\caption{Going from the orthogonality graph, $G$, of $\Gamma_{18}$ to the hypergraph $\Gamma_G$ (on the right) to which our noise-robust noncontextuality inequality pertains.}
	\label{figOgamma}
\end{figure}
	
	\begin{eqnarray}\label{bellksexpr}
	{\rm Expr}_1&\equiv&\sum_{v\in V(\Gamma_{18})}p(v),\nonumber\\
	{\rm Expr}_2&\equiv&\sum_{v\in V(\Gamma_3)}p(v),\nonumber\\
	{\rm Expr}_3&\equiv&\sum_{v\in V(\Gamma_{18})}p(v)+\sum_{v\in V(\Gamma_3)}p(v).
	\end{eqnarray}
	
	We have:
	
	\begin{eqnarray}\label{bellksexprvalue}
	&&{\rm Expr}_1\overset{\rm \mathcal{C}(\Gamma_{27})}{\leq} 8\overset{\mathcal{CE}^1(\Gamma_{27})}{<}9\overset{\mathcal{G}(\Gamma_{27})}{=}9,\nonumber\\
	&&{\rm Expr}_2\overset{\rm \mathcal{C}(\Gamma_{27})}{\leq} 1\overset{\mathcal{CE}^1(\Gamma_{27})}{=}1\overset{\mathcal{G}(\Gamma_{27})}{<}\frac{3}{2},\nonumber\\
	&&{\rm Expr}_3\overset{\rm \mathcal{C}(\Gamma_{27})}{\leq} 9\overset{\mathcal{CE}^1(\Gamma_{27})}{<}10\overset{\mathcal{G}(\Gamma_{27})}{<}10.5.
	\end{eqnarray}
	
	Thus, ${\rm Expr}_3$ is a Bell-KS expression that discriminates between probabilistic models at all three levels of the hierarchy. Indeed, the upper bound on ${\rm Expr}_3$ for $\mathcal{CE}^1(\Gamma_{27})$ models can be saturated by projective quantum realizations of the hypergraph, in particular the standard realization with 18 rays, with the zero operator for the no-detection events \cite{CEGA}. The fact that there exists such a Bell-KS expression as ${\rm Expr}_3$ means that the ${\rm CE}^1$ upper bounds from the CSW approach can be violated by a general probabilistic model, i.e., the upper bounds for ${\rm CE}^1$ models and general probabilistic models don't agree, and we cannot take the graph-theoretic upper bounds of CSW for granted in our noise-robust noncontextuality inequalities. Indeed, the general probabilistic upper bound for any Bell-KS expression defined on a contextuality scenario is a hypergraph invariant --- in the sense that it is a property that is shared by all hypergraphs isomorphic to each other --- that may or may not be expressible as a graph invariant \`a la CSW.
	
    What, then, do the bounds given by graph invariants of CSW for $O(\Gamma_{18})$ {\em mean} in our generalization of the CSW framework? Following our approach, outlined in Sec.~III.B, we can go from $G=O(\Gamma_{18})$ to the hypergraph $\Gamma_G=\Gamma_{O(\Gamma_{18})}$ (see Fig.~\ref{figOgamma})
     for which we have (by construction) $\mathcal{C}(\Gamma_{O(\Gamma_{18})})\neq\varnothing$ (so that the underlying hypergraph is no longer KS-uncolourable) and $\mathcal{CE}^1(\Gamma_{O(\Gamma_{18})})=\mathcal{G}(\Gamma_{O(\Gamma_{18})})$ (so that, for any Bell-KS expression, the upper bound given by the fractional packing number $\alpha^*(G,w)$ in the CSW framework agrees with the general probabilistic upper bound). Since this construction proceeds by converting all maximal cliques in $\Gamma_{18}$ to hyperedges in $\Gamma_{O(\Gamma_{18})}$ and adding a new vertex to each such hyperedge, it achieves both purposes: firstly, adding a (no-detection) vertex to every maximal clique that is a hyperedge in $\Gamma_{18}$ ensures the KS-colourability of $\Gamma_{O(\Gamma_{18})}$, i.e., $\mathcal{C}(\Gamma_{O(\Gamma_{18})})\neq \varnothing$, and secondly, adding a vertex to every maximal clique that is not a hyperedge in $\Gamma_{18}$ ensures that $\mathcal{CE}^1(\Gamma_{O(\Gamma_{18})})=\mathcal{G}(\Gamma_{O(\Gamma_{18})})$. Once these two properties are satisfied, the graph invariants of CSW \cite{CSW} become applicable to any Bell-KS expression defined for any set of vertices in the subhypergraph $\Gamma_{18}$ of $\Gamma_{O(\Gamma_{18})}$.
     
      Our noise-robust noncontextuality inequality then applies to the KS-colourable hypergraph $\Gamma_{O(\Gamma_{18})}$, where the graph invariants of CSW make sense, rather than the KS-uncolourable hypergraph $\Gamma_{18}$. On the other hand, an appropriate noise-robust noncontextuality inequality for the KS-uncolourable hypergraph $\Gamma_{18}$ is, then, the one reported in Ref.~\cite{KunjSpek}.\footnote{The approach for KS-uncolourable hypergraphs will be further developed in hypergraph-theoretic terms in forthcoming work~\cite{kunjunc}.}

\end{appendix}

 \end{document}